\documentclass[11pt]{article}
\usepackage[letterpaper, left=1in, right=1in, top=1in,bottom=1in]{geometry}
\usepackage{amsfonts}     
\usepackage{natbib}
\usepackage{multirow}
\usepackage{lscape}
\usepackage{nicefrac}       % compact symbols for 1/2, etc.
\usepackage{bbm}
\usepackage[ruled,noend]{algorithm2e}
\usepackage{bm}
\usepackage{fancybox}
\usepackage{tikz}
\usetikzlibrary{positioning,chains,fit,shapes,calc,arrows}
\usetikzlibrary{arrows.meta}
\usepackage{graphicx}
\usepackage{pgf}
\usepackage{mathtools}
\usepackage{mathrsfs}
\usepackage{amsmath,amsthm,amssymb}
\usepackage{comment}
\usepackage{xfrac}
\usepackage{accents}
\usepackage{cancel}
\usepackage{enumitem}

\usepackage{auxiliary}

\usepackage{float}

\usepackage{url}      

\usepackage[pagebackref=true]{hyperref} \PassOptionsToPackage{pagebackref=true}{hyperref}   

\usepackage{cleveref}

\usepackage{color-edits}            % Need the color package
\title{Scheduling in Queueing Systems with \\ Uncertain and Evolving Holding Costs}

\author{Caner Gocmen\thanks{Meta Platforms, \texttt{caner@meta.com}} \and Thodoris Lykouris\thanks{Massachusetts Institute of Technology, \texttt{lykouris@mit.edu}} \and Deeksha Sinha\thanks{Meta Platforms, \texttt{deekshasinha@meta.com}}
\and Wentao Weng\thanks{University of Michigan, \texttt{wwentao@umich.edu}}}

\date{First version: May 2025\\Current version: November 2026}

\begin{document}

\maketitle

\begin{abstract}
    % !TEX root = main.tex
In queueing systems, the holding costs of jobs may evolve in uncertain ways. For example, in
content moderation for social media platforms, the cost of delaying the review of a content piece is proportional to its view trajectory, which fluctuates and is a priori unknown. To capture such uncertain and evolving holding cost patterns, we consider a queueing model where job states evolve based on a Markov chain with state-dependent instantaneous holding costs. We demonstrate that in the presence of such uncertain and evolving holding costs, the two canonical algorithmic principles, instantaneous-cost ($c\mu$-rule) and expected-remaining-cost ($c\mu/\theta$-rule), are suboptimal. By viewing each job as a Markovian ski-rental problem, we develop a new index-based algorithm, \textsc{Opportunity-adjusted Remaining Cost} ($\alg$), that adjusts to the opportunity of serving jobs in the future when uncertainty partly resolves. We show that the suboptimality gap of $\alg$ scales as $\tilde{O}(\sqrt{N})$, where $N$ is the system size. This bound shows that $\alg$ achieves asymptotic optimality for overloaded systems when the system size $N$ scales to infinity. Moreover, the bound is independent of the state-space size, which is a desirable property when job states contain contextual information. We corroborate our results with an extensive simulation study based on two holding cost patterns (online ads and user-generated content) that arise in content moderation for social media platforms. Our simulations based on synthetic and real datasets demonstrate that $\alg$ consistently outperforms existing practice, which is based on the two canonical algorithmic principles.
\end{abstract}

\section{Introduction}\label{sec:intro}
% !TEX root = main.tex
Content moderation is important for any platform \citep{gillespie2018custodians} and is especially technically challenging for large social media platforms \citep{halevy2022preserving} due to the scale of content from billions of users. To address this technical challenge, social media platforms like Meta and TikTok rely on an AI-human pipeline and employ tens of thousands of human reviewers \citep{meta-humanreview,tiktok}. In particular, Artificial Intelligence (AI) models initially filter content that clearly violates the platform's policies (henceforth referred to as \emph{policy-violating content}) \citep{facebook-standard, tiktok-standard} and route ambiguous content to human reviewers (see Section~\ref{sec:background} for further discussion of this pipeline). Human reviewers then scrutinize the latter content to identify further policy violations. The central objective of this human review system is to reduce the prevalence of policy-violating content (henceforth referred to as \emph{prevalence}), which is typically proportional to its number of views \citep{halevy2022preserving}. As such, a content piece's \emph{holding cost}, which is proportional to its accumulated views while waiting for review, differs across content pieces. Given the large volume of content and the heterogeneity in holding costs, the scheduling algorithm, which determines the review order, plays a vital role in the prevalence reduction efforts for content moderation.

There are two canonical algorithmic principles on how to schedule jobs with heterogeneous holding costs, which have been extensively studied with applications from manufacturing systems \citep{pinedo2012scheduling} to service systems such as call centers \citep{gans2003telephone} and hospitals \citep{armony2015patient}. First,  the  (generalized) $c\mu$-rule greedily serves jobs with the maximum (service-rate-weighted) \emph{instantaneous holding cost}. This algorithm is optimal for linear holding costs \citep{CoxSmith61} as well as heavy-traffic optimal for convex holding costs \citep{van1995dynamic,mandelbaum2004scheduling}. Second, the $c\mu/\theta$-rule prioritizes jobs with highest (service-rate-weighted) \emph{expected remaining holding cost}. This algorithm is asymptotically optimal when jobs have linear holding costs and abandon the system at an exponential rate \citep{atar2010cmu,atar2011asymptotic}. The optimality guarantees for both algorithmic principles rely on the assumption that the holding cost of a job follows a deterministic and known function of the wait time until the job gets served or (independently) abandoned.

This known-holding-cost setting is not applicable to content moderation. In content moderation, the holding cost of a job (content waiting for human review) is proportional to the number of views it accumulates \citep{halevy2022preserving}. The number of views a content piece may obtain over time is inherently uncertain; that said, the uncertainty gradually resolves over time \citep{cha2009analyzing}. Such evolving uncertainty manifests in social media platforms due to the large sources of inherent randomness contributing to a content piece's view trajectory. For example, a content piece's view trajectory depends on whether it is liked or reshared by friends \citep{rizoiu2017expecting,haimovich2021popularity}, whether it is recommended by the platform \citep{haimovich2021popularity}, or even whether it turns out to provoke positive emotion \citep{berger2012makes}. Despite the available tools for content popularity prediction (see \cite{zhao2015seismic,rizoiu2017expecting, haimovich2021popularity} and the references therein), perfectly predicting a content piece's view trajectory is unrealistic \citep{salganik2006experimental,cheng2014can}, which separates content moderation from the typical known-holding-cost setting. We note that uncertain and evolving holding costs are not unique to content moderation but  also arise in other applications. For example, in healthcare, patients’ health conditions can deteriorate or improve over time \citep{hu2022optimal,akan2012broader}; in call centers, low-priority customers may randomly upgrade their priority and then have higher holding cost rates \citep{down2010n}. As a result, although this paper is primarily motivated by content moderation, our insights extend beyond this application.

\begin{figure}
\centering
\includegraphics[width=2.1in]{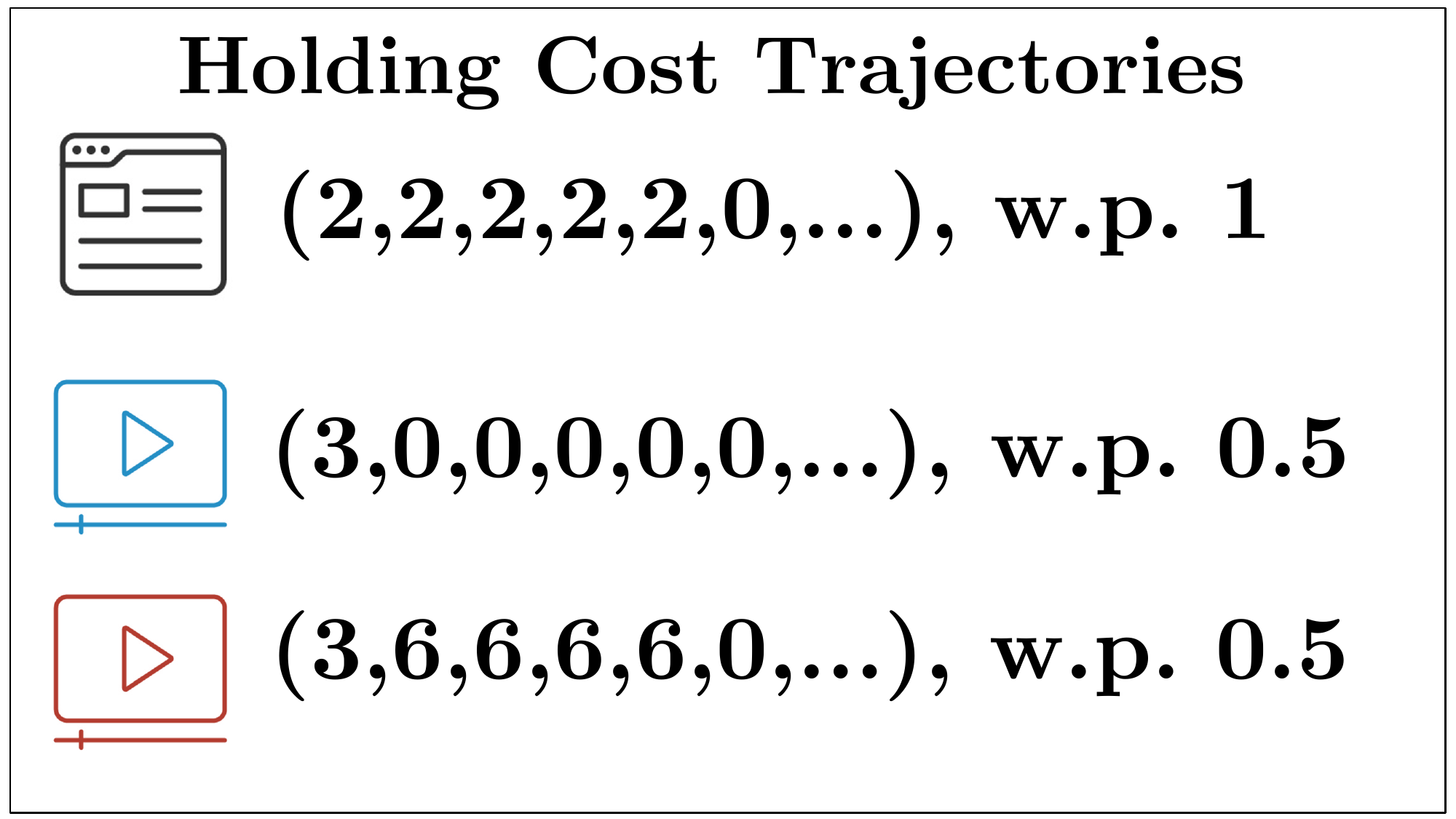}
\includegraphics[width=2.1in]{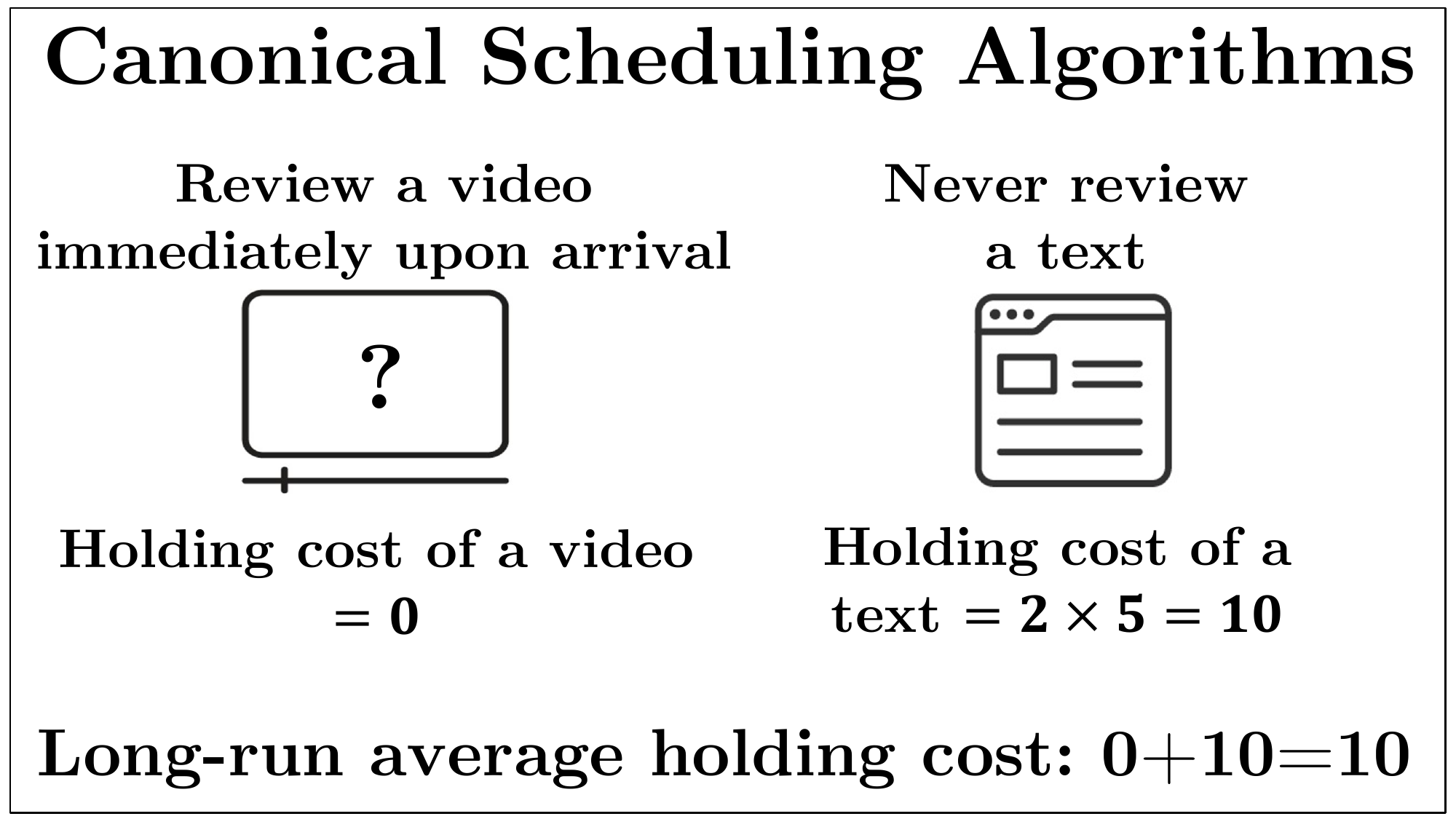}
\includegraphics[width=2.1in]{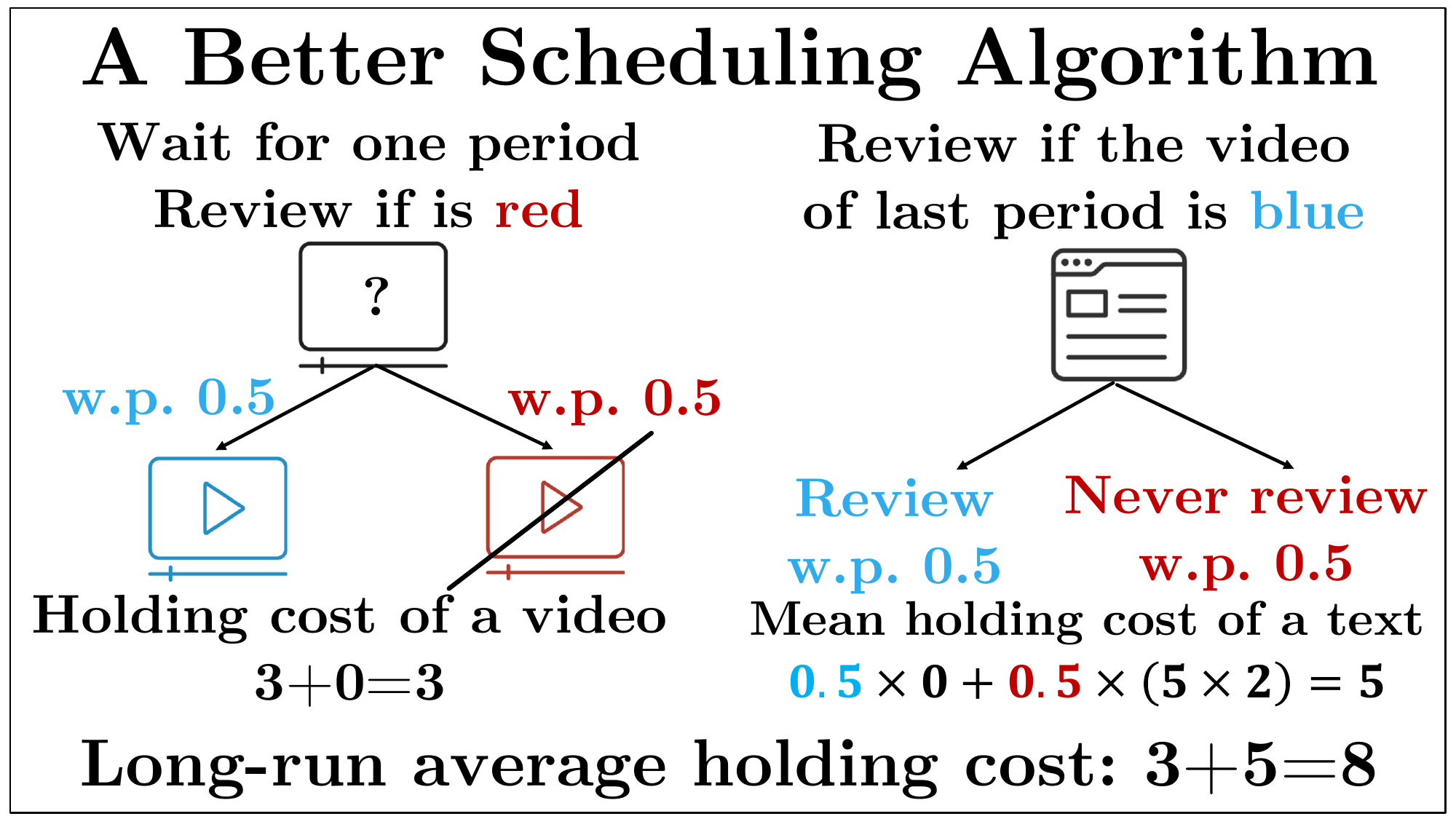}
\caption{In each period, both a \textsc{Text} and a \textsc{Video} (with unknown color) arrive with the given  holding cost trajectories. With review capacity of only one post per period, existing algorithmic principles (instantaneous or expected remaining cost) always serve new \textsc{Video} jobs. A better algorithm Waits for one period to know if a \textsc{Video} is of low (blue) or high (red) virality. This enables it to save capacity by not reviewing a blue \textsc{Video}.}
\label{fig:failure}
\end{figure}

When holding costs are uncertain and evolving, the two canonical algorithmic principles discussed above fail to deliver near-optimal performance; Figure~\ref{fig:failure} provides an example illustrating this suboptimality. In this example, the system operates in discrete-time periods and there are three possible job holding cost trajectories, which we pictorially refer to as \textsc{Text}, \textsc{BlueVideo}, and \textsc{RedVideo}. All jobs incur zero cost after five periods. A \textsc{Text} job incurs a cost of $2$ for the initial five periods, a \textsc{BlueVideo} job incurs a cost of $3$ for the first period but has $0$ cost afterwards, and a \textsc{RedVideo} job incurs a cost of $3$ for the first period and a constant cost of $6$ afterwards. The algorithm observes the current holding costs (but not the future), serves one job, and pays the holding cost for unserved jobs. In each period, two jobs arrive: one is \textsc{Text} (the algorithm knows since only \textsc{Text} jobs have holding cost $2$) and the other one is \textsc{Video}, which has equal probability to be a \textsc{BlueVideo} or a \textsc{RedVideo} job (the algorithm knows it is either of the two since the holding cost is $3$, but not the actual realization since it cannot observe future holding costs). Both the canonical algorithmic principles will prioritize a \textsc{Video} job over any \textsc{Text} job as it has a larger instantaneous cost ($3 > 2$) and a larger expected remaining cost ($0.5\times 3 + 0.5 \times (3+6\times 4) > 5 \times 2$). The algorithm will serve a \textsc{Video} job for every period and all \textsc{Text} jobs will backlog in the queue. However, serving a \textsc{BlueVideo} job is suboptimal as it will incur no more cost after its first period. Instead, a better algorithm can account for the future opportunity to serve a \textsc{Video} job once its uncertainty is resolved. If the algorithm prioritizes  a new \textsc{Text} job over a new \textsc{Video} job, it can avoid servicing a \textsc{BlueVideo} job while still being able to serve a \textsc{RedVideo} job in the next period. This illustrates that with uncertain holding costs, canonical algorithmic principles are suboptimal: the instantaneous-cost principle completely ignores the future holding cost a job can incur, while the expected-remaining-cost principle disregards the opportunity of serving a job in the future when the job's uncertainty partially resolves. 

Thus motivated, this paper tackles the following research question:   
\begin{center}
\emph{How should we redesign the scheduling algorithm when holding costs are uncertain and evolving}?
\end{center}

\subsection{Our methodological contributions}

\paragraph{Model.} To capture uncertain and evolving holding costs, we consider a discrete-time queueing system where the instantaneous holding cost of a waiting job is determined by the job's state and the state evolves according to a known Markov chain. In particular, the system is specified by the following model primitives: a system size  $N$, an arrival rate $\lambda$, a service rate $\mu$, a cost vector $\bolds{c}$, and a transition kernel between the states. In each period $t$, a random number of servers $R(t) \sim \mathrm{Bin}(N,\mu)$ are available to serve jobs. The scheduling algorithm selects at most $R(t)$ jobs from the waiting queue $\set{Q}(t)$ and those jobs leave the system without incurring further cost. All remaining jobs $j$ with state~$S_j(t)$ incur an instantaneous holding cost $c(S_j(t))$. Moreover, each job $j$ transitions to a new state $S_j(t+1)$ in the next period according to the Markov chain, which we assume is a directed tree. A job can also transition into an empty state $\perp$, meaning that the job abandons the system without getting service. At the end of a period, a random number of new jobs, $A(t) \sim \mathrm{Bin}(N,\lambda)$,  join the system. A feasible scheduling algorithm selects jobs for service with known model primitives and current states of all waiting jobs but no knowledge of jobs' future states and thus cost trajectories. The goal is to design a feasible scheduling algorithm with minimum holding cost. A scheduling algorithm is \emph{asymptotically optimal} if, as the system size $N$ increases, the gap between its long-run average holding cost and that of an optimal feasible scheduling algorithm is sublinear in $N$. 

\paragraph{Algorithm.} The key tension, as demonstrated in Figure~\ref{fig:failure}, is whether to serve a job \emph{now} (to prevent its future uncertain cost using current capacity) or \emph{postpone} its service (to save current capacity but endure an instantaneous cost). Relaxing the capacity constraint with a capacity  price~$\gamma$ that we later specify, this tension is captured by a ski-rental problem. In ski-rental, a skier chooses, over a time horizon, whether to pay a one-time cost of buying equipment (similar to how serving a job pays the capacity  price $\gamma$) or to pay a daily rental cost (similar to how postponing the service of a job yields an instantaneous holding cost). Although a typical ski-rental problem assumes a uniform daily rental cost that becomes zero after an ex-ante unknown time \citep{karlin2001dynamic}, the daily rental cost in our setting is uncertain as the instantaneous cost of a job changes over time. 

Thanks to the Markovian structure of jobs, this problem, which we refer to as \emph{Markovian ski-rental}, can be formulated as a dynamic program  for any given capacity price $\gamma$. Specifically, the cost-to-go function $V(\gamma,i)$ for a job with current state $i$ satisfies $V(\gamma, i) = \min\left\{\gamma, c(i) + V^f(\gamma, i)\right\}$ with $V^f(\gamma,i)$ being the expected (future) cost-to-go function for the next state conditioning on the current state. This recursion shows that the optimal cost for a job with current state $i$ is the minimum between (i) the cost $\gamma$ of serving it now, or (ii) the sum of the instantaneous cost $c(i)$ and the future minimum cost $V^f(\gamma, i)$. Based on this formulation, our scheduling algorithm, named \textsc{Opportunity-adjusted Remaining Cost} (or $\alg$ in short), computes
a suitable capacity price~$\gamma^\star$ from the dual of a linear program (LP), which is a fluid relaxation of the original problem. It then creates an index for a job with state $i$ by $\ind_{\alg}(i) = c(i) + V^f(\gamma^\star, i)$ and serves the $R(t)$ jobs with highest indices.  Theoretically, this algorithm enjoys a suboptimality gap bound independent of the state-space size and achieves asymptotic optimality (Theorem~\ref{thm:pafou}). Practically, it motivates an easy-to-implement heuristic that shows strong numerical performance (see Sections~\ref{sec:background}~and~\ref{sec:numerics}).

\paragraph{Analysis.} Our analysis contains four components. First, using the tree structure, we derive the fluid equilibrium for any priority algorithm, which operates based on a fixed ordering of states and serves jobs with states in this ordering (Section~\ref{sec:fluid-sol}). The class of priority algorithms includes any index-based scheduling algorithm, such as $\alg$, with a suitable tie-breaking rule. Second, using the complementary slackness theorem, we prove that the fluid equilibrium for $\alg$ is optimal for the fluid LP (Lemma~\ref{lem:optimal-fluid}). Third, we show that the steady-state distribution of queue lengths under any priority algorithm behaves similarly to its fluid equilibrium up to a $\tilde{O}(\sqrt{N})$ term (Lemma~\ref{lem:stochastic}). Lastly, we establish that the fluid LP is a lower bound for the long-run average holding cost of any feasible algorithm (Lemma \ref{lem:lower-bound}). These combined show that the suboptimality gap of $\alg$ is at most $\tilde{O}(\sqrt{N})$, thus establishing its asymptotic optimality.

Our main contribution lies in the third component (Section~\ref{sec:stochastic}) which shows that the steady-state distribution for a priority algorithm is close to its fluid equilibrium. In particular, for any priority algorithm, there is a minimum set of states, which we call the \emph{top set}, such that fully serving jobs with these states would clear all the jobs the  algorithm should serve in its fluid equilibrium. We establish that the algorithm (with high probability) serves all the jobs in the subtree of this top set by analyzing a corresponding Lyapunov function. The Lyapunov analysis contains two steps. The first step shows that this Lyapunov function has a negative drift if its value is large enough (Lemma~\ref{lem:drift-bound}). A main challenge is to upper bound the numbers of jobs with states in the top set, which are \emph{correlated} random variables. We tackle this challenge by analyzing a system with no job-level correlation and establish a large deviation result (Lemma~\ref{lem:queue-upper-bound}). The second step translates the negative drift into a high probability bound of the Lyapunov function. A difficulty is that our Lyapunov function can, albeit with small probability, change significantly within one period. Existing techniques rely on an almost sure upper bound on the change and thus fail to give satisfactory guarantees. To address this issue, we develop a new geometric tail bound (Lemma~\ref{lem:tail-bound}) that only requires the Lyapunov function to have a small change with high probability. More broadly, our techniques contribute to the literature on restless bandits by providing performance guarantees of priority algorithms without assuming the canonical Global Attractor Property (GAP); see related work in Section~\ref{sec:related-work} and further comparison in Appendix~\ref{app:comp}.

\subsection{Application to content moderation for social media platforms}\label{sec:background}

Our work provides a better scheduling algorithm for content moderation in social media platforms. The volume of content in social media platforms is too large to be all reviewed by humans \citep{avadhanula2022}. As a result, current content moderation systems combine AI and human reviewers, where an initial AI filter determines whether a content piece violates any of the platform's policies and admits it for further human review if the AI is uncertain \citep{lykouris2024learning}. The content that is deemed policy-violating either by the AI system or by human reviewers is then removed from the platform. The scheduling algorithm for the human review system determines the order in which humans review the content 
that was admitted by the AI system, i.e., the content enqueued for human review. Let $p\violating(j)$ be the probability that a content piece $j$ is policy-violating and $\cumview(j)$ be its cumulative number of views until a human reviews it (or the end of a time horizon if a human never reviews it). The system seeks to minimize the \emph{predicted policy-violating views}, defined by the sum of $p\violating(j)\times \cumview(j)$ among enqueued content~$j$.\footnote{We use the predicted instead of the actual policy-violating views, which is based on the actual human labels of content, because the latter is unobservable for unreviewed content. See Section~\ref{sec:instantiate} for a more detailed discussion.} There are three scheduling algorithms used in practice: 
\begin{itemize}
\item \textsc{pViolating} \citep{Linkedin} prioritizes the content with the largest policy-violating probability, which is typically available from (a group of) machine learning (ML) models \citep{avadhanula2022}.
\item \textsc{Velocity} \citep{avadhanula2022} prioritizes the content with the largest number of views in the last period weighted by the policy-violating probability. This mimics the first canonical algorithmic principle discussed above and optimizes instantaneous cost similar to the $c\mu$-rule.
\item $\textsc{pIV}$ \citep{makhijani2021} prioritizes content with the highest \emph{predicted} number of remaining views weighted by the policy-violating probability. This mimics the second canonical algorithmic principle discussed above and optimizes the expected remaining cost similar to the $c\mu/\theta$ rule.
\end{itemize}
These baseline algorithms do not carefully incorporate the uncertainty in the content's view trajectory. The first two do not account for future views at all, while $\textsc{pIV}$ only considers the expected number of future views and, as a result, may overly tailor to content that may become extremely viral but with low probability of being policy-violating. In contrast, our algorithm adapts to this uncertainty, yielding significant practical improvements as we detail below. 

To avoid the overhead of training a reinforcement learning model for the cost-to-go function~$V^f$, we propose a simplified implementation of $\alg$ based on hindsight approximation \citep{sinclair2023hindsight}, which we call $\hindalg$. Specifically, we approximate the cost-to-go function by the product of (i) the policy-violating probability and (ii) the predicted remaining views capped by a hyper-parameter~$\gamma$. We prioritize content with the highest $\hindalg$ index, (informally) defined by
\[
\ind_{\hindalg} = p\violating \cdot \left(\text{Views in the last period} + \text{Prediction of }\min\{\gamma, \text{Remaining views}\}\right).
\]
This algorithm is easy to deploy in practice as it is index-based like the existing algorithms. Although it requires building a new ML model to predict $\min\{\gamma, \text{Remaining views}\}$, there are ML models predicting a content piece's remaining views \citep{zhao2015seismic,rizoiu2017expecting,haimovich2021popularity}. One can reuse the training pipeline of these models to create the new one by capping the target value with $\gamma$.

Our simulations based on a semi-synthetic dataset (Section~\ref{sec:numerics}), which contains actual view trajectories of YouTube videos and synthetic violation probabilities, demonstrate substantial improvement of $\hindalg$ over existing scheduling approaches. For a wide range of reviewing capacity, $\hindalg$ consistently delivers at least $3.6\%$ reduction in policy-violating views compared to existing practices. Another way to measure the  benefit of $\hindalg$ is through its reviewer-hour savings, i.e., the reduction in the number of human reviews needed to prevent the same amount of policy-violating views as a baseline scheduling algorithm. Given that social media platforms employ tens of thousands of reviewers \citep{meta-humanreview,tiktok}, even one percent reviewer-hour savings can unleash thousands of reviewer hours per day with no deterioration in moderation quality. Our simulations demonstrate that $\hindalg$ offers $11\%$ to $20\%$ reviewer-hour savings; if realized in practice, this can translate to millions of reviewer hours for one year. We also conduct robustness checks to show that these benefits persist even when (i) the simulations are for alternative synthetic datasets tailored to ads and user-generated content or (ii) the probability of violation  is uncalibrated or updated dynamically (either by Bayesian learning or bandit learning).

\subsection{Related work}\label{sec:related-work}
% !TEX root = intro.tex
\noindent\textbf{Bayesian learning and scheduling.} Our work lies in the recent literature of using online information (learning) to enable better scheduling decisions \citep{walton2021learning}. Specifically, our model resembles a \emph{Bayesian} setting where dynamic job classes capture uncertainties in job information that gradually resolve over time. Leveraging information in such uncertainties can help develop better scheduling algorithms. For example, \cite{bassamboo2016scheduling, bassamboo2023optimally} show how to use information on the job abandonment distribution to achieve optimal queue length as some jobs leave the system sooner than others. To capture uncertainty in patients' future health conditions, \cite{hu2022optimal} study a queueing model where patients randomly transition between a moderate and an urgent state. They propose an optimal index-based scheduling algorithm that incorporates state transition information. Jobs with dynamic classes are also used to model liver allocation systems \citep{akan2012broader} as well as service systems with customer upgrades \citep{down2010n}. Our contribution to this literature is using a general Markov chain to model job state transitions and proposing an asymptotically optimal algorithm with a suboptimality gap bound independent of the number of possible job states. Although jobs in our model and these studies change states only while they are waiting, in another literature they change states only upon service but can either rejoin the system or leave \citep{alizamir2013diagnostic, massoulie2016capacity, bimpikis2019learning, shah2020adaptive}. Such models capture expert systems where inspection of a job by the server reveals new information. Lastly, there is research \citep{singh2025feature,zhang2025admission} focusing on \emph{offline} learning settings where a job has an unknown job type and the decision maker trains a machine learning model based on historical data to find the optimal mapping between a job's observed covariate and its latent type. Our work differs from this literature by allowing the decision maker to update her belief of a job's type as new information comes.

\paragraph{Restless bandits and scheduling.} Our model can be viewed as restless bandits with dynamically arriving and departing arms. In a typical restless bandit setting, a decision maker activates a subset of arms and the arms transition according to a known Markov chain depending on whether the arm is activated or not \citep{whittle1988restless}. Our work is close to the development of index-based algorithms in restless bandits; see \cite{nino2023markovian} for a survey. \cite{whittle1988restless} derives the Whittle index algorithm which is asymptotically optimal under an indexability condition \citep{whittle1988restless,weber1990index}. \cite{larranaga2015asymptotically,ayesta2017scheduling,aalto2024whittle} show the optimality of Whittle index for certain multiclass queueing systems with abandonment. However, Whittle index can be difficult to compute \citep{akbarzadeh2022conditions}  and the indexability condition can be hard to verify \citep{bertsimas2000restless}. Later work develops indices based on Lagrangian duals \citep{bertsimas2000restless, brown2020index, verloop2016asymptotically, avrachenkov2024lagrangian}; however, the algorithms require the hard-to-verify GAP condition to be asymptotically optimal \citep{weber1990index, hong2023restless}. Recent work has devised non-index-based algorithms with asymptotic optimality without the GAP assumption \citep{hong2023restless, hong2024unichain, hong2024achieving,gast2024linear}. Concurrent to our work, \cite{li2025improving} show how scheduling in a multi-class queueing system with convex holding costs translates into a restless bandit problem and design a heuristic based on Whittle index. They demonstrate that their algorithm has numerically improved performance, though no theoretical guarantee is given. Our work contributes to this literature by showing that an index-based algorithm can still be asymptotically optimal without the GAP assumption if, when activated, an arm stops incurring any cost and, when not activated, evolves according to a Markov chain that reaches an absorbing state with non-zero probability; see Appendix~\ref{app:comp} for further discussion.

\paragraph{Techniques on showing steady-state optimality.} Existing work in the analysis of index-based scheduling algorithms mostly focuses on analyzing a fluid version of the stochastic system \citep{hu2022optimal,long2020dynamic} or establishing ``fluid optimality'' \citep{atar2010cmu}, i.e., fixing a time horizon $T$ the algorithm becomes optimal as the system size $N$ goes to $\infty$; see \cite{puha2022fluid} for a general treatment for fluid analysis of scheduling algorithms in a multi-class multi-server system. Our work differs from these studies in two aspects: first, our system is discrete-time; second and more importantly, we bound the steady state ($T \to \infty$) gap for the stochastic system for any system size $N$. This requires new analytical tools as discussed in \cite{atar2011asymptotic}. We achieve this via the drift method. Building on tools from \cite{hajek1982hitting, bertsimas2001performance}, the drift method was established in \cite{stolyar2004maxweight, eryilmaz2012asymptotically, maguluri2016heavy} on showing heavy-traffic optimality of the \textsc{MaxWeight} algorithm where the system size is fixed but the load scales to a critical point. Our work is closer to the literature of using drift method to analyze load balancing algorithms in a continuous-time many-server regime; see e.g. \cite{ying2017stein, liu2020steady,weng2020optimal,varma2023power}. Our contribution is extending techniques from this literature to a discrete-time setting where the arrival intensity per period scales to infinity. This requires the development of a new geometric tail bound (see the discussion below Lemma~\ref{lem:tail-bound}).

\noindent\textbf{Analytics for content moderation.} Our work contributes to the recent effort of using analytics to reduce AI error and improve human review efficiency in content moderation. Focusing on the interface between AI and humans, \cite{avadhanula2022} apply online learning to decide, based on risk scores from multiple AI models, which content should be prioritized for human reviews. \cite{lykouris2024learning} abstract the content moderation pipeline into classification, admission, and scheduling decisions and propose a near-optimal algorithm that balances classification error and delay in human reviews. Focusing on the human review system, \cite{makhijani2021} propose a simulation framework that allows the platform to evaluate the impact of different routing and capacity decisions. Given that social media platforms rely on global labor for content moderation \citep{roberts2019behind}, \cite{allouah2023fair} study how to fairly allocate content moderation workload to different review teams. The work closest to ours is \cite{lee2024design}, who extend the $c\mu-$rule to develop an index-based scheduling algorithm accounting for the impact of prediction errors on content holding costs. Their algorithm is heavy-traffic optimal assuming known and convex holding costs. As discussed above, the holding cost of a job in content moderation does not satisfy this assumption as it is proportional to its uncertain (and non-convex) view trajectory. Our work addresses this complexity by designing a near-optimal index-based scheduling algorithm for uncertain and non-convex holding costs.

\section{Model}\label{sec:model}
% !TEX root = main.tex
\paragraph{Uncertain and evolving holding cost trajectory.} We model the uncertain nature of jobs' holding cost trajectory by a \emph{known} Markov chain. The Markov chain has a finite state space $\set{S}$, a length $L$, a cost vector $\bolds{c} = (c(i))_{i \in \set{S}}$ with $c(i) > 0$ for any $i \in \set{S}$, and a transition kernel $P=(P(i,k))_{i,k \in \set{S}}$ denoting the probability of transitioning from a state $i$ to another state $k$. We assume the Markov process has a tree structure,\footnote{A tree Markov chain can model any stochastic process by defining states as the observed history. Although the number of states can be exponential in the state space of the original stochastic process, our guarantees will not scale with the number of states.} such that the state space partitions into $L$ levels of $\set{S}_0, \ldots, \set{S}_{L - 1}$. Each state $i \in \set{S}_{\ell}$ with $\ell > 0$ has a unique \emph{parent} $\pa(i) \in \set{S}_{\ell - 1}$ such that its transition probability $p(i) \equiv P(\pa(i), i) > 0$ and $P(k, i) = 0$ for any $k \neq \pa(i).$ Let $\child(i)$ denote the set of \emph{children} states for state $i$, i.e., $\child(i) = \{k \in \set{S} : \pa(k) = i\}$. We allow $\sum_{k \in \child(i)} P(i,k) < 1$; in that case, state $i$ has probability $1 - \sum_{k \in \child(i)} P(i, k)$ to transition into an empty state $\perp$ after which no future cost will incur, i.e., the job abandons the system. Letting $\theta = \min_{i \in \set{S}} 1 - \sum_{k \in \child(i)} P(i,k)$ be the minimum abandonment probability across states, we assume throughout this paper that $\theta > 0$. Since the Markov chain is a tree, there is a unique state in $\set{S}_0$, which we call the root and denote by $\rootNode \in \set{S}_0$. For the root node we set $p(\rootNode) = 1$.

\paragraph{A discrete-time queueing system.} 
Given a system size $N$, normalized arrival rate $\lambda \in (0,1)$ and service rate $\mu \in (0,1)$, the queueing system operates in periods $ \{1,2,\ldots\}.$\footnote{We assume $\lambda,\mu < 1$ so that in each period the system has nonzero probability to see no arrival and no service.} Initially there is no waiting job in the system. For a period $t \geq 1$, the following events happen.  
\begin{enumerate}
\item \emph{Server capacity:} A random number of servers $R(t)$ are available with  $R(t) \sim \mathrm{Bin}(N, \mu)$.
\item \emph{Service decision:} Observing the states of waiting jobs, the system selects a set $\set{R}(t)$ of jobs from the queue $\set{Q}(t)$ of jobs waiting for service, which satisfies $\set{R}(t) \subseteq \set{Q}(t)$ and $|\set{R}(t)| \leq R(t).$ Selected jobs leave the system without incurring any cost.
\item \emph{Holding cost:} For all remaining jobs \emph{waiting for service}, they incur a holding cost. Denote by $S_j(t) \in \set{S}$ the random but observable state of a job $j$ in period $t$. The total cost incurred in this period is $\sum_{j \in \set{Q}(t) \setminus \set{R}(t)} c\left(S_j(t)\right).$  
\item \emph{Transition:} Any remaining job $j \in \set{Q}(t) \setminus \set{R}(t)$ transitions (evolves) into a new state $S_j(t+1)$ following the transition kernel $P(S_j(t),S_j(t+1))$. A job transitioning into the empty state~$\perp$ implies that it \emph{abandons}, i.e., it leaves the system without getting service. 
\item \emph{Arrivals:} A set of new jobs $\set{A}(t)$ arrive, with cardinality  of $\set{A}(t)$ being $A(t)$. Further, $A(t)$ follows a binomial distribution $\mathrm{Bin}(N,\lambda)$ with $\lambda > 0$ being the normalized arrival rate. All new jobs have the same initial state,  i.e.,  the root node of the Markov chain. Thus, $S_j(t+1) = \rootNode$ for any new job $j \in \set{A}(t)$. 
\item \emph{Queue update:} The queue of waiting jobs for the next period is then $\set{Q}(t+1) = \{j \in \set{Q}(t) \setminus \set{R}(t) \colon S_j(t+1) \neq \perp\} \cup \set{A}(t)$. 
\end{enumerate}

A scheduling algorithm $\nalg$ selects the set $\set{R}(t)$ of jobs for servers to work on in each period $t$. We say that an algorithm $\nalg$ is feasible, if the selection is only based on the historical observed information $\{R(\tau),\set{Q}(\tau), \{S_j(\tau)\}_{j \in \set{Q}(\tau)}\}_{\tau \leq t} \cup \{\set{R}(\tau),\set{A}(\tau)\}_{\tau < t}$,  model primitives $N,\lambda, \mu$, and the Markov chain of jobs, including the cost of each state and the transition kernel. However, the algorithm cannot use future information including the future available servers, the future number of arrivals, and the future states of jobs. 

\paragraph{Objectives.}
Let $\set{Q}(t,\nalg)$ be the set of jobs in queue and $\set{R}(t,\nalg)$ be the selected set of jobs for service under an algorithm $\nalg$ in period $t$. We want to select a feasible algorithm that minimizes the long-run average holding cost defined by
\begin{equation}\label{eq:def-cost}
C(\nalg) = \lim \sup_{T \to \infty} \frac{1}{T} \expect{\sum_{t=1}^T \sum_{j \in \set{Q}(t,\nalg) \setminus \set{R}(t,\nalg)} c(S_j(t))},
\end{equation}
where the expectation is taken over the randomness in (i) server capacity $R(t)$, (ii) algorithm's service decision, (iii) job transitions, and (iv) arrivals $A(t)$. Letting the set of feasible algorithms be $\Pi$, minimizing the long-run average holding cost is equivalent to minimizing the \emph{suboptimality gap}, which is 
the increase in cost compared to an optimal feasible algorithm:
\begin{equation}\label{eq:def-regret}
\subopt(\nalg) = C(\nalg) - \inf_{\nalg' \in \Pi} C(\nalg').
\end{equation}
Our goal is to design a scheduling algorithm with small suboptimality gap. Specifically, varying the system size $N$ gives a sequence of queueing systems. For each system, an algorithm $\nalg$ has a corresponding long-run average holding cost $C(N, \nalg)$ and suboptimality gap $\subopt(N, \nalg)$. An algorithm $\nalg$ has vanishing suboptimality gap if $\subopt(N,\nalg) / N \to 0$ as $N \to \infty$.

Vanishing suboptimality gap implies \emph{asymptotic optimality} for overloaded systems ($\lambda > \mu$), which is the case for content moderation. By asymptotic optimality, we mean that as the system size increases, the algorithm obtains nearly the same long-run average holding cost as the optimal algorithm: 
$\lim_{N \to \infty}\frac{C(N,\nalg)}{ \inf_{\nalg' \in \Pi} C(N,\nalg')} = 1.$
Vanishing suboptimality gap for $\nalg$ implies the above convergence for overloaded systems as for any algorithm $\nalg'$, its cost $C(N,\nalg')$ is at least $c(\rootNode)(\lambda - \mu)N$, i.e., linear in $N$.

\begin{remark}\label{remark:general-momdel}   
The assumption that job state transitions form a directed tree with a finite depth  is without loss of generality. For any algorithm with vanishing suboptimality gap in our model, we can slightly modify it to achieve the same guarantee in a model where state transitions form a general discrete-time Markov chain (DTMC). This claim is valid as long as there is non-zero abandonment probability. The main idea is to expand the DTMC into a tree by the state history of a job and remove states with a depth larger than $L = \lceil -\ln_{1-\theta} N \rceil$ as a job is unlikely to reach them. See Appendix~\ref{app:general-markov} for details.
\end{remark}

\paragraph{Notation.} To ease the exposition of jobs' tree Markov chain, we introduce the following notation: for a state $i \in \set{S}$, its ancestor set $\anc(i)$ includes states on the unique path from the root $\rootNode$ to state~$i$ (both included). The subtree of a state $i$, $\sub(i)$, refers to all states $k$ such that $i$ is an ancestor of $k$, i.e., $i \in \anc(k).$ We also write $\sub(\set{X})$ to denote the subtree of a set of states $\set{X}$ such that $\sub(\set{X}) = \bigcup_{i \in \set{X}} \sub(i).$ Similarly, $\anc(\set{X}), \pa(\set{X}), \child(\set{X})$ denote the (unioned) set of ancestors, parents and children for states in $\set{X}.$ Throughout this paper, we use $\mathbb{R}$ to denote the set of real numbers, $\mathbb{R}_{\geq 0}$ to denote the set of non-negative real numbers, and $\set{R}_{\geq 0}^{\set{S}}$ to denote the set of non-negative vectors defined over set $\set{S}$. For any real value $x$, $(x)^+$ denotes $\max(x,0).$ If $x$ is a positive integer, we use $[x]$ to denote the set $\{1,\ldots,x\}.$ We use the notation $\perp$ to denote an empty state and view $\{\perp\}$ as an empty set $\emptyset$ with no element. We include a table of notation in Appendix~\ref{app:table} for frequently used notation in this paper.
\section{Markovian Ski-Rental, Fluid Relaxation, and Our Algorithm}\label{sec:algorithm}
% !TEX root = main.tex
This section presents our algorithm \textsc{Opportunity-adjusted Remaining Cost} ($\alg$) that achieves asymptotic optimality. We first give intuition of the algorithm via the connection between our problem and a Markovian ski-rental problem in Section~\ref{sec:ski-rental}. Solving the ski-rental problem requires the input of a suitable capacity price, which we address in Section~\ref{sec:fluid-lp} by analyzing the dual of a  linear program. Section~\ref{sec:pafou} gives the full algorithm description of $\alg$ and its guarantee. 

\subsection{Intuition: Markovian ski-rental}\label{sec:ski-rental}
To provide the intuition behind $\alg$, we focus on a warm-up setting where, instead of a capacity constraint $R(t)$ for each period, there is a cost $\gamma$ that the algorithm pays for each job it selects to serve. In this warm-up setting, each job is independent of other jobs and follows a \emph{Markovian ski-rental} problem. In typical ski-rental, a skier who is going to ski for an uncertain number of days, makes daily decisions on whether to pay a fixed cost of buying equipment or to pay a unit cost of renting them. Our setting is a Markovian version of this problem as, in every period, the algorithm can either pay a fixed cost $\gamma$ to serve a job (buying) or pay state-dependent cost $c(i)$ if the job is in state $i$ (renting). Formally, suppose there is a job whose initial state is $S_1 = i$ but its future states $S_2,S_3,\dots$ are uncertain. In a period $\tau$, the algorithm can pay an upfront cost $\gamma$ to stop the job from incurring any further cost or do nothing. In the latter case, the job incurs a cost $c(S_{\tau})$ and transitions into the next state according to the Markov chain in Section~\ref{sec:model}. 

We can minimize the total cost of a job via the following dynamic program: letting $V(\gamma, i)$ be the cost-to-go function for state $i$, the Bellman equation gives:
\begin{equation}\label{eq:bellman}
V(\gamma, i) = \min\left\{\gamma, c(i) + V^f(\gamma,i)\right\} \text{ where }V^f(\gamma,i) = \sum_{k \in \child(i)} P(i,k)V(\gamma, k).
\end{equation}
In the minimization problem, the term $\gamma$ corresponds to the cost of serving a job, and the term $c(i) + V^f(\gamma,i)$ represents the cost of \emph{not serving it}, which involves both the known instantaneous cost $c(i)$ and the expected future cost $V^f(\gamma,i).$ 

The cost of not serving a job, $c(i) + V^f(\gamma,i)$, naturally gives an index measuring the value of serving each job, which motivates $\alg$. Specifically,  given a suitable capacity  price $\gamma^\star$ capturing the cost of serving a job, the algorithm assigns, in each period, an index to each waiting job with state $i$, which adjusts to the opportunity of serving the job in the future: 
\begin{equation}\label{eq:index-alg}
\ind_{\alg}(i) = c(i) + V^f(\gamma^\star,i).
\end{equation}
 The index has two components. The first component is the \emph{instantaneous prevented cost}, i.e., serving a job prevents its cost for this period, which is certain and equal to $c(i)$ for a job with state~$i$. The second component is the \emph{future prevented cost} $V^f(\gamma^\star,i)$, i.e.,  serving a job also prevents its future cost. However, this cost is uncertain as the job can transition to different states. The key novelty of our algorithm is to calculate this future uncertainty via the cost-to-go function derived in \eqref{eq:bellman}. Therefore, even if the job can have very high future cost, $\alg$ may  not act now; instead, it may wait until the uncertainty gradually realizes, adjusting for the opportunity to serve this job later. Implementing $\alg$ requires a suitable capacity price $\gamma^\star$. 

\subsection{Computing the capacity price: a fluid linear program and its dual}\label{sec:fluid-lp}
To correctly price the capacity, we consider a fluid version of the  problem, which removes randomness in the arrival and service process. To that end, for a job state $i \in \mathcal{S}$, let $q_i$ be the normalized fluid queue length of state-$i$ jobs which approximates the number of waiting jobs of state $i$ at the beginning of a period. Moreover, let $\nu_i$ be the normalized fluid number of served state-$i$ jobs. That is, one expects $Nq_i \approx \left|\{j \in \set{Q}(t) \colon S_j(t) = i\}\right|$ and $N\nu_i \approx \left|\{j \in \set{R}(t) \colon S_j(t) = i\}\right|$ when the system size is $N$. Defining the (fluid) queue-length decision variable $\bolds{q} = (q_i)_{i \in \set{S}}$ and the (fluid) service decision variable $\bolds{\nu} = (\nu_i)_{i \in \set{S}}$, we obtain \eqref{eq:orifluid}, which we explain next.

The objective of \eqref{eq:orifluid} corresponds to the long-run average holding cost defined in \eqref{eq:def-cost} where each state-$i$ job that is in the queue and not served, captured by the term $q_i - \nu_i$, contributes cost $c(i)$. The constraints in \eqref{eq:orifluid} represent the  transition constraint and the capacity constraint. The transition constraint requires that the number of state-$i$ jobs in a period is (in expectation) equal to the number of unserved state-$\pa(i)$ jobs in the last period multiplied by the probability of state $\pa(i)$ transitioning into state $i$. The capacity constraint requires that an algorithm cannot serve more jobs than what are currently waiting or the available capacity.

\framebox(220,150){
 \begin{minipage}{3in}
\begin{equation}\label{eq:orifluid}
\tag{OriginalFluid}
\begin{aligned}
&C^\star = \min_{\bolds{q}, \bolds{\nu} \in \mathbb{R}_{\geq 0}^{\set{S}}} \sum_{i \in \set{S}} c(i)(q_i-\nu_i)\\
&~\text{s.t.} ~~ q_{\rootNode} = \lambda, \\
& q_i = (q_{\pa(i)} - \nu_{\pa(i)})P(\pa(i),i),~\forall i \in \set{S} \setminus \{\rootNode\} \\
&\nu_i \leq q_i~\forall i \in \set{S}  \\
&\sum_{i \in \set{S}} \nu_i \leq \mu.
\end{aligned}
\end{equation}
\end{minipage}
}
\framebox(220,150){
\begin{minipage}{3in}
\begin{equation}
\label{eq:simfluid}
\tag{\text{SimplerFluid}}
\begin{aligned}
C^\star &= \lambda c^f(\rootNode)-\max_{\bolds{\nu} \in \mathbb{R}_{\geq 0}^{\set{S}}}\left(\sum_{a \in \set{S}} \nu_a c^{f}(a)\right) \\
&\text{s.t.}  \sum_{a \in \anc(i)} \nu_a \cdot \frac{\pi(i)}{\pi(a)} \leq \lambda \pi(i),~\forall i\\
& ~~~~~~~\sum_{i \in \set{S}} \nu_i \leq \mu.
\end{aligned}
\end{equation}
\end{minipage}
}

We reformulate \eqref{eq:orifluid} into \eqref{eq:simfluid} which only contains the service decision variable. To do so, recalling that $\anc(i)$ denotes the set of ancestors of a state $i$ (including $i$), let $\pi(i) = \prod_{k \in \anc(i)} p(k)$ be the probability of a job passing state $i$ if not served in the meantime. The transition constraint in \eqref{eq:orifluid} implies that, given a service decision variable $\bolds{\nu}$, the queue decision variable $\bolds{q}$ is
\begin{equation}\label{eq:def-fluid-q}
q_i = \lambda \pi(i) - \sum_{a \in \anc(i) \setminus \{i\}} \nu_a \cdot \frac{\pi(i)}{\pi(a)}.
\end{equation}
To intuitively understand this equation, the term $\lambda \pi(i)$ captures the amount of state-$i$ jobs if there is no service applied to any job whose state is an ancestor of state $i$ (hereafter referred to as an \emph{ancestor job}). The second term $\sum_{a \in \anc(i) \setminus \{i\}} \nu_a \cdot \frac{\pi(i)}{\pi(a)}$ captures the expected amount of served ancestor jobs that would have  become state-$i$ jobs if they had not been served, since $\frac{\pi(i)}{\pi(a)}$ is the probability of a job becoming state $i$ conditioning on it being state $a$. As a result, the second term subtracts all jobs that were served in states prior to $i$ and therefore never reached $i$. Combining it with the requirement that $\nu_i \leq q_i$ and $q_{\rootNode} = \lambda$, the constraints in \eqref{eq:orifluid} are equivalent to 
the constraints in \eqref{eq:simfluid} because
\begin{equation}\label{eq:equiv-constraint}
\nu_i \leq q_i \Leftrightarrow \nu_i \leq \lambda \pi(i) - \sum_{a \in \anc(i) \setminus \{i\}} \nu_a \cdot \frac{\pi(i)}{\pi(a)} \Leftrightarrow \sum_{a \in \anc(i)} \nu_a \cdot \frac{\pi(i)}{\pi(a)}  \leq \lambda \pi(i).
\end{equation}
Moreover, the objective of \eqref{eq:orifluid} becomes 
\begin{equation}\label{eq:sim-obj}
\sum_{i \in \set{S}} c(i) (q_i - \nu_i) = \sum_{i \in \set{S}} c(i)\left(\lambda \pi(i) - \sum_{a \in \anc(i)} \nu_a \cdot \frac{\pi(i)}{\pi(a)}\right) = \lambda \sum_{i \in \set{S}} \pi(i)c(i) - \sum_{a \in \set{S}} \nu_a \sum_{i \in \sub(a)} c(i)\cdot \frac{\pi(i)}{\pi(a)},
\end{equation}
where we recall that $\sub(a)$ is the subtree of state $a$. We define $c^{f}(a) = \sum_{i \in \sub(a)} c(i)\cdot \frac{\pi(i)}{\pi(a)}$ as
the expected \emph{future} cost of a job conditioning on it being in state $a$. The objective of \eqref{eq:orifluid} then becomes the same as the one of \eqref{eq:simfluid}. Hence, the two programs are equivalent.

Given that the term $\lambda c^f(\rootNode)$ in \eqref{eq:simfluid} is a constant, we translate the objective of \eqref{eq:simfluid} to $P^\star = \lambda c^f(\rootNode) - C^\star$, which corresponds to the \emph{maximum prevented cost} in the objective of \eqref{eq:simfluid} and is the objective for a linear program. With the \emph{state duals} $\bolds{\beta} \in \mathbb{R}_{\geq 0}^{\set{S}}$ corresponding to the duals of the first constraint in \eqref{eq:simfluid} and the \emph{capacity dual} $\gamma \geq 0$ being the dual of the second constraint in \eqref{eq:simfluid}, the dual problem of $P^\star$ is given by $D^\star = \min_{\gamma\geq 0} D^{\star}(\gamma)$, where 
\begin{equation}
\label{eq:dual}
\tag{\text{Dual}}
\begin{aligned}
D^{\star}(\gamma) &= \min_{\bolds{\beta} \in \set{R}^{\set{S}}_{\geq 0}} \lambda \sum_{i \in \set{S}} \pi(i) \cdot \beta_i  + \mu \cdot \gamma\\
\text{s.t.}~~~&  \gamma + \sum_{i \in \sub(a)} \beta_i \cdot \frac{\pi(i)}{\pi(a)} \geq c^f(a),~\forall a \in \set{S}.
\end{aligned}
\end{equation}
Intuitively, the dual $\gamma$ corresponds to the capacity price in the Markovian ski-rental problem (Section~\ref{sec:ski-rental}). Indeed, we show that the optimal state duals closely connect with the cost of not serving a job given a capacity price $\gamma$, $c(i)+V^f(\gamma,i)$, in \eqref{eq:bellman}. The below result (shown in Appendix~\ref{app:lem-dual-structure}) shows that $c(i)+V^f(\gamma,i)$ is related to an optimal state dual. Moreover, there is a clean form of $D^\star(\gamma)$ that $\alg$ later relies on.
\begin{lemma}\label{lem:dual-structure} 
For the dual program $D^\star(\gamma)$ in \eqref{eq:dual}, (i) it has an optimal solution $\bolds{\beta}^\star(\gamma)$ with $\beta_i^\star(\gamma) = \max\left\{0, c(i) + V^f(\gamma, i) - \gamma\right\}$ and 
(ii) $D^\star(\gamma) = \mu \cdot \gamma + \lambda\left(c^f(\rootNode) - V(\gamma,\rootNode)\right).$ 
\end{lemma}

\paragraph{Connection between the fluid relaxations and the two canonical algorithmic principles.}
The two equivalent optimization problems provide natural motivations for the two canonical algorithmic principles in the literature (discussed in Section~\ref{sec:intro}), which respectively solve \eqref{eq:orifluid} and \eqref{eq:simfluid} \emph{greedily} without considering the job state transitions. 
\begin{itemize}
\item The instantaneous cost principle (i.e. the $c\mu$-rule) schedules jobs with the highest current cost~$c(i)$ which ostensibly minimizes the objective of \eqref{eq:orifluid} by allocating service capacity to those states with highest instantaneous cost $c(i)$. It is however suboptimal because it ignores the impact of service on queue length $\bolds{q}$ through the transition equation. 
\item The expected remaining cost principle (i.e., the $c\mu/\theta$-rule) seeks to optimize \eqref{eq:simfluid} as it allocates capacity to those states $a$ with highest expected remaining cost $c^f(a)$. However, the first constraint in \eqref{eq:simfluid} implies that allocating capacity to a state $a$ also limits the demand (number of jobs) for a future state $i$ in the subtree of $a$. Not considering such \emph{externalities} is suboptimal because, as in Figure~\ref{fig:failure}, it may be better to delay the capacity for a new \textsc{Video} job to a \textsc{RedVideo} job, whose uncertainty is resolved and we can avoid wasting capacity in serving a \textsc{BlueVideo} job.
\end{itemize}

\subsection{The index algorithm $\alg$ and its guarantee}\label{sec:pafou}

Our algorithm $\alg$ (Algorithm~\ref{algo:pafou}) consists of an offline training component and an online scheduling component. In the offline training component, $\alg$ computes the optimal capacity dual $\gamma^{\star}$ with the following single-variable optimization problem: 
\[
\gamma^\star \in \arg\min_{\gamma \geq 0} \left\{ \mu \cdot \gamma - \lambda V(\gamma, \rootNode) \right\}.
\]
This problem gives the optimal capacity dual $\gamma^\star$ by the form of $D^\star(\gamma)$ in Lemma~\ref{lem:dual-structure}.  A subroutine of solving $\gamma^\star$ is the calculation of the cost-to-go functions $V(\gamma,\cdot)$ and $V^f(\gamma,\cdot)$ given $\gamma$. Computationally, they are solvable with time complexity linear in $|\set{S}|$ via the Bellman equation \eqref{eq:bellman} and the tree structure of the Markov chain. In practice, we rely on approximate dynamic programming techniques that Section~\ref{sec:numerics} further discusses. The output of the training component is an index for any state $i$ as in \eqref{eq:index-alg}:
$\ind_{\alg}(i)=c(i)+V^f(\gamma^\star, i).$
In the online scheduling component, $\alg$ operates as an index algorithm. For a period $t$, $\alg$ assigns to each waiting job $j \in \set{Q}(t)$ an index $\ind_{\alg}(S_j(t))$ based on its current state and serves jobs in a decreasing order of their indices. Technically, two states can have the same indices. We assume a \emph{consistent} tie-breaking rule such that the algorithm always prioritizes one state over the other state. This rule, e.g., assigns unique identifiers to states and compares the identifier for two states with the same indices. 

\begin{algorithm}[H]
\LinesNumbered
\DontPrintSemicolon
  \caption{
  \textsc{Opportunity-adjusted Remaining Cost} ($\alg$)}\label{algo:pafou}
  \KwData{cost vector $\bolds{c}$ and transition probability $P$;~Normalized arrival and service rates $\lambda, \mu$}
  \tcc{Offline training}
  Solve $\gamma^\star = \arg\min_{\gamma \geq 0} \left\{\mu \gamma - \lambda V(\gamma,\rootNode)\right\}$ and $V^f(\gamma^\star,\cdot)$ by \eqref{eq:bellman} \; 
  \tcc{Online scheduling}
  \For{$t = 1$ \KwTo $T$}{
    Observe the set of waiting jobs $\set{Q}(t)$ and the random services $R(t)$ \;
    Assign to each job $j \in \set{Q}(t)$ an index $\ind_{\alg}(S_j(t))$\; \label{algoline:index}
    Serve the $R(t)$ jobs from $\set{Q}(t)$ with highest indices using a consistent tie-breaking rule \;
    Incur cost of unserved jobs and observe new arrivals \;
  }
\end{algorithm}

Our main result for $\alg$ is as follows. Let $c_{\max} = \max_{i \in \set{S}} c(i)$ be the maximum instantaneous cost.
Recall that the transition probability from every state to the empty state is non-zero, i.e., for some $\theta > 0$, $\sum_{k \in \child(i)} P(i,k) \leq 1 - \theta, \forall i\in\set{S}$. We show that $\alg$ enjoys a $\tilde{O}(\sqrt{N})$ suboptimality gap for system size $N$. 
\begin{theorem}\label{thm:pafou}
For any system size $N$, the suboptimality gap of $\alg$ is at most $U\Big( \theta, \ln(NL)\Big)\sqrt{N}$, where $U\Big(\theta, \ln(NL)\Big) = 1930c_{\max}\theta^{-3.5}\ln^2(73NL).$ 
\end{theorem}
Theorem~\ref{thm:pafou} shows two appealing properties of $\alg$. First, it is asymptotically optimal when the system is overloaded as it enjoys vanishing suboptimality gap. Second, its suboptimality gap is \emph{independent} of the size of the state space $|\set{S}|$. This is favorable in practice where the state of a job can incorporate contextual information; see e.g. \cite{haimovich2021popularity} and \eqref{eq:features} in our numerical section. The logarithmic dependence on the tree depth $L$ arises from our analysis technique (e.g., our tail bound in Lemma~\ref{lem:tail-bound}) rather than reflecting a fundamental limitation. In fact, the same trick as in Remark~\ref{remark:general-momdel} can remove this dependence by restricting the tree length $L$ to be approximately $-\ln(N)/\ln(1-\theta)$.   

To solve the value function, $\alg$ requires knowledge of the transition kernel. In practice, this kernel is unknown; what is available is a dataset of the holding-cost trajectories and the state evolution of historical jobs. To address this challenge, our numerical section follows a model-free reinforcement learning approach (Section~\ref{sec:prac-imp}) to directly estimate the value function based on historical data, without estimating the actual transition kernel. Our theoretical results can be expanded with a robustness guarantee of $\alg$ when it is run with inaccurate value functions. In particular, if the algorithm has access to estimated indices $\{\ind_{\esti}(i)\}_{i \in \set{S}}$ that satisfy $\max_{i \in \set{S}} |\ind_{\alg}(i) - \ind_{\esti}(i)| \leq \varepsilon$ for some error $\varepsilon$, it enjoys a suboptimality gap that is at most $(4\varepsilon/\theta)N + U(\theta,\ln(NL))\sqrt{N}$ (see Appendix~\ref{app:robustalg}). $\alg$ is thus robust to estimation error in the indices in the sense that its performance only degrades linearly  with the error.
\begin{remark}
As discussed above, the theoretical guarantees of our algorithm  degrade linearly with the error in value function estimation (Appendix~\ref{app:robustalg}). In practice this error is tied to how fine-grain the state space is (e.g., how much information the platform captures about a content piece's evolution) and the number of samples available to estimate the value function. Our error bound thus provides guidance for practitioners to navigate the tradeoffs between model complexity, sample complexity, and the algorithmic efficiency.
\end{remark}
\section{Analysis}\label{sec:analysis}
% !TEX root = main.tex
The proof of Theorem~\ref{thm:pafou} contains four components. The first component identifies a fluid solution to \eqref{eq:orifluid} from any (state) priority algorithm $\prio(\bolds{o})$ which, given an ordered list $\bolds{o}$ of states, serves jobs based on this ordering. Applied with a consistent tie-breaking rule, $\alg$ is such a priority algorithm. The second component of the proof is showing that the fluid solution corresponding to $\alg$ is indeed an optimal solution to \eqref{eq:orifluid}. The third component of the proof shows that the long-run average holding cost of a priority algorithm is upper bounded by the (fluid) cost of the corresponding fluid solution in \eqref{eq:orifluid} (scaled by the system size $N$), plus a term that is of the order of $\sqrt{N}$ ignoring poly-logarithmic factors. The last component shows that the optimal value of \eqref{eq:orifluid} scaled by $N$ indeed lower bounds the long-run average holding cost of any feasible algorithm.

As a warm up, Section~\ref{sec:fluid-sol} constructs a fluid solution for any priority algorithm and motivates it via a water-filling argument. Section~\ref{sec:proof-pafou} presents key lemmas and the full proof of Theorem~\ref{thm:pafou}. Section~\ref{sec:optimal-fluid} shows that the constructed fluid solution of $\alg$ is optimal to \eqref{eq:orifluid}. Section~\ref{sec:stochastic} bounds the long-run average holding cost of a priority algorithm by the cost of its constructed fluid solution.

\subsection{Corresponding fluid solutions of priority algorithms}\label{sec:fluid-sol}

Our analysis starts by constructing a fluid solution for any priority algorithm. This fluid solution provides guidance on what the stochastic system would look like when the scheduling decisions follow the given priority algorithm. Formally, a \emph{priority ordering} $\bolds{o} = (o_1,\ldots,o_{|\set{S}|})$  ranks states in~$\set{S}$ from highest ($o_1$) to lowest priorities ($o_{|\set{S}|}$) and specifies a priority algorithm $\prio(\bolds{o}).$  In a period~$t$ with queue $\set{Q}(t)$ and $R(t)$ available servers, the algorithm serves jobs in the order prescribed by~$\bolds{o}.$ Let $Q_i(t)$ be the number of state-$i$ jobs in $\set{Q}(t)$ and $o^{-1}(i)$ be the position of state $i$ in the priority ordering. Abusing our notation, we use vector $\bolds{R}(t) = (R_i(t))_{i \in \set{S}}$ where $R_i(t)$ is the number of served state-$i$ jobs in period $t$. Under a priority ordering $\bo$, 
\begin{equation}\label{eq:priority-service}
R_i(t) = \min\left(Q_i(t), \left(R(t) - \sum_{h=1}^{o^{-1}(i)-1} Q_{o_h}(t)\right)^+\right).
\end{equation}
That is, the number of served state-$i$ jobs is the minimum between the number of waiting state-$i$ jobs $Q_i(t)$ and the remaining available servers after serving jobs with states of higher priorities. Index-based algorithms, including $c\mu-$rule, $c\mu/\theta-$rule, and $\alg$, are priority algorithms as long as they use consistent tie-breaking rules. The priority ordering of an index-based algorithm orders states with decreasing indices and breaks ties with the consistent tie-breaking rule. 

Under any priority algorithm, the queueing system has a stationary distribution. Formally, letting the \emph{system state} in period $t$ be $(\bolds{Q}(t), R(t))$ with $\bolds{Q}(t) = (Q_i(t))_{i \in \set{S}}$, the system evolves as a discrete-time Markov chain. Restricting to system states reachable from the initial empty queue ($Q_i(1) = 0$ for any $i$), the below lemma (proved in Appendix~\ref{app:lem-aperiodic}) shows that the Markov chain for system states is aperiodic and irreducible. Since it also has finitely many states, it has an unique stationary and limiting distribution by \cite[theorem 9.4]{harchol2013performance}, which we denote by $(\bolds{Q}(\infty), R(\infty)).$
\begin{lemma}\label{lem:aperiodic}
Restricting to system states reachable from $(\bolds{Q}(1) = \bolds{0},R(1)=n)$ for some $n \in \{0,\ldots,N\}$, the system states form an aperiodic, irreducible, and finite Markov chain.
\end{lemma}

A first step of understanding the long-run average holding cost of $\prio(\bolds{o})$ is \emph{guessing} the mean (or informally \emph{equilibrium}) of $\bolds{Q}(\infty)$ and $\bolds{R}(\infty) = \lim_{t \to \infty} \bolds{R}(t)$. The latter exists because $\bolds{R}(t)$ is a function of $\bolds{Q}(t)$ and the binomial random variable $R(t)$ in \eqref{eq:priority-service}. We denote the equilibrium of a priority algorithm $\prio(\bo)$ by $\bolds{q}^{\bolds{o}} = (q^{\bo}_i)_{i \in \set{S}}$ and $\bolds{\nu}^{\bolds{o}} = (\nu^{\bo}_i)_{i \in \set{S}}$. 

\begin{figure}
\centering
\includegraphics[width=6in]{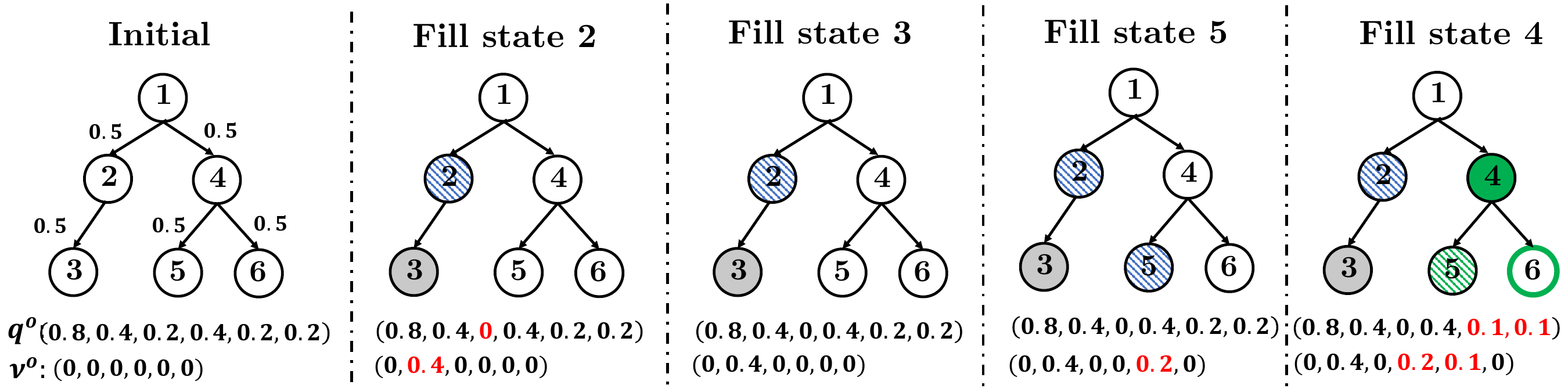}
\caption{An example for the water-filling procedure of constructing $(\bolds{q}^{\bo}, \bolds{\nu}^{\bo})$. There are six states, $\set{S} = \{1,2,3,4,5,6\}$ with transition probability as given in the leftmost figure. The arrival rate and service rate are $\lambda = 0.8$ and $\mu = 0.7.$ The priority ordering is $(2,3,5,4,6,1).$ States with stripes are blocked states (State $2$ is fully, State $5$ is partially). State $4$ is the partially-served state and State $6$ is a partially-reduced state. State $3$ is an empty state and State $1$ is a un-reduced state.} 
\label{fig:water-filling}
\end{figure}
Our construction of the equilibrium  $(\bolds{q}^{\bo}, \bolds{\nu}^{\bo})$ simulates a water-filling procedure for $\prio(\bo)$. To give intuition, we follow Figure~\ref{fig:water-filling} for a concrete example with six states $\set{S} = \{1,2,3,4,5,6\}$ and a priority ordering $\bolds{o} = (2,3,5,4,6,1)$. The arrival rate is $\lambda = 0.8$ and the service rate is $\mu = 0.7$. A state-$1$ job has equal probability to evolve to a state-$2$ or state-$4$ job. A state-$2$ job has equal probability to evolve to a state-$3$ job or abandon. A state-$4$ job has equal probability to evolve to a state-$5$ or state-$6$ job. Jobs of states $3,5$ or $6$ abandon if not getting served. The water-filling procedure serves jobs based on the priority ordering $\bo$. Initially no state receives service and thus $q_i^{\bo} = \lambda \pi(i)$ and $\nu_i^{\bo} = 0$ for any $i$. Following the priority, the procedure works as follows:
\begin{itemize}
\item First, it fills state $2$. Since $\mu = 0.7$ and $q^{\bo}_2 = 0.4$, it has enough capacity to serve all state-$2$ jobs by setting $\nu^{\bo}_2 = 0.4$. This also has downstream effect: there will be no more state-$3$ jobs as all state-$2$ jobs are served. Therefore $q^{\bo}_3$ becomes zero. 
\item Second, it fills state $3$. However, there is no state-$3$ job and thus it consumes no service.
\item Third, it fills state $5$. Filling it is similar to the case for state $2$ since the remaining service rate is $0.7-0.4 = 0.3$ and $q^{\bo}_5 = 0.2$. It increases $\nu^{\bo}_5 = 0.2$ to serve all state-$5$ jobs.
\item Fourth, it fills state $4$. Unlike the above steps, the remaining service rate $0.7 - 0.4 - 0.2 = 0.1$ is insufficient to serve all state-$4$ jobs as $q^{\bo}_4 = 0.4 > 0.1.$ One intuitive step would be to set $\nu^{\bo}_4 = 0.1$ as that is the remaining capacity. However, this ignores the \emph{downstream effect} of serving state $4$ on state $5$. Specifically, serving one state-$4$ job reduces the (expected) number of state-$5$ jobs by half, which increases the service rate available to serve state-$4$ jobs. As a result, setting $\nu_4^{\bo} = 0.2$ reduces $q^{\bo}_5$ from $0.2$ to $0.1$ and thus we only need $\nu^{\bo}_5 = 0.1$ to serve all state-$5$ jobs. It then leaves a capacity of $0.7-0.4-0.1 = 0.2=\nu_4^{\bo}$ for state $4$.
\item Finally, since no capacity remains after the last step, states $1$ and $6$ receive no service. 
\end{itemize}

To formalize the above water-filling procedure,  we introduce the notion of \emph{top sets}. For a set of states $\set{X}$, serving its top set, $\TopS(\set{X})$, ensures service of all states in $\set{X}$ with minimum capacity. Given that the Markov chain is a rooted tree, the top set is defined as the unique set such that
\begin{equation}\label{eq:top-set}
(i) \TopS(\set{X}) \subseteq \set{X};\quad(ii) \set{X} \subseteq \sub(\TopS(\set{X}));\quad(iii) k \not \in \sub(i),\quad\forall i \neq k \in \TopS(\set{X}). 
\end{equation}
The first and third conditions ensure that $\TopS(\set{X})$ is a minimum subset of $\set{X}$ while the second condition ensures that $\set{X}$ is inside this subset's subtree $\sub(\TopS(\set{X}))$.

Top sets enable a succinct construction of the equilibrium $(\bolds{q}^{\bo}, \bolds{\nu}^{\bo})$. Recall that the probability of a job passing state $i$ is $\pi(i) = \prod_{k \in \anc(i)} p(k)$. Let $m \in \{0,\ldots,|\set{S}|\}$ be the maximum position in the ordering $\bo$ such that the service rate $\mu$ is sufficient to serve all states in $\TopS(o_{[m]})$ where the notation $\bo_{\set{H}}$ denotes the set $\{o_h \colon h \in \set{H}\}$. Formally, $m$ is the maximum position with $\sum_{i \in \TopS(\bo_{[m]})} \lambda \pi(i) \leq \mu$. If the inequality is strict and $m < |\set{S}|,$ let $\partialNode = o_{m+1}$ be the (only) state that is partially served. Otherwise, we let $\partialNode = \perp$. The water-filling procedure prescribes an equilibrium based on six types of states (labeled as in Figure~\ref{fig:water-filling}):  
\begin{enumerate}[label=\textnormal{(T-\arabic*)}]
\item An \emph{un-reduced} state, $i \in \set{S} \setminus \sub(\partialNode) \setminus \sub(\bo_{[m]})$, is such that neither its ancestors nor itself receives any service. This implies that $q_i^{\bo} = \lambda \pi(i)$ and $\nu_i^{\bo} = 0$. \label{item:un-reduced}
\item A \emph{fully-blocked} state, $i \in \TopS(\bo_{[m]}) \setminus \sub(\partialNode)$, is such that it receives full service but its ancestors receive no service. This implies that $q_i^{\bo} = \nu_i^{\bo} = \lambda \pi(i)$. \label{item:fully-blocked}
\item An \emph{empty} state, $i \in \sub(\bo_{[m]}) \setminus \TopS(\bo_{[m]})$, has an ancestor in $\TopS(\bo_{[m]})$. This implies that $q_i^{\bo} = \nu_i^{\bo} = 0$. \label{item:empty}
\item the unique \emph{partially-served} state $\partialNode$ receives non-zero service but is not fully served. Given that none of its ancestors are served, $q_{\partialNode}^{\bo} = \lambda \pi(\partialNode)$. Accounting for the downstream effect that serving state $\partialNode$ has on decreasing the required service of its subtree, $\nu_{\partialNode}^{\bo} = \frac{\mu - \lambda\sum_{i \in \TopS(\bo_{[m]})} \pi(i)}{\kappa}$ where $\kappa$ is the degeneracy parameter defined by
\begin{equation}\label{eq:def-degeneracy}
\kappa = 1 - \sum_{i \in \TopS(\bo_{[m]}) \cap \sub(\partialNode)} \pi(i) / \pi(\partialNode).
\end{equation}
The degeneracy parameter $\kappa$ must be positive as otherwise we can set $m$ to include $\partialNode$. If $\partialNode = \perp$, we define $\kappa = +\infty$, so $\nu_{\partialNode}^{\bo} = 0.$
\label{item:partially-served}
\item A \emph{partially-blocked} state, $i \in \TopS(\bo_{[m]}) \cap \sub(\partialNode)$, receives full service but has a partially-served ancestor. This implies $q_i^{\bo} = \nu_i^{\bo} = \lambda \pi(i) - \nu_{\partialNode}^{\bo}\pi(i) / \pi(\partialNode).$ \label{item:partially-blocked}
\item A \emph{partially-reduced} state, $i \in \sub(\partialNode) \setminus \{\partialNode\} \setminus \sub(\bo_{[m]})$, does not receive any service but has a partially-served ancestor. This implies $q_i^{\bo} = \lambda \pi(i) - \nu_{\partialNode}^{\bo}\pi(i) / \pi(\partialNode)$ and $\nu_i^{\bo} = 0.$ \label{item:partially-reduced}
\end{enumerate}
Summarizing the above, blocked states (either fully or partially) form the set $\TopS(\bo_{[m]})$. Empty states are those in the subtree of blocked states. The partially-served state is $\partialNode$ and the un-reduced states are all the remaining states, i.e., those not in the subtree of $\bo_{[m]}$ and $\partialNode.$ The below result (proved in Appendix~\ref{app:lem-feasible-nu}) verifies that $(\bolds{q}^{\bo}, \bolds{\nu}^{\bo})$ is feasible to \eqref{eq:orifluid} and that the equilibrium uses all capacity when the partially-served state is non-empty. 
\begin{lemma}\label{lem:feasible-nu}
The equilibrium $(\bolds{q}^{\bo},\bolds{\nu}^{\bo})$ is feasible to \eqref{eq:orifluid}. Moreover, if  $\partialNode \neq \perp$, then $\sum_{i \in \set{S}} \nu_i^{\bo} = \mu.$
\end{lemma}

As discussed above, $\alg$ is a priority algorithm. We denote its equilibrium, obtained by taking $\bolds{o}$ to be its priority ordering of states, by $(\bolds{q}^{\alg},\bolds{\nu}^{\alg})$. 

\subsection{Proof of Theorem~\ref{thm:pafou}}\label{sec:proof-pafou}
This section lists three key lemmas and presents the proof of Theorem~\ref{thm:pafou}. Throughout the proof, we assume $N \geq 3$; as otherwise the theorem holds trivially. The first lemma shows that the equilibrium of $\alg$ is an optimal solution to \eqref{eq:orifluid}.
\begin{lemma}\label{lem:optimal-fluid}
The equilibrium of $\alg$ has optimal fluid cost: $\sum_{i \in \set{S}} c(i) (q_i^{\alg} - \nu_i^{\alg}) = C^\star$.
\end{lemma}

The second lemma shows that the long-run average holding cost for any priority algorithm is within $\tilde{O}(\sqrt{N})$ to the (fluid) cost of its equilibrium (as in the objective of \eqref{eq:orifluid}). Recall that $\theta = \min_{i \in \set{S}} \left(1 - \sum_{k \in \child(i)} P(i,k)\right)$ is the minimum abandonment probability of every state. Recall the following quantity which is a polynomial of $\theta^{-1}$ and $\ln(NL)$:
\begin{equation*}
U(\theta,\ln(NL)) \coloneqq 1930c_{\max}\theta^{-3.5}\ln^2(73NL).
\end{equation*}
\begin{lemma}\label{lem:stochastic}
If $\theta > 0$, then for any priority ordering $\bo$ and system size $N$, its long-run average holding cost is upper bounded by $C(N,\prio(\bo)) \leq N\sum_{i \in \set{S}} c(i) (q_i^{\bo} - \nu_i^{\bo}) + U(\theta,\ln(NL))\sqrt{N}.$
\end{lemma}

The last lemma shows that the optimal fluid cost in \eqref{eq:orifluid} lower bounds the long-run average holding cost of any feasible algorithm.
\begin{lemma}\label{lem:lower-bound}
For any system size $N$ and feasible algorithm $\nalg$, the fluid LP \eqref{eq:orifluid} lower bounds the long-run average holding cost by $C(N,\nalg) \geq NC^\star.$
\end{lemma}
Sections~ \ref{sec:optimal-fluid} and \ref{sec:stochastic} prove Lemmas~\ref{lem:optimal-fluid} and \ref{lem:stochastic} respectively. We defer the proof of Lemma~\ref{lem:lower-bound} to Appendix~\ref{app:lem-lower-bound} as it follows a standard argument.
\begin{proof}[Proof of Theorem~\ref{thm:pafou}]
Since $\alg$ is a priority algorithm, for any system size $N$,
\begin{align*}
\subopt(N,\alg)&=C(N,\alg) - \inf_{\nalg \in \Pi} C(N,\nalg)\\ 
\overset{\text{Lemma }\ref{lem:lower-bound}}{\leq} &C(N,\alg) - NC^\star \\
\overset{\text{Lemma}~\ref{lem:stochastic}}{\leq}&N\sum_{i \in \set{S}} c(i)(q_i^{\alg} - \nu_i^{\alg}) +U(\theta,\ln(NL))\sqrt{N} - NC^\star \\
\overset{\text{Lemma}~\ref{lem:optimal-fluid}}{=} &NC^\star + U(\theta,\ln(NL))\sqrt{N} - NC^\star = U(\theta,\ln(NL))\sqrt{N}.
\end{align*}
\end{proof}

\subsection{Optimality of $\alg$'s fluid cost (Lemma~\ref{lem:optimal-fluid})}\label{sec:optimal-fluid}
To show fluid optimality, we partition states into three classes $\set{S}_{\textsc{Hi}}, \set{S}_{\textsc{Eq}}, \set{S}_{\textsc{Lo}}$ corresponding to states with indices $\ind_{\alg}(i) = c(i) + V^f(\gamma^\star,i)$ higher, equal, or lower than $\gamma^\star$ respectively. The next lemma shows that if a priority ordering $\bo$ ranks $\set{S}_{\mhigh} \succ \set{S}_{\mequal} \succ \set{S}_{\mlow}$, i.e., ranks these three classes in order without ordering within each, then its equilibrium is optimal to \eqref{eq:orifluid}. The proof follows a complementary slackness argument and is provided in Appendix~\ref{app:prio-optimal}.
\begin{lemma}\label{lem:prio-optimal}
If a priority ordering $\bo$ ranks $\set{S}_{\mhigh} \succ \set{S}_{\mequal} \succ \set{S}_{\mlow}$, its equilibrium is fluid optimal. 
\end{lemma}
\begin{proof}[Proof of Lemma~\ref{lem:optimal-fluid}]
Since $\alg$ ranks all states in a decreasing order of $c(i) + V^f(\gamma^\star,i)$, its corresponding priority ordering must rank $\set{S}_{\textsc{Hi}}$ before $\set{S}_{\textsc{Eq}}$ before $\set{S}_{\textsc{Lo}}$. Lemma~\ref{lem:prio-optimal} thus applies. 
\end{proof}

\begin{remark}\label{remark:gittins}
By Lemma \ref{lem:prio-optimal}, any priority algorithm that ranks $\set{S}_{\mhigh} \succ \set{S}_{\mequal} \succ \set{S}_{\mlow}$ is fluid-optimal, which implies asymptotic optimality as in the proof of Theorem~\ref{thm:pafou}. This generality in our framework allows us to show asymptotic optimality of a different algorithm based on Gittins index \citep{gittins2011multi,whittle1988restless,scully2025gittins} in Appendix~\ref{app:gittins}. This is interesting as Gittins index typically works for classical bandit problems while we have a restless bandit problem. In our data-driven application (Section~\ref{sec:numerics}), we do not use this Gittins-based algorithm as there are  significant computational challenges of solving Gittins index in data-driven settings; see Section 3.4 of \cite{scully2025gittins} for further discussion.
\end{remark}

\subsection{Technical crux: gap between the stochastic and fluid cost (Lemma~\ref{lem:stochastic})}\label{sec:stochastic}
The key idea of Lemma~\ref{lem:stochastic} is to show $Q_i(\infty) - R_i(\infty) \approx q_i^{\bo} - \nu_i^{\bo}$ for any priority algorithm $\prio(\bo).$ That is, the limiting distribution in the stochastic system closely mimics the equilibrium for the fluid system. Throughout this section, we fix a priority ordering $\bo$ and focus on a system with a scheduling algorithm $\prio(\bo)$ that selects $\bolds{R}(t)$ based on \eqref{eq:priority-service}. To ease notation, let $Z_i(t) = Q_i(t) - R_i(t)$ be the random variable of \emph{remaining state-$i$ jobs} after service in period $t$ and $z_i^{\bo} = q_i^{\bo} - \nu_i^{\bo}$ be the corresponding fluid equilibrium. We also use $\bolds{Z}(t)$ {and} $ \bolds{z}^{\bo}$ to denote the associated vectors. Recall from Section~\ref{sec:fluid-sol} that $m$ is the maximum position such that the service rate $\mu$ is sufficient to serve all states in $\TopS(\bo_{[m]})$: $\sum_{i \in \TopS(\bo_{[m]})} \lambda \pi(i) \leq \mu$. If the inequality is strict and $m < |\set{S}|$, we denote the partially-served state by $\partialNode = o_{m+1}$ and if not, $\partialNode = \perp$. Using the structure of $\bolds{z}^{\bo}$ (Section~\ref{sec:fluid-sol}), our proof of Lemma~\ref{lem:stochastic} splits into two components showing that $Z_i(\infty) \approx 0$ for any $i \in \bo_{[m]}$ and that $Z_{\partialNode}(\infty) \approx z_{\partialNode}^{\bo}$.

\paragraph{Fully serving blocked states and their subtree.} Our core result is that  $Z_i(\infty) \approx 0$ for any $i \in \bo_{[m]}$ since the equilibrium satisfies $z_i^{\bo} = 0$. This result is not trivial: although by the definition of $\bo_{[m]}$, the service rate $\mu$ satisfies $\sum_{i \in \TopS(\bo_{[m]})} \lambda \pi(i) \leq \mu$, the priority algorithm $\prio(\bo)$ may not simply use all available service to states in $\TopS(\bo_{[m]})$. This is because some states in  $\bo_{[m]} \setminus \TopS(\bo_{[m]})$ can have higher priority
than those in $\TopS(\bo_{[m]}).$\footnote{For example, if the priority ordering is $(3,2,5,4,6,1)$ in Figure~\ref{fig:water-filling}, then $\bo_{[m]}$ is $\{3,2,5\}$ and $\TopS(\bo_{[m]}) = \{2,5\}.$ However, state $3 \in \bo_{[m]} \setminus \TopS(\bo_{[m]})$ has a higher priority than both $2$ and $5$.} Our key insight is that, although these states can have higher priority, their steady-state number of jobs is close to zero. As a result, the priority algorithm has enough capacity to serve nearly all jobs of states in $\TopS(\bo_{[m]}).$ This implies, in the steady state, that almost no remaining job is in $\bo_{[m]}$ since the state of a job must first be in $\TopS(\bo_{[m]})$ in order to be in $\bo_{[m]} \setminus \TopS(\bo_{[m]})$.

We leverage state space collapse (SSC) to formalize the above intuition.  Formally, fix a parameter $\delta \in (0,e^{-1})$ that we later tune. The following result establishes state space collapse: the steady-state number of remaining jobs in $\bo_{[m]}$ and its subtree is with high probability close to zero. 

\begin{lemma}\label{lem:iterative-ssc}
If the minimum abandonment probability $\theta > 0$, then 
\[
\Pr\left\{\sum_{i \in \sub(o_{[m]})} Z_i(\infty) > 48\theta^{-1.5}\sqrt{N}\ln^2\frac{1}{\delta}\right\} \leq 73NL\delta^2\ln\frac{1}{\delta}.
\]
\end{lemma}
The proof relies on a few results. The first result shows that the number of waiting jobs with states in any set $\set{X}$ and any period $t$ is at most $N\lambda \sum_{i \in \set{X}} \pi(i) + \tilde{O}(\sqrt{N})$ with high probability. 

\begin{lemma}\label{lem:queue-upper-bound}
For any $t \geq 1$, set of states $\set{X} \subseteq \set{S}$ and $\delta \in (0,1/e]$, the queue length $\sum_{i \in \set{X}} Q_i(t)$ is almost surely upper bounded by $N\min(L,|\set{X}|)$. Moreover, it has a high probability upper bound:
\[
\Pr\left\{\sum_{i \in \set{X}} Q_i(t) \leq N\lambda\sum_{i \in \set{X}}\pi(i) + 3\sqrt{N/\theta}\ln \frac{1}{\delta}\right\} \geq 1 - \delta^2.
\]
\end{lemma}

The intuition behind this result is that without any service, the expected number of jobs with states in $\set{X}$ is equal to $N\lambda \sum_{i \in \set{X}} \pi(i)$ and accounting for service should not increase this number. The difficulty in proving this result is that the numbers of jobs in different states are correlated with each other and thus typical concentration bounds do not immediately apply. We address this issue in its proof (Appendix~\ref{app:lem-queue-upper-bound}) by considering a system with no job departures, which makes jobs uncorrelated with each other and allows the use of concentration bounds.

The second result is a new geometric tail bound for a Lyapunov function, showing that: if with high probability, a Lyapunov function has (i) bounded increase and (ii) negative drift when it is high enough, then the steady-state value of this function decays geometrically fast. 
\begin{lemma}\label{lem:tail-bound}
Fix a Markov chain $\{\bolds{X}(t)\}_{t \geq 1}$ with a state space $\set{X}$ and a limiting  distribution. Suppose there is an event $\set{G}$, a Lyapunov function $\Phi\colon \set{X} \to \mathbb{R}_{\geq 0}$, and constants $\varepsilon, v_{\whp}, v_{\max},B, \Delta \geq 0$ with $\Delta \leq v_{\whp} \leq v_{\max}$, such that $\Pr\{\bolds{X}(\infty) \in \set{G}^c\} \leq \varepsilon$, $\expect{\Phi(\bolds{X}(\infty))} < +\infty$, and $\forall t$,
\begin{enumerate}
\item[(i)] for any $\bolds{x} \in \set{G}$, $\Pr\{\Phi(\bolds{X}(t+1)) - \Phi(\bolds{X}(t)) => v_{\whp} | \bolds{X}(t) = \bolds{x}\} \leq \varepsilon$;
\item[(ii)] for any $\bolds{x} \in \set{G}$ with $\Phi(\bolds{x}) \geq B$, $\Pr\{\Phi(\bolds{X}(t+1)) - \Phi(\bolds{X}(t)) > -\Delta | \bolds{X}(t) = \bolds{x}\} \leq \varepsilon$.
\item[(iii)] $\Phi(\bolds{X}(t+1)) - \Phi(\bolds{X}(t)) \leq v_{\max}$ almost surely.
\end{enumerate}
Then for any integer $\ell \geq 0$, the steady-state Lyapunov function satisfies
\[
\Pr\{\Phi(\bolds{X}(\infty)) > B + 2v_{\whp} \ell\} \leq \left(\frac{v_{\whp}}{v_{\whp}+\Delta}\right)^{\ell+1}+\frac{6(\ell+1)\varepsilon v_{\max}}{v_{\whp}}.
\]

\end{lemma}
The proof of Lemma~\ref{lem:tail-bound} is given in Appendix~\ref{app:lem-tail-bound}. Our analysis instantiates this lemma to a designed Lyapunov function by taking $\varepsilon\approx 1/N^2$; $v_{\whp}, \Delta, B \approx \sqrt{N}; \ell \approx \ln N$; and $v_{\max} \approx N$, which shows that the Lyapunov function is larger than $\tilde{O}(\sqrt{N})$ with probability less than $\tilde{O}(1/N)$.

Although bounds with forms similar to Lemma~\ref{lem:tail-bound} exist in the literature (e.g. \cite[theorem 1]{bertsimas2001performance} and \cite[lemma 4.6]{weng2020optimal}), they do not distinguish between $v_{\whp}$ (a high probability bound on the increase in Lyapunov function)  and $v_{\max}$ (an almost sure upper bound on the increase). When we scale the system size $N$ to infinity, these two quantities can be substantially different ($\sqrt{N}$ versus $N$ as in later analysis) and our bound yields better guarantee than the one implied by \cite[theorem 1]{bertsimas2001performance}, which was not designed for such an asymptotic regime. Admittedly, there exist bounds, such as the one in \cite{hajek1982hitting}, that only require the change in the Lyapunov function to have a bounded moment generating function instead of an absolute upper bound. That said, as we demonstrate in Appendix~\ref{app:comp-hajek}, the bound in \cite{hajek1982hitting} induces a suboptimal dependence on $N$ for our purpose. Therefore, the careful distinction of $v_{\whp}$ and $v_{\max}$ in Lemma~\ref{lem:tail-bound} is a key innovation that enables our guarantee. 

\paragraph{Proof sketch of Lemma~\ref{lem:iterative-ssc}.} 
Letting the system state be $\bolds{X}(t) = (\bolds{Q}(t), R(t), \bolds{Q}(t+1))$, the proof (formalized in Appendix \ref{app:lem-iterative-ssc}) is by showing the below Lyapunov function has negative drift:
\[
\Phi(\bolds{X}(t)) = \sum_{i \in \sub(\bo_{[m]})} Z_i(t) = \sum_{i \in \sub(\bo_{[m]})} (Q_i(t) - R_i(t))
\]
where $Q_i(t)$ is the number of state-$i$ jobs at the beginning of a period $t$ and $R_i(t)$ is the number of served state-$i$ jobs in period $t$. The drift of $\Phi(\bolds{X}(t))$ consists of a positive and a negative part:
\begin{equation}\label{eq:drift-positive-negative}
\Phi(\bolds{X}(t + 1)) - \Phi(\bolds{X}(t)) = \sum_{i \in \sub(\bo_{[m]})} Q_i(t+1) - \left(\sum_{i \in \sub(\bo_{[m]})} R_i(t+1) + \sum_{i \in \sub(\bo_{[m]})} Z_i(t)\right).
\end{equation}
Assuming that the available service $R(t+1)$ is close to $N\mu$ (which is true by a concentration bound up to a $\sqrt{N}$ term) and given that states in $\bo_{[m]}$ have higher priority than states in $\set{S} \setminus \bo_{[m]}$, the service allocated to serve jobs with states in $\bo_{[m]}$ is around $N\mu$ unless there are fewer such jobs. In the latter case, the drift in \eqref{eq:drift-positive-negative} must be non-positive as the first two terms cancel out. In the former case, let $\set{T} = \TopS(\bo_{[m]})$, the top set of states in $\sub(\bo_{[m]})$. \eqref{eq:drift-positive-negative} gives
\begin{align*}
\Phi(\bolds{X}(t + 1)) - \Phi(\bolds{X}(t)) &\lesssim \sum_{i \in \sub(\bo_{[m]})} Q_i(t+1) - N\mu - \sum_{i \in \sub(\bo_{[m]})} Z_i(t) \\
&= \underbrace{\sum_{i \in \set{T}} Q_i(t+1) - N\mu}_{(\text{service})} + \underbrace{\sum_{i \in \sub(\bo_{[m]}) \setminus \set{T}} Q_i(t+1) - \sum_{i \in \sub(\bo_{[m]})} Z_i(t)}_{(\text{abandonment})}.
\end{align*}
The (service) term is non-positive because the service rate is sufficient to serve all jobs in $\set{T}$:
\begin{equation}\label{eq:bound-q-by-mu}
\sum_{i \in \set{T}} Q_i(\infty) \lesssim N\lambda \sum_{i \in \set{T}} \pi(i) \leq N\mu,
\end{equation}
where the first inequality is by Lemma~\ref{lem:queue-upper-bound} and the second inequality uses the definition of~$m$ which is the maximum position that allows all jobs in $\set{T}$ to be served in expectation.

In addition, the (abandonment) term is (with high probability) negative because (1) any job with states in $\sub(\bo_{[m]}) \setminus \set{T}$ in period $t + 1$ must be a remaining job with states in $\sub(\bo_{[m]})$ in period $t$; (2) a remaining job in period $t$ has probability at least $\theta$ to abandon the system. These combined suggest that $(\text{abandonment}) \lesssim -\theta \sum_{i \in \sub(\bo_{[m]})} Z_i(t).$ 

As a result, we can upper bound the drift $\Phi(\bolds{X}(t + 1)) - \Phi(\bolds{X}(t)) \lesssim -\theta \Phi(\bolds{X}(t))$. The proof of Lemma~\ref{lem:iterative-ssc} formalizes this via the below lemma (proved in Appendix~\ref{app:lem-drift-bound})  and applies Lemma~\ref{lem:tail-bound} to translate the drift bound into a high probability bound on the Lyapunov function.

\begin{lemma}\label{lem:drift-bound}
There exists an event $\set{G}$, such that: (i) $\Pr\{\bolds{X}(\infty) \in \set{G}\} \geq 1-3\delta^2$ and (ii) for any $t$, conditioning on $\bolds{X}(t) \in \set{G}$, with probability at least $1 - \delta^2$, the drift satisfies 
\begin{equation}\label{eq:self-reflected}
\Phi(\bolds{X}(t+1)) - \Phi(\bolds{X}(t)) \leq -\theta \Phi(\bolds{X}(t)) + 6\sqrt{N/\theta}\ln\frac{1}{\delta}.
\end{equation}
Moreover, for any $t$, $\Phi(\bolds{X}(t+1)) - \Phi(\bolds{X}(t)) \leq NL$ almost surely. 
\end{lemma}

\paragraph{Bounding the number of jobs in the partially-served state and its subtree.}

We show that if the number of remaining  jobs in $\bo_{[m]}$ is small,  the expected number of remaining jobs in the partially-served state, $\expect{Z_{\partialNode}(\infty)}$, is not much larger than the fluid value $Nz_{\partialNode}^{\bo}.$ The proof of Lemma~\ref{lem:iterative-ssc} does not apply to the partially-served state because~\eqref{eq:bound-q-by-mu} does not hold beyond the maximum position $m$ for which all jobs in $\TopS(\bo_{[m]})$ can be served (in expectation). That said, for the partially-served state, it suffices to bound the \emph{expectation} of the remaining number of jobs rather than provide a high probability guarantee. This is because unlike those fully-served states, there is no downstream effect of serving the partially-served state, as it is the last state that would get capacity in the priority ordering $\bo.$ The proof of Lemma~\ref{lem:connect-partial} is in Appendix~\ref{app:lem-connect-partial}. 

\begin{lemma}\label{lem:connect-partial}
If there exists $\tilde{U}$ such that $\Pr\{\sum_{i \in \sub(\bo_{[m]})} Z_i(\infty) \leq \tilde{U}\sqrt{N}\} \geq \frac{N-1}{N}$, then $\expect{Z_{\partialNode}(\infty)} \leq Nz_{\partialNode}^{\bo}+\frac{4\sqrt{N/\theta}\ln N+\tilde{U}\sqrt{N}+3}{\theta}.$
\end{lemma}

\paragraph{Combining the pieces.} The final step in the proof of Lemma~\ref{lem:stochastic} is the below lemma. Following Lemma~\ref{lem:connect-partial}, this lemma shows that if there is a high-probability bound on $\sum_{i \in \sub(\bo_{[m]})} Z_i(\infty),$ then it implies a bound on the difference between a priority algorithm's long-run average holding cost and its equilibrium's fluid cost. 
\begin{lemma}\label{lem:ssc-cost-diff}
If there exists $\tilde{U}$ such that $\Pr\{\sum_{i \in \sub(\bo_{[m]})} Z_i(\infty) \leq \tilde{U}\sqrt{N}\} \geq \frac{N-1}{N}$, then the cost difference $C(N,\prio(\bo))-N\sum_{i \in \set{S}}c(i)(q_i^{\bo} - \nu_i^{\bo})$ is upper bounded by $10c_{\max} \theta^{-2}\left(\sqrt{N/\theta}\ln N+\tilde{U}\sqrt{N}\right).$
\end{lemma}
\begin{proof}[Proof sketch of Lemma~\ref{lem:ssc-cost-diff}.]
The proof (provided in Appendix~\ref{app:lem-ssc-cost-diff}) uses the fact that the long-run average holding cost of $\prio(\bo)$ is equal to $C(N, \prio(\bo)) = \sum_{i \in \set{S}} c(i) \expect{Z_i(\infty)}.$ As a result, recalling that $z_i^{\bo} = q_i^{\bo} - \nu_i^{\bo}$ is the fluid remaining number of jobs, the expected difference between $C(N,\prio(\bo))$ and the fluid cost of the equilibrium $N\sum_{i \in \set{S}} c(i) z_i^{\bo}$ is 
\begin{align*}
C(N,\prio(\bo)) - N\sum_{i \in \set{S}} c(i)z_i^{\bo}&=\underbrace{\sum_{i \in \sub(\bo_{[m]})}c(i)\expect{Z_i(\infty) - Nz_i^{\bo}}}_{\text{fully served}} + \underbrace{\sum_{i \in \sub(\partialNode) \setminus \sub(\bo_{[m]})} c(i)\expect{Z_i(\infty)-Nz_i^{\bo}}}_{\text{partially served}} \\
&\hspace{0.5in}+ \underbrace{\sum_{i \in \set{S} \setminus \sub(\bo_{[m]}) \setminus \sub(\partialNode)} c(i)\expect{Z_i(\infty) - Nz_i^{\bo}}}_{\text{never served}}.
\end{align*}
The proof proceeds by bounding each of the three terms:
\begin{itemize}
\item For the first term (fully served), $z_i^{\bo} = 0$ for $i \in \sub(\bo_{[m]})$ as such a state is fully-blocked \ref{item:fully-blocked}, empty \ref{item:empty}, or partially-blocked \ref{item:partially-blocked}. Moreover, the lemma assumption yields that $\expect{\sum_{i \in \sub(\bo_{[m]})} Z_i(\infty)}$ is at most $\tilde{O}(\sqrt{N})$, and thus the first summation is at most $\tilde{O}(\sqrt{N})$.
\item For the second term (partially served), observe that $\sub(\partialNode) \setminus \sub(\bo_{[m]})$ includes both the partially-served state $\partialNode$ \ref{item:partially-served} and those partially-reduced states \ref{item:partially-reduced}. For any such state $i$, $z_i^{\bo} = z_{\partialNode}^{\bo} \cdot (\pi(i)/\pi(\partialNode))$ by construction. Moreover, the expected number of remaining jobs satisfies $\expect{Z_i(\infty)} \leq \expect{Z_{\partialNode}(\infty)} \cdot (\pi(i) / \pi(\partialNode))$ by induction. The second term is thus upper bounded by 
\[
c_{\max} \left(\sum_{i \in \sub(\partialNode) \setminus \sub(\bo_{[m]})} \pi(i) / \pi(\partialNode)\right) \left(\expect{Z_{\partialNode}(\infty)} - Nz_{\partialNode}^{\bo}\right) = \tilde{O}(\sqrt{N}),
\]
which uses Lemma~\ref{lem:connect-partial} to bound $\expect{Z_{\partialNode}(\infty)} - N z_{\partialNode}^{\bo}$.
\item For the last term (never served), any $i \in \set{S} \setminus \sub(\bo_{[m]}) \setminus \sub(\partialNode)$ is an un-reduced state \ref{item:un-reduced} with $z_i^{\bo} = \lambda\pi(i).$ Since $\expect{Z_i(\infty)} \leq N\lambda \pi(i)$, which is the expected number of state-$i$ jobs assuming no service, the last summation is at most zero.
\end{itemize}
\end{proof}

The proof of Lemma~\ref{lem:stochastic} (Appendix~\ref{app:lem-stochastic}) combines Lemma~\ref{lem:iterative-ssc} for a suitable $\delta$ with Lemma~\ref{lem:ssc-cost-diff}. 
\section{Data-Driven Application to Content Moderation}\label{sec:numerics}
% !TEX root = main.tex
With a focus on content moderation, this section evaluates the operational benefit of designing a scheduling algorithm incorporating uncertain and evolving holding costs. The key takeaway from our simulations is that a practical implementation of our algorithm, called $\hindalg$, has the potential to improve the efficiency of typical heuristics employed by social media platforms. Section~\ref{sec:instantiate} instantiates our model to the human review system in content moderation. Section~\ref{sec:set-up} describes the set-up of our simulation and three heuristics employed by social media platforms. Section~\ref{sec:prac-imp} gives a practical implementation of $\alg$ (Algorithm~\ref{algo:pafou}) with hindsight approximation and machine learning. Section \ref{sec:video-real} evaluates the effectiveness of various algorithms using a semi-synthetic dataset with actual trajectories from YouTube videos and randomly generated probabilities of violation. Section~\ref{ssec:robustness} showcases the robustness of our algorithm.

\subsection{Instantiating our model to content moderation}\label{sec:instantiate}
Our model captures the scheduling component of a human review system in content moderation, which typically operates as an AI-human pipeline (see \cite{lykouris2024learning}). To evaluate our algorithm in an environment similar to a practical content moderation system, we consider the below setting. 

Over a fixed time horizon (e.g., one week), new pieces of content are enqueued into the human review system after user reporting or AI filtering. A content piece $j$ is characterized by two quantities: 
\begin{enumerate}
    \item[(i)] an indicator $\violating(j) \in \{0,1\}$ on whether the content violates the platform's policy and 
    \item[(ii)] a view function $\cumview(j,t)$ denoting the cumulative number of views this content will receive if left uninterrupted on the platform for $t$ time units after its creation (e.g., hours).
\end{enumerate} Both the violation indicator and the view function are unknown when a content piece is enqueued. Letting $\tau(j, T)$ be the time until a human review or the end of the horizon $T$ for content piece $j$, we consider the objective of minimizing \emph{policy-violating views} created by all enqueued content pieces: 
\begin{equation}\label{eq:def-vioviews}
\vioviews(T) = \sum_{\text{any enqueued content }j} \violating(j) \times \cumview(j, \tau(j, T)).
\end{equation}
Minimizing the number of policy-violating views represents a typical objective of social media content moderation systems. For example, Meta measures the ``prevalence of violating content'' (a normalized version of policy-violating views) to evaluate the effectiveness of their content moderation system \citep{meta_prevalence_metric,meta_cser}; YouTube \citep{youtube_vvr_2021,google_youtube_policy_views} and Snapchat \citep{snap_transparency_h1_2025} also use a similar metric under the name ``violative view rate'', which was noted by \cite{youtube_vvr_2021} as ``the best way for us to understand how harmful content impacts viewers, and to identify where we need to make improvements.''

We cannot directly optimize policy-violating views in our model because this requires the knowledge of whether a content piece is violating, which is unknown until a human review. In this section, we focus on an approximate metric, predicted policy-violating views, that replaces the actual violation of content piece $j$ with a probability of violation $p\violating(j)$:
\begin{equation}\label{eq:def-pvioviews}
p\vioviews(T) = \sum_{\text{any enqueued content }j} p\violating(j) \times \cumview(j, \tau(j, T)).
\end{equation}
The probability of violation in practice corresponds to outputs of machine learning models that predict whether a content piece violates certain platform policies. These models are widely applied in existing content moderation systems, see e.g. \cite{avadhanula2022} for Meta and \cite{Linkedin} for LinkedIn.

We consider the following instantiation of the model in Section~\ref{sec:model} for content moderation. The system operates in discrete periods. A job in this system corresponds to an enqueued content piece. The cost for content piece $j$ at time $t$, which is $c(S_j(t))$ in our model, is defined by $c(S_j(t)) = p\violating(j) \times \expect{\view(j,t) \mid S_j(t)}$ where $\view(j,t)$ is the number of views content piece $j$ gets in period~$t$. Here the state $S_j(t)$ captures any existing information for content piece $j$, including, for example, the predicted probability of policy-violating $p\violating(j)$ and the number of views in the last few periods. In each period, a random number of human reviewers become available to review content and a random number of new content pieces arrive into the system. Minimizing the cost defined in \eqref{eq:def-cost} is equivalent to minimizing a long-run average of the predicted policy-violating view metric in~\eqref{eq:def-pvioviews}. 

\noindent \textbf{Discussion on assumptions.}
There are three assumptions when we work with \eqref{eq:def-pvioviews} as an approximation of \eqref{eq:def-vioviews}. The first assumption is the independence between  the actual violation event of a content piece and its view trajectory, under which \eqref{eq:def-pvioviews} becomes an unbiased estimate of \eqref{eq:def-vioviews}. Although we do not test this assumption formally, it is consistent with the literature which uses data of non-policy-violating posts and randomly sampled posts to train a virality prediction model \citep{haimovich2021popularity}. The second and third assumptions are that the predictor $p\violating(j)$ is well-calibrated and is fixed across the lifetime of a content piece. These may not be true in practice, and we discuss in Section~\ref{ssec:robustness} on simulations that show the robustness of our algorithm against these assumptions.

We can avoid these three assumptions by directly optimizing the expectation of \eqref{eq:def-vioviews} using another instantiation of our model. In particular, we can set the instantaneous holding cost for a state $s$ to be the expected per-period policy-violating views \eqref{eq:def-vioviews} of a content piece, conditioning on this content piece having state $s$. We can then run our algorithm for this model instantiation to minimize the expected number of policy-violating views. We provide more details in Appendix~\ref{app:model-general}.

However, we choose not to base our algorithm design on this general model instantiation due to the limited amount of suitable data available in practice. To see this, the general instantiation requires a dataset with full (untruncated) view trajectories of policy-violating content pieces to reliably estimate the expected per-period policy-violating views. These policy-violating pieces are either (i) identified by the platform proactively or (ii) known to the platform much later through user reports or measurement effort \citep{meta_prevalence_metric}. In the case of (i), the platform is expected to remove the pieces swiftly, and thus their view trajectory is truncated. For (ii), though the platform can observe the full view trajectory, an efficient content moderation system should leave very few such pieces since each such piece denotes a failure point for the system. In practice, a platform may keep a  set of hold-out data for untruncated view trajectories of policy-violating content. However, it is costly to do so, and the size of suitable data is limited for the general instantiation.

The separation of $p\violating$ and the view trajectory in \eqref{eq:def-pvioviews} addresses the difficulty in data collection (though with added assumptions). To work with \eqref{eq:def-pvioviews}, the platform only needs two datasets: one containing actual violation of selected content pieces and the other one containing the view trajectories of another set of content pieces. Either dataset is easy to collect: for example, the first one can be obtained by offline manual labeling \citep{dataset-wiki}; and the second one is available through  the view trajectories of the massive amount of content pieces live on the platform \citep{haimovich2021popularity}. Therefore, for our algorithm to be practically relevant, we choose to focus on the formulation in~\eqref{eq:def-pvioviews}.

\subsection{Simulation set-up}\label{sec:set-up}
The simulation assumes an offline dataset $\set{D}$. For each content piece $j$ in the dataset we have access to:
(i) a predicted probability of violation $p\violating(j)$; (ii) ground-truth information of actual violation $\violating(j)$; (iii) ground-truth trajectory of number of views $\widetilde{\view}(j,d)$ denoting the number of views for content piece $j$ in the $d-$th period since its creation.\footnote{For a content piece $j$ that arrives in period $\tau$, its view in period $t$ is denoted by $\view(j,t)$ or $\widetilde{\view}(j,t-\tau)$.} Specifically, given a maximum length $L$ of content's view trajectory (which can be different from the time horizon $T$), the dataset is of the form 
\begin{equation}\label{eq:dataset}
\left\{p\violating(j), \violating(j), \left(\widetilde{\view}(j, d), d \leq L\right)\right\}_{j \in \set{D}}.
\end{equation}
We randomly and equally partition the dataset into a training set $\set{D}_{\text{train}}$ and a test set $\set{D}_{\text{test}}$ as our algorithm requires an offline training component. The predictions $p\violating$ are all calibrated such that the ground truth actual violation $\violating(j)$ is sampled from an independent Bernoulli random variable with mean $p\violating(j)$ for any content piece $j$ (Appendix~\ref{app:uncalibrated} considers the impact of uncalibrated predictions). The distribution of view trajectories is at the core of our simulation. In Section~\ref{sec:video-real}, we use a real-world dataset containing view trajectories of YouTube videos to evaluate the algorithms. To check the robustness of our algorithm in different business settings, we also conduct simulations on a synthetic dataset for user-generated content in Appendix~\ref{app:organic-synthetic} and ads content in Appendix~\ref{app:paid}. As an example, Figure~\ref{fig:example-trajectory} illustrates the uncertainty in view trajectories from these three settings via the five trajectories with the highest cumulative views in the corresponding training sets. 
\begin{figure}
\centering
\includegraphics[width=6in]{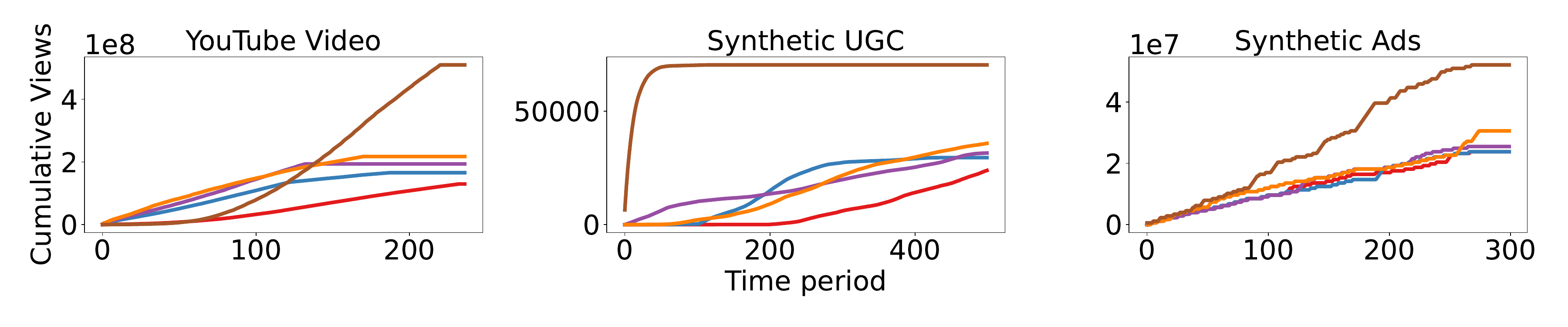}
\caption{View trajectories of five content pieces with highest cumulative views in the three datasets}
\label{fig:example-trajectory}
\end{figure}

\paragraph{Simulating a scheduling algorithm.} Given a scheduling algorithm $\nalg$, the simulation evaluates its performance via the policy-violating views metric (defined in \eqref{eq:def-vioviews}). Each run contains $T = 500$ periods where the events in Section~\ref{sec:model} happen in order for a period $t$: 
\begin{enumerate}
\item \emph{Server capacity:} a random number of reviewers $R(t)\sim \mathrm{Bin}(N,\mu)$ become available;
\item \emph{Service decision:} the algorithm selects a set of size $R(t)$ from the content enqueued for review;
\item \emph{Holding cost:} a non-reviewed content piece $j$ incurs policy-violating views $\violating(j) \times \widetilde{\view}(j, d(j,t))$ where $d(j,t)$ is the number of time periods since the arrival of content piece $j$;
\item \emph{Transition:} non-reviewed content pieces transition with $d(j,t+1) = d(j,t) + 1$;
\item \emph{Arrivals:} a number $A(t) \sim \mathrm{Bin}(N,\lambda)$ of new content pieces are sampled i.i.d. from $\set{D}_{\text{test}}$;
\item \emph{Queue update:} the queue consists of the union of non-reviewed and new content. If~$d(j,t+1) > L$ for a content piece $j$, the content leaves the queue because it will no longer incur any views. In our model, this means that the content transitions to the empty state.
\end{enumerate}

The simulation takes $N = 1000$ and $\lambda = 0.1$. The service rate $\mu$ varies by $\mu = \lambda \times \ratio$ with the $\ratio$ ranging from $1\%$ to $20\%$. This parameter, review ratio, captures the percentage of enqueued content that can be human reviewed on average. We restrict the review ratio up to $20\%$ to focus on an overloaded system motivated by the fact that significantly more more social media content is created than what humans can review. The policy-violating views of an algorithm are averaged over $20$ independent runs.

\paragraph{Baseline algorithms.} Each setting simulates a practical version of our $\alg$ algorithm (details in Section~\ref{sec:prac-imp}) and three heuristic algorithms. All heuristics assign indices to waiting jobs and review the jobs with the highest indices. The three index heuristics are: 
\begin{itemize}
\item $p\violating$ scheduling (\textsc{pViolating}): the index of content piece $j$ in period $t$ is $\ind_{\textsc{pV}}(j,t) = p\violating(j)$; this heuristic has been deployed for content moderation at LinkedIn \citep{Linkedin}. 
\item Velocity scheduling (\textsc{Velocity}): the index of content piece $j$ in period $t$ is $\ind_{\textsc{Vel}}(j, t) = p\violating(j) \times \view(j,t - 1),$ i.e., the probability of being policy-violating multiplied by the number of views in the \emph{last} period. This heuristic has been deployed at Meta \citep{avadhanula2022}.
\item pIV scheduling (\textsc{pIV}): the index of content piece $j$ for period $t$ is $\ind_{\textsc{pIV}}(j,t) = p\violating(j) \times \text{predicted remaining views of content piece $j$ from period $t$}$.
This heuristic was considered in \cite{makhijani2021}, which documented a simulator for Meta's content moderation system. Different from the above two heuristics, pIV requires an additional prediction model for future views. 
\end{itemize}

One can view $\textsc{Velocity}$ as a practical implementation of the generalized $c\mu$-rule \citep{van1995dynamic} as it prioritizes content with largest increase in holding cost, approximated by the views this content piece got in the last period. Similarly, $\textsc{pIV}$ follows the same idea of the $c\mu/\theta$-rule \citep{atar2010cmu} by prioritizing content with largest remaining holding cost. Both are canonical scheduling algorithms from the literature handling heterogeneous job holding costs (as discussed in Section~\ref{sec:intro}).

\subsection{Practical implementation of our algorithm}\label{sec:prac-imp}
The first step of implementing $\alg$ is identifying a  state definition of content piece capturing its view trajectory. We encode the state $S_j(t)$ for content piece $j$ in period $t$ by a five-dimensional vector: 
\begin{equation}\label{eq:features}
S_j(t) = \Big(d(j,t), \cumview(j, d(j,t)-1), (\widetilde{\view}(j, d(j,t) - k))_{ k \in\{1,2,3\}}\Big),
\end{equation}
corresponding to the number of periods since its creation, its total views before this period, and the per-period view in the last three periods. In practice, the state can include thousands of content features,  static (e.g., who created the content) or dynamic (e.g., number of likes and shares); see \cite{rizoiu2017expecting, haimovich2021popularity} and the references therein for the literature studying the prediction of content view trajectory.

With the state representation, the next step is to solve  for $\gamma^\star$ and the cost-to-go function $V^f(\gamma^\star, \cdot).$ Fixing $\gamma$, the cost-to-go function $V^f(\gamma,\cdot)$ is solvable via standard value-function based reinforcement learning (RL) algorithms such as deep Q-learning \citep{mnih2015human}. However, such an RL algorithm is difficult to train in content moderation given the billions of pieces of new content being created per day. Moreover, $\alg$ requires solving for $\gamma^\star$ on top of solving the cost-to-go function. 

We thus explore a \emph{hindsight approximation} version of $\alg$, which we call $\hindalg$ and is motivated by \cite{sinclair2023hindsight}. Recall from Section~\ref{sec:ski-rental} that the cost-to-go function $V^f(\gamma,i)$ represents the minimum expected future cost of a Markovian ski-rental problem, conditioning on an initial state $S_1 = i$. In its definition \eqref{eq:bellman}, the future cost trajectory is unknown through the uncertainty in future states $S_2,\ldots,S_L$ where we recall from Section~\ref{sec:model} that the single-job Markov chain is a tree with at most $L$ levels. Now suppose instead that we do know the future state trajectory is $S_k = s_k, k \leq L$ and thus the future cost trajectory is $c(s_2),\ldots,c(s_L).$ The hindsight minimum cost is equal to $\min(\gamma, c(s_2)+\ldots+c(s_L))$ since the optimal action is to buy at a cost of $\gamma$ if the total future renting cost is higher than this buying cost; otherwise, the optimal action is to always rent. A hindsight cost-to-go function is the expectation of the hindsight minimum cost, defined by 
\[
\tilde{V}(\gamma, i) = \expectsub{S_2,\ldots,S_L}{\min\left(\sum_{\ell=2}^L c(S_{\ell}), \gamma\right) \mid S_1 = i, S_{k+1} \sim P(S_{k}, \cdot), \forall k}.
\]

Recall that when we instantiate our model for content moderation the cost $c(S_j(t))$ for content piece $j$ in period $t$ corresponds to $p\violating(j) \times \widetilde{\view}(j, d(j, t))$. Putting aside $p\violating(j)$ and assuming the cost for a content piece in a period is equal to its number of views in that period, the above hindsight cost-to-go function $\tilde{V}$ leads to a simple offline training procedure. In particular, given the training set $\set{D}_{\text{train}}$, for content piece $j$ let $\hat{s}(j, \tau)$ be the state representation of the view trajectory of content piece $j$ in the $\tau-$th period and $\text{futureView}(j,\tau)$ be its future views after the $\tau-$th period, i.e., $\text{futureView}(j,\tau) = \sum_{k = \tau + 1}^L \widetilde{\view}(j, k).$ We train a regression model for any given $\gamma$ by
$\set{M}_\gamma = \text{Regression}\left(\hat{s}(j,\tau) \to \min\left\{\gamma, \text{futureView}(j,\tau)\right\} \colon j \in \set{D}_{\text{train}}, 1 \leq \tau \leq L\right)$,
i.e., given the current state $s$ of a content piece, the regressor outputs $\set{M}_\gamma(s)$, which predicts the expectation of the minimum between $\gamma$ and the future number of views. In our implementation, we fit an XGBoost regressor \citep{chen2016xgboost} with maximum depth equal to $5$ and number of estimators equal to~$100$. 

When $\gamma$ is set sufficiently large, $\set{M}_\gamma$ predicts the future views of a content piece, which is what we use for $\textsc{pIV}$ scheduling in Section~\ref{sec:set-up}. For Algorithm~\ref{algo:pafou}, we need to set a suitable hyper-parameter $\gamma^\star$.  In our simulation, we take $\gamma^\star$ to be the $99$th percentile of content pieces' total views in the training set $\set{D}_{\text{train}}.$

Summarizing $\hindalg$, it takes as input a hyper-parameter $\gamma^\star$ and trains a regression model $\set{M}_{\gamma^\star}$. Then similar to \eqref{eq:index-alg}, the index for content piece $j$ in period $t$ is given by 
\begin{equation*}
\ind_{\hindalg}(j,t) = p\violating(j) \times \left[\view(j, t - 1) + \set{M}_{\gamma^\star}(\hat{s}(j,d(j,t)))\right].
\end{equation*}
When $\gamma^\star = 0$, the index becomes the same as that of $\textsc{Velocity}$. Alternatively, if $\gamma^\star = \infty$, the index behaves like $\textsc{pIV}$ though it also considers the views in the last period. Tuning $\gamma^\star$ seeks a balance between the \emph{accurate} instantaneous view information and the \emph{uncertain} future view information. 

\subsection{Simulations on YouTube video data}\label{sec:video-real}
We use the \textsc{Active} dataset from \cite{rizoiu2017expecting}, which contains the daily number of views of certain YouTube videos, to evaluate the algorithms. 

\paragraph{Data description.} We provide a general discussion of the \textsc{Active} dataset, and refer the readers to \cite{rizoiu2017expecting} for the detailed data collection process. The dataset contains $14,041$ YouTube videos. Each video comes with a vector of daily number of views, the daily number of times it is tweeted or shared, and video features such as the category and duration of this video. For the purpose of our analysis, we only use the daily number of views. The dataset tracks the daily views for a minimum of $119$ days and a maximum of $237$ days among these videos. The maximum number of total views is $509,245,784$ and the minimum is $18$. See Figure~\ref{fig:example-trajectory} for examples of view trajectories of these videos. 

We generate a dataset of the form \eqref{eq:dataset} for our simulation as follows. The dataset contains all videos in the \textsc{Active} dataset and uses the daily number of views as $\widetilde{\view}(j,d)$ for a video $j$ and a period $d \leq L = 237.$ We fill the views by zero if the original dataset does not include enough days of views for a video. The probability of policy-violating $p\violating(j)$ is a uniform random variable in $[0,1].$ We equally and randomly separate the dataset into a training and a test set. For the simulation, we take $\gamma^\star$ equal to $9,885,761$ ($99$th percentile of content pieces' total views in the training set). For a prediction target $\min(\gamma^{\star},\text{futureView})$, the trained XGBoost model obtains a root mean square error (RMSE) of $370,959$ in the training set and a RMSE of $416,019$ in the test set. When the prediction target is uncapped (corresponding to $\textsc{pIV}$), the trained XGBoost model has a RMSE of $4,019,970$ for the training set and a RMSE of $1,867,245$ for the test set (the RMSE is higher in the training set because there is one extremely viral video). This highlights the benefit of capping the future views for regression, which makes it easier to train an accurate prediction model.

\paragraph{Simulation results.} Recall that the review ratio $r = \mu / \lambda$ captures the average fraction of enqueued content pieces that can be reviewed by human reviewers. Varying the review ratio $r \in \set{R} = \{0.01 + 0.002k\colon 0 \leq k < 95\}$, Figure~\ref{fig:youtube-calibrated} shows the effectiveness of $\hindalg$ compared to the other heuristics considered in Section~\ref{sec:set-up}. Letting $\vioviews(\nalg, r)$ be the number of policy-violating views for algorithm $\nalg$ when the review ratio is $r$, Figure~\ref{fig:youtube-calibrated}(a) plots $ \vioviews(\nalg, r)$ for $\nalg \in \{\textsc{pViolating}, \textsc{Velocity}, \textsc{pIV}, \hindalg\}$ and different review ratio $r$. The plot shows that $\hindalg$ clearly outperforms other algorithms with at least $3.6\%$ reduction in policy-violating views compared to the status-quo approaches ($\textsc{pIV}$ and $\textsc{Velocity}$).

As our dataset contains the full view trajectories, our simulation also explores a full-information benchmark, $\textsc{Oracle}$, that operates like $\hindalg$ but prioritizes based on the actual future views (capped by $\gamma^\star$). The figure shows a wide gap between $\textsc{Oracle}$ and $\hindalg$. This gap is due to both (1) the fact that XGBoost does not necessarily minimize the prediction error on the test set and (2) the unavoidable uncertainty in a post's future view trajectory (that is, even a \emph{perfect} prediction model will incur non-zero prediction error).

\begin{figure}[hbtp]
\centering
\includegraphics[width=\textwidth]{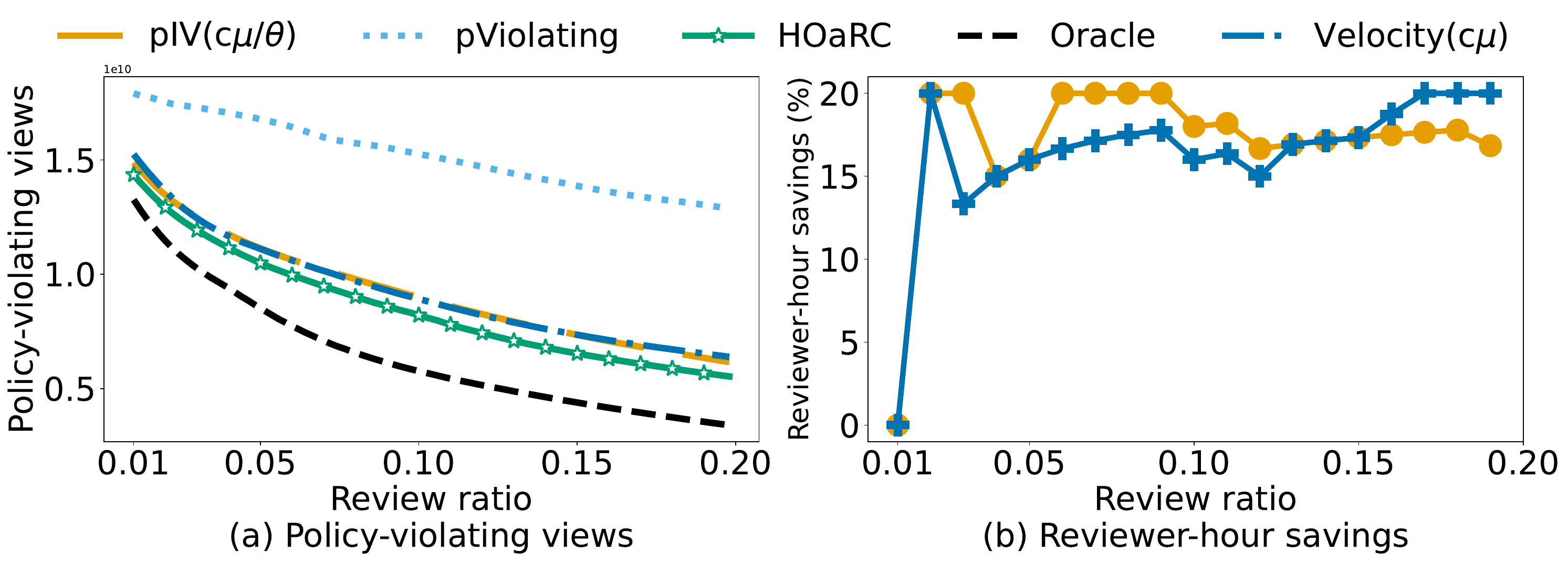}
\caption{YouTube data: policy-violating views of algorithms and reviewer-hour savings by $\hindalg$}
\label{fig:youtube-calibrated}
\end{figure}

To highlight how such prevalence reduction effort translates into reviewer-hour savings for a platform, Figure~\ref{fig:youtube-calibrated}(b) shows the percentage of reviewer-hour savings by using $\hindalg$ (we omit the comparison with $\textsc{pViolating}$ to focus on the better benchmark algorithms). That is, for an algorithm $\nalg$ and a review ratio $r$, we find the minimum review ratio $r'$ from the list $\set{R}$ such that $\vioviews(\hindalg, r') \leq \vioviews(\nalg, r)$ and then show $1 - \frac{r'}{r}$. This quantity captures the reduction in capacity (thus capturing reviewer-hour savings) to achieve the same amount of policy-violating views when we switch from an algorithm $\nalg$ to $\hindalg$. The plot shows that $\hindalg$ achieves $11\%$ to $20\%$ reviewer-hour savings when compared to $\textsc{pIV}$ or $\textsc{Velocity}$. 

\subsection{Robustness checks} \label{ssec:robustness}
Appendix~\ref{app:simulations} includes additional simulations testing the robustness of our performance gain.

\noindent\textbf{1. Robustness of insight on typical view patterns}. Besides the YouTube dataset in Section~\ref{sec:video-real}, we also conduct the simulations on two synthetic datasets that aim to mimic view patterns of user-generated content (Appendix~\ref{app:organic-synthetic}) and online ads (Appendix~\ref{app:paid}). Simulation results on these datasets are consistent with the results on the YouTube dataset (though the improvement is weaker for the online ads setting) and show that $\hindalg$ generally leads to fewer policy-violating views than the benchmarks.

\noindent \textbf{2. Robustness to uncalibrated predictions.} Our model instantiation in Section~\ref{sec:instantiate} assumes that the platform has well-calibrated predictions $p\violating(j)$ on the probability of violation for each content piece, which may not be true in practice. Appendix~\ref{app:uncalibrated} studies the impact of uncalibrated predictions by randomly perturbing $p\violating(j)$ and finds that the performance gain of our algorithm is robust to calibration errors.

\noindent \textbf{3. Robustness to evolution in the probability of violation.} In Appendix~\ref{app:pvio-update}, we study a setting with evolving $p\violating(j)$. For each content piece, the platform has a prior distribution on its actual violation and dynamically updates this distribution with new information. We find that if the initial prior is accurate, then the original version of $\hindalg$ (without incorporating the evolution) retains favorable performance. With an inaccurate prior, incorporating the actual evolution in $p\violating(j)$ into the design of $\hindalg$ (at the expense of a more costly dataset as discussed above) is fruitful: the resulting algorithm performs much better than other algorithms that do not truly capture the evolution in the probability of violation.

\noindent \textbf{
    4. Robustness to joint bandit learning.} Recently, \cite{avadhanula2022,lykouris2024learning} proposed bandit-learning-based approaches for better predicting $p\violating$. In Appendix~\ref{app:bandit-learning}, we conduct simulations to study how scheduling algorithms interact with such bandit learning. The simulation compares six algorithmic combinations (two $p\violating$ estimation  approaches, offline learning and bandit learning, and three scheduling approaches $\{\textsc{Velocity}, \textsc{pIV}, \hindalg\}$). Consistent with prior work, we find that bandit learning on $p\violating$ generally leads to fewer policy-violating views than using offline learning. Moreover, combining bandit learning with our proposed scheduling algorithm $\hindalg$ consistently gives the fewest policy-violating views, showing that our scheduling innovation is compatible with prior learning-based approaches.
\section{Conclusion}\label{sec:conclusion}
% !TEX root = main.tex

This paper studies the problem of scheduling in queueing systems with uncertain and evolving holding costs, which arises in  content moderation for social media platforms. Unlike existing approaches that  optimize either based on the instantaneous or the expected remaining holding cost, we design an algorithm that adjusts to the opportunity of serving a job in the future once its uncertainty partly resolves. Our algorithm is asymptotically optimal and outperforms existing approaches. On the analytical front, our result relies on  two key contributions: (a) a complete characterization of the fluid equilibrium for any priority algorithm and (b) a new tail bound that separates high-probability and almost-sure upper bounds on the drift of a Lyapunov function.

Our work opens up several interesting directions to further improve the operational efficiency of content moderation in social media platforms. First, our model assumes access to the Markov chain of jobs (that depends on the full view trajectory of jobs). In practice, the future view trajectory of jobs that are removed from the platform is unobservable. Designing a scheduling algorithm that can sufficiently learn the view trajectories despite this censored feedback is an intriguing open direction. Second, our algorithm assumes perfect knowledge of the underlying Markov chain. On the other hand, the instantaneous-cost based algorithmic principle (the $c\mu-$rule) does not require such information. Designing an index-based algorithm that achieves the Pareto frontier between these two scenarios and adjusts to the potential estimation error is an interesting open question.   Third, our model and algorithm require stationary arrival patterns and review capacity. The optimal capacity dual $\gamma^\star$ depends on both the arrival rate $\lambda$ and the service rate $\mu$. However, in practice, the availability of human reviewers may fluctuate over time \citep{makhijani2021} and the platform can have time-varying influx of content. Designing a scheduling algorithm that handles non-stationarity is a theoretically challenging and practically important open question.  

\subsection*{Acknowledgements}
This work was supported by an MIT-Meta Sponsored Research Agreement. 

\bibliographystyle{alpha}
\bibliography{references}

\newpage
\appendix

\section{Tables of Notation}\label{app:table}
% !TEX root = intro.tex
This section summarizes key notation used in the paper. It contains three tables of notation, classified by whether the notation is mainly for jobs' tree-shaped Markov chain (Table \ref{table:tree}), system dynamics (Table \ref{table:dynamics}), or the analysis (Table \ref{table:analysis}).

\begin{table}[!hbp]
\centering
\caption{Notation for jobs' tree-shaped Markov chain}
\label{table:tree}
\begin{tabular}{|c|c|}
\hline
Symbol                                                        & Meaning                                                                                                                                                         \\ \hline
$\set{S}$                                                     & state space                                                                                                                                                     \\ \hline
$L$                                                           & \begin{tabular}[c]{@{}c@{}}number of levels of the tree\\ (maximum length of a job's cost trajectory)\end{tabular}                                              \\ \hline
$\theta$                                                      & minimum abandonment probability                                                                                                                                 \\ \hline
$\set{S}_0,\ldots,\set{S}_{\ell},\ldots,\set{S}_{L-1}$                              & states on level $\ell$ of the tree                                                                                                              \\ \hline
$P(i,k)$                                                      & the probability of transitioning from state $i$ to state $k$                                                                                                    \\ \hline
$p(i)$                    & the probability of transitioning from its parent to state $i$                                                                     \\ \hline
$c(i)$                                                        & the instantaneous holding cost for a job with state $i$                                                                                                         \\ \hline
$\rootNode$                                                   & the root of the tree                                                                                                                                            \\ \hline
$\pa(i)$                                                      & the parent of state $i$                                                                                                                                         \\ \hline
$\child(i)$                                                   & the set of child(ren) of state $i$                                                                                                                              \\ \hline
$\anc(i)$                                                     & the set of ancestors of state $i$                                                                                                                               \\ \hline
$\sub(i)$                                                     & the set of states in the subtree of state $i$                                                                                                                   \\ \hline
$\pa(\set{X}), \child(\set{X}), \anc(\set{X}), \sub(\set{X})$ & \begin{tabular}[c]{@{}c@{}}the union of parents, children, ancestors or subtrees\\ over all states in $\set{X}$\end{tabular}                                    \\ \hline
$\TopS(\set{X})$                                               & \begin{tabular}[c]{@{}c@{}}the top set of states in $\set{X}$, which is the minimum subset of $\set{X}$\\ whose subtree is a superset of $\set{X}$\end{tabular} \\ \hline
\end{tabular}
\end{table}

\begin{table}[!hbp]
\centering
\caption{Notation for system dynamics}
\label{table:dynamics}
\begin{tabular}{|c|c|}
\hline
Symbol                                    & Meaning                                                                                                                                            \\ \hline
$N$                                       & the system size which scales to $\infty$                                                                                                           \\ \hline
$\lambda, \mu$                            & the (normalized) arrival and service rates                                                                                                         \\ \hline
$R(t)$                                    & the random number of available servers in period $t$                                                                                               \\ \hline
$\bolds{Q}(t) = (Q_i(t))_{i \in \set{S}}$ & the number of state-$i$ jobs at the beginning of period $t$                                                                                        \\ \hline
$\bolds{R}(t) = (R_i(t))_{i \in \set{S}}$ & the number of served state-$i$ jobs in period $t$                                                                                                  \\ \hline
$\bolds{Z}(t) = (Z_i(t))_{i \in \set{S}}$ & \begin{tabular}[c]{@{}c@{}}the number of remaining state-$i$ jobs after service in period $t$, \\ which is equal to $Q_i(t) - R_i(t)$\end{tabular} \\ \hline
$C(\set{A}), C(N,\set{A})$                & \begin{tabular}[c]{@{}c@{}}the long-run average holding cost of an algorithm $\set{A}$ \\ (when the system size is $N$)\end{tabular}               \\ \hline
$\subopt(\set{A}), \subopt(N,\set{A})$          & the suboptimality gap of algorithm $\set{A}$ (when the system size is $N$)                                                                                 \\ \hline
\end{tabular}
\end{table}

\begin{table}[]
\centering
\caption{Notation for analysis}
\label{table:analysis}
\begin{tabular}{|c|c|}
\hline
Symbol                                & Meaning                                                                                                                                                                                                             \\ \hline
$\bolds{q}=(q_i)_{i \in \set{S}}$     & fluid-scaled queue length for state $i$                                                                                                                                                                             \\ \hline
$\bolds{\nu}=(\nu_i)_{i \in \set{S}}$ & fluid-scaled service for state $i$                                                                                                                                                                                  \\ \hline
$C^\star$                             & the minimum cost in the fluid LP \eqref{eq:orifluid}                                                                                                                                                                \\ \hline
$P^\star$                             & the maximum prevented cost, which is $\lambda c^f(\rootNode) - C^\star$                                                                                                                                             \\ \hline
$\pi(i)$                              & \begin{tabular}[c]{@{}c@{}}the probability of a job becoming state $i$, \\ which is $\prod_{a \in \anc(i)} p(a)$\end{tabular}                                                                                \\ \hline
$\pi(i) / \pi(a)$                     & the probability of a state-$a$ job to become a state-$i$ job                                                                                                                                                        \\ \hline
$c^f(a)$                              & the conditional future cost for state $a$, $\sum_{i \in \sub(a)} c(i)\pi(i)/\pi(a)$                                                                                                                                 \\ \hline
$\gamma$                              & capacity dual / price                                                                                                                                                                                               \\ \hline
$\bolds{\beta}^\star(\gamma)$         & the optimal state duals when the capacity dual is $\gamma$                                                                                                                                                          \\ \hline
$V^f(\gamma,i)$                       & the future cost-to-go function conditioning on price $\gamma$ and state $i$                                                                                                                                         \\ \hline
$\bo = (o_1,\ldots,o_{|\set{S}|})$    & a priority ordering denoting the order in which states are served                                                                                                                                                   \\ \hline
$\bo_{[h]}$                           & the first $h$ states, $\{o_1,\ldots,o_h\}$ in the priority ordering $\bo$                                                                                                                                           \\ \hline
$m$                                   & \begin{tabular}[c]{@{}c@{}}the maximum position with $\sum_{i \in \TopS(\bo_{[m]})} \lambda \pi(i) \leq \mu$,\\ i.e., the service rate can serve (in expectation) all jobs in the top set of $\bo_{[m]}$\end{tabular} \\ \hline
$\partialNode$                        & \begin{tabular}[c]{@{}c@{}}the partially-served state which is $o_{m+1}$ \\ unless the above inequality is equality or $m = |\set{S}|$\end{tabular}                                                              \\ \hline
$\kappa$                              & the degeneracy parameter defined in \eqref{eq:def-degeneracy}                                                                                                                                                       \\ \hline
$\bolds{q}^{\bo}, \bolds{\nu}^{\bo}$  & fluid equilibrium for priority algorithm based on ordering $\bo$                                                                                                                                                    \\ \hline
$\bolds{q}^{\alg},\bolds{\nu}^{\alg}$ & fluid equilibrium for $\alg$                                                                                                                                                                                        \\ \hline
\end{tabular}
\end{table}

\newpage

\section{Detailed Comparison with Existing Results in Restless Bandits}\label{app:comp}
% !TEX root = main.tex
Our setting can be viewed as  restless bandits with dynamic populations (RBDP), which was studied by \cite{verloop2016asymptotically,zayas2019asymptotically, fu2022optimal}. Among these works, \cite{verloop2016asymptotically} studies the most general model of which ours is a special case. To highlight our technical contribution, we compare our results with those in \cite{verloop2016asymptotically}, which studies asymptotically optimal algorithms for RBDP.

In the model of \cite{verloop2016asymptotically}, there is a set of arms each with a state. At any time, each arm can be made passive or active, but the number of active arms cannot exceed a fixed budget $\alpha$. Arms can be of different types and each type has two Markov chains dictating the transitions of an arm's state under the passive or the active action. New arms may arrive at type-dependent arrival rates. With a state-dependent holding cost rate, the goal is to find a policy that minimizes the long-run average holding costs among all arms. Our model in Section~\ref{sec:model} is then a special case of this model (if one ignores the subtle difference between continuous-time and discrete-time settings): a job corresponds to an arm, its state transitions according to the Markov chain if the system does not serve the job (i.e., it is passive) and transitions to $\perp$ if the system serves it (i.e., it is active).

Focusing on a scaling regime where the arrival rates scale with the budget $\alpha \to \infty$, the analytical framework of \cite{verloop2016asymptotically} is as follows. With the optimal fluid LP solution (similar to ours in \eqref{eq:orifluid}), \cite{verloop2016asymptotically} defines a set of priority policies $\Pi^\star$ such that any policy in this set satisfies 
\begin{itemize}
\item states that are always active are given priority over states that are sometimes passive; 
\item states that are sometimes active and sometimes passive are given priority over states that are always passive;
\item if the capacity constraint is not tight in the LP solution, then the policy never activates states that are always passive in the LP solution.
\end{itemize}
Indeed, specializing to our settings, one can verify that $\alg$ is a policy in $\Pi^\star$.

\cite{verloop2016asymptotically} then shows asymptotic optimality for policies in $\Pi^\star$ in three steps: First, the optimal LP solution must be a lower bound for the (scaled) long-run average holding cost of any policy (lemma 4.3 of \cite{verloop2016asymptotically}), which is akin to our Lemma~\ref{lem:lower-bound}. Second,  the optimal LP solution must be an equilibrium point of the fluid system for any policy in $\Pi^\star$ (lemma 4.9 of \cite{verloop2016asymptotically}), as what we show in Lemma~\ref{lem:optimal-fluid} for $\alg$. Third, the (scaled) steady-state queue length vector converges to the optimal LP solution (proposition 4.14 of \cite{verloop2016asymptotically}).

The key distinction between the proof of \cite{verloop2016asymptotically} and ours is the assumptions needed for the  step establishing convergence. To show such convergence, \cite{verloop2016asymptotically} assumes two crucial conditions. The first condition requires the tightness and uniform integrability of the sequence of steady-state queue length distributions when $\alpha \to \infty$. This condition is easy to satisfy when there is abandonment, i.e., an arm being kept passive must eventually leave the system (proposition 4.13 of \cite{verloop2016asymptotically}). The second condition is the canonical Global Attractor Property (GAP) \citep{weber1990index, verloop2016asymptotically}. This property means that the fluid dynamics under a given policy in $\Pi^\star$ must converge to the LP solution, which is an equilibrium, regardless of the initial condition. However, as commented by \cite{verloop2016asymptotically}, despite the crucial role of the GAP property in the proof, there are ``in general no sufficient conditions'' and this property is ``often verified only numerically''. Furthermore, \cite{verloop2016asymptotically} views that ``finding sufficient conditions under which the global attractor property holds for policies in $\Pi^\star$ is therefore important on its own''.

Compared to \cite{verloop2016asymptotically}, our main analytical contribution is   proving asymptotic optimality without the GAP assumption for a broad class of RBDP capturing queueing systems with uncertain holding costs. In particular, beyond showing asymptotic optimality, our result in Lemma~\ref{lem:stochastic} provides a non-asymptotic $\tilde{O}(\sqrt{N})$ bound between the long-run average holding cost of the stochastic system and the cost of the fluid equilibrium. We are able to achieve this result without the GAP assumption by the Lyapunov drift argument in Section~\ref{sec:stochastic}, which is tailored to the special water-filling structure in Section~\ref{sec:fluid-sol} provided by our queueing models.

\section{Extension to General Markov Chain (Remark~\ref{remark:general-momdel})}\label{app:general-markov}

% !TEX root = main.tex
This section provides further detail to Remark~\ref{remark:general-momdel} on how our results can accommodate job state transitions with a general discrete-time Markov chain (DTMC). Suppose that job states  transition according to a DTMC with a finite state space $\set{S}$ and a transition kernel $P=(P(i,k))_{i,k\in \set{S}}$. The job has a fixed initial state $\rootNode$ and has a non-zero abandonment probability $\theta$, i.e., $1-\sum_{k \in \set{S}} P(i,k) \geq \theta$ for any state $i$. Denote this problem instance $\set{I}$.

We construct another problem instance $\widetilde{\set{I}}$ with the same model primitives as $\set{I}$ but a different Markov chain. In $\widetilde{\set{I}}$, job states evolve according to a finite-depth tree-shaped Markov chain. The tree has a depth $L = \lceil -\ln_{1-\theta} N \rceil$ and a state space $\{\rootNode\}\cup (\{\rootNode\} \times \set{S}) \cup \cdots \cup (\{\rootNode\} \times \set{S} \times \overset{\text{$(L-1)$ times}}{\cdots}\times \set{S})$. Here $\times$ is the Cartesian product such that for two sets $\set{A}, \set{B}$, their product $\set{A} \times \set{B}$ is $\{(a,b):a \in \set{A}, b \in \set{B}\}$. The constructed state of a job  corresponds to the full history of actual states that the job has been in since it arrives to the system. If a job has state $(s_1,s_2,\ldots,s_{\ell})$, it  transitions to state $(s_1,s_2,\ldots,s_{\ell}, s_{\ell+1})$ with probability $P(s_{\ell},s_{\ell+1})$ unless $\ell = L$, in which case the job transitions to $\perp$ deterministically. The two problem instances $\set{I}$ and $\widetilde{\set{I}}$ are coupled in the sense that we can simulate a sample path in $\widetilde{\set{I}}$ with a sample path in $\set{I}$ by ignoring all jobs that have stayed for more than $L$ periods. 

Suppose that we have an algorithm $\widetilde{\nalg}$ for $\widetilde{\set{I}}$. Denote its long-run average holding cost in $\widetilde{\set{I}}$ by $C(\widetilde{\set{I}}, \widetilde{\nalg})$ and its suboptimality gap by $\subopt(\widetilde{\set{I}}, \widetilde{\nalg})$. To run this algorithm in $\set{I}$, in each period we ignore all jobs that have stayed for more than $L$ periods and let $\widetilde{\nalg}$ run as if it is in $\widetilde{\set{I}}$. Let this algorithm be $\nalg$. Its long-run average holding cost in $\set{I}$, $C(\set{I}, \nalg)$, is no more than the long-run average holding cost of $\widetilde{\nalg}$ in $\widetilde{\set{I}}$,  $C(\widetilde{\set{I}},\widetilde{\nalg})$, except for the jobs that we ignore in $\widetilde{\set{I}}$, which are those that have stayed for more than $L$ periods. Since the probability that a job arriving in period $\tau$ stays until period $t$ is at most $(1-\theta)^{t-\tau}$, we  upper bound $C(\set{I}, \nalg)$ by
\begin{align*}
C(\set{I}, \nalg) &\leq C(\widetilde{\set{I}}, \widetilde{\nalg}) + \lim\sup_{T \to \infty}\frac{1}{T}\sum_{t=L}^T \sum_{\tau=1}^{t-L} N\lambda c_{\max} (1-\theta)^{t - \tau} \\
&\leq C(\widetilde{\set{I}}, \widetilde{\nalg}) + \frac{Nc_{\max}(1-\theta)^L}{\theta} \leq C(\widetilde{\set{I}}, \widetilde{\nalg}) + \frac{c_{\max}}{\theta}.
\end{align*}
Moreover, for any algorithm $\nalg'$ for instance $\set{I}$, we can run it on a sample path of instance $\widetilde{\set{I}}$ by keeping virtual copies of jobs that abandon because they have stayed for more than $L$ periods and then following the service decision of $\nalg'$. This leads to an algorithm $\widetilde{\nalg'}$ for $\widetilde{\set{I}}$. By the coupling, the long-run average holding cost of $\nalg'$ in $\set{I}$ is at least that of $\widetilde{\nalg'}$ in $\widetilde{\set{I}}$, i.e., $C(\set{I}, \nalg') \geq C(\widetilde{\set{I}}, \widetilde{\nalg'}).$ As a result, 
\[
C(\set{I},\nalg) - C(\set{I},\nalg') \leq C(\widetilde{\set{I}}, \widetilde{\nalg}) + \frac{c_{\max}}{\theta} - C(\widetilde{\set{I}}, \widetilde{\nalg'}) \leq \subopt(\widetilde{\set{I}}, \widetilde{\nalg}) + \frac{c_{\max}}{\theta}.
\]
Since $\nalg'$ is chosen arbitrarily, the above inequality implies that the suboptimality gap of the constructed algorithm $\nalg$ for instance $\set{I}$ is at most the suboptimality gap of $\widetilde{\nalg}$ for instance $\widetilde{\set{I}}$ plus an additional term $\frac{c_{\max}}{\theta}$. Taking $\widetilde{\nalg}$ to be $\alg$, Theorem~\ref{thm:pafou} shows that its suboptimality gap is $\tilde{O}(\sqrt{N})$ for $\widetilde{\set{I}}$. Therefore, the constructed algorithm $\nalg$ enjoys a suboptimality gap of $\tilde{O}(\sqrt{N}) + \frac{c_{\max}}{\theta} = \tilde{O}(\sqrt{N})$ for the instance $\set{I}$ with the general Markov chain.

\section{Additional Proofs}\label{app:analysis}
% !TEX root = main.tex
\subsection{Proof of Lemma~\ref{lem:dual-structure} (Section~\ref{sec:fluid-lp})}\label{app:lem-dual-structure}
Fix a capacity dual $\gamma \geq 0$. Recall that $\bolds{\beta}^\star(\gamma) = (\beta^\star_i(\gamma))_{i \in \set{S}}$ with $\beta^\star_i(\gamma) = \max(0, c(i) + V^f(\gamma,i) - \gamma).$  The following result establishes a connection between the cost-to-go function and $\bolds{\beta}^\star(\gamma).$
\begin{lemma}\label{lem:sub-tree-beta}
For any $a \in \set{S}$, the cost-to-go function $V(\gamma,a)= c^f(a) - \sum_{i \in \sub(a)} \beta^\star_i(\gamma)\pi(i) / \pi(a)$.
\end{lemma}
\begin{proof}
For ease of notation, we omit the dependence on $\gamma$ for $\bolds{\beta}^\star(\gamma)$ and use $\bolds{\beta}^\star = \bolds{\beta}^\star(\gamma)$ with $\beta^\star_i = \beta^\star_i(\gamma).$  Denoting the right-hand-side of the lemma by
$$\tilde{V}(a) = c^f(a) -\sum_{i \in \sub(a)} \beta^\star_i\pi(i) / \pi(a) = \sum_{i \in \sub(a)} (c(i) - \beta^\star_i)\pi(i) / \pi(a),$$ we show by induction that $\tilde{V}(a) = V(\gamma,a)$ defined in \eqref{eq:bellman}. Recall that states in $\set{S}$ are partitioned into states $\set{S}_0,\ldots,\set{S}_{L-1}$ in each level. Our induction moves from the largest level $L-1$ to the lowest level $0$. First, for states $a \in \set{S}_{L-1}$ in the highest level,  $\sub(a) = \{a\}.$ Then $\tilde{V}(a) = c(a) - \beta_a^\star = \min(c(a),\gamma) = V(\gamma,a)$ by \eqref{eq:bellman}. As our inductive hypothesis, suppose for some $0 \leq \ell < L-1$, we have shown $\tilde{V}(i) = V(\gamma,i)$ for any $i \in \set{S}_{\ell+1}$; our inductive step now establishes this for $a \in \set{S}_{\ell}$. 
\begin{align*}
\tilde{V}(a) = \sum_{i \in \sub(a)} (c(i) - \beta^\star_i)\pi(i) / \pi(a) &\overset{(i)}{=} c(a)-\beta^\star_a + \sum_{k \in \child(a)} \sum_{i \in \sub(k)} (c(i)-\beta_i^\star)\pi(i) / \pi(k) \cdot P(a,k) \\
&\overset{(ii)}{=} c(a) - \beta^\star_a + \sum_{k \in \child(a)} P(a,k)V(\gamma,k) \overset{(iii)}{=} c(a) - \beta^\star_a + V^f(\gamma,a),
\end{align*}
where (i) is because $\pi(k) = \pi(a)P(a,k)$ for $k \in \child(a)$; (ii) is because of the induction hypothesis and $k \in \set{S}_{\ell+1}$; (iii) is by the definition of $V^f(\gamma,a)$ in \eqref{eq:bellman}. Using $\beta_a^\star = \max(0, c(a)+V^f(\gamma,a) - \gamma)$ gives $\tilde{V}(a) = \min(\gamma, c(a)+V^f(\gamma,a)) = V(\gamma,a),$ finishing the induction.
\end{proof}

\begin{proof}[Proof of Lemma~\ref{lem:dual-structure}]
Define the function $f(\bolds{\beta}) = \sum_{i \in \set{S}} \pi(i)\beta_i$ for a state dual vector $\bolds{\beta} = (\beta_i)_{i \in \set{S}},$ which is the first part of $D^\star(\gamma)$ in \eqref{eq:dual}. Define the feasible set of state duals for $D^\star(\gamma)$ by
\[
\set{B}(\gamma) = \left\{\bolds{\beta} \in \set{R}_{\geq 0}^{\set{S}}\colon c^f(a) - \sum_{i \in \sub(a)} \beta_i \pi(i) / \pi(a) \leq \gamma, \forall a \in \set{S}\right\}.
\]
Then $D^\star(\gamma) = \mu\cdot \gamma + \lambda\min_{\bolds{\beta} \in \set{B}(\gamma)} f(\bolds{\beta})$ by \eqref{eq:dual}. The problem $\min_{\bolds{\beta} \in \set{B}(\gamma)} f(\bolds{\beta})$ is an LP whose dual is given by
\begin{equation}
\label{eq:dual-D}
\begin{aligned}
 & \max_{\bolds{y} \in \mathbb{R}_{\geq 0}^{\set{S}}}\sum_{a \in \set{S}} (c^f(a)-\gamma)y_a \\
&\text{s.t.}  \sum_{a \in \anc(i)} y_a \cdot \frac{\pi(i)}{\pi(a)} \leq \pi(i),~\forall i.
\end{aligned}
\end{equation}

We show the first result of the lemma (that $\bolds{\beta}^\star(\gamma)$ is an optimal solution to $D^\star(\gamma)$) by proving two properties: (1) it is feasible, i.e., $\bolds{\beta}^\star(\gamma) \in \set{B}(\gamma)$; (2) there exists $\bolds{y}^\star(\gamma)$ that is feasible to \eqref{eq:dual-D} such that $f(\bolds{\beta}^\star(\gamma)) = \sum_{a \in \set{S}} (c^f(a)-\gamma)y_a$ and thus by weak duality $\bolds{\beta}^\star(\gamma)$ is optimal. 

To prove $\bolds{\beta}^\star(\gamma) \in \set{B}(\gamma)$, Lemma~\ref{lem:sub-tree-beta} shows that $c^f(a) - \sum_{i \in \sub(a)}\beta_i^\star(\gamma)\pi(i) / \pi(a)$ is equal to $V(\gamma, a)$, which is upper bounded by $\gamma$ by its definition \eqref{eq:bellman}. Therefore, $\bolds{\beta}^\star(\gamma) \in \set{B}(\gamma).$

To show $\bolds{\beta}^\star(\gamma)$ is optimal, we construct a feasible $\bolds{y}^\star$ to \eqref{eq:dual-D} with $f(\bolds{\beta}^\star(\gamma)) = \sum_{a \in \set{S}}(c^f(a)-\gamma)y^\star_a.$ This shows $\bolds{\beta}^\star(\gamma)$ is optimal because by weak duality, any feasible $\bolds{\beta} \in \set{B}(\gamma)$ must have $f(\bolds{\beta}) \geq \sum_{a \in \set{S}}(c^f(a)-\gamma)y^\star_a.$ We construct $\bolds{y}^\star = (y^\star_a)_{a \in \set{S}}$ as follows from the lowest level $\set{S}_0$ to the highest level $\set{S}_{L-1}$, which ensures that for any state, at most one of its ancestors $a$ has non-zero $y^\star_a$:
\begin{equation}\label{eq:def-y}
y^\star_a = \left\{
\begin{aligned}
\pi(a),~&\text{if }c(a)+V^f(\gamma,a) > \gamma \text{ and }y^{\star}_k = 0,~\forall k \in \anc(a), k \neq a \\
0,~&\text{otherwise}.
\end{aligned}
\right.
\end{equation}
Since $y^\star_i \geq 0$ and $\sum_{a \in \anc(i)}y_a \pi(i) / \pi(a) \leq \pi(i)$ for any $i$, $\bolds{y}^\star$ is feasible to \eqref{eq:dual-D}. Moreover, by rearranging terms, we verify that
\begin{align*}
\sum_{a \in \set{S}} (c^f(a) - \gamma)y^\star_a - \sum_{i \in \set{S}} \pi(i)\beta_i^\star(\gamma) &=  \underbrace{\sum_{a \in \set{S}} y^\star_a\left( c^f(a) - \gamma - \sum_{i \in \sub(a)}\beta_i^\star(\gamma) \pi(i) / \pi(a)\right)}_{\text{Term 1}} \\
&~- \underbrace{\sum_{i \in \set{S}}\beta_i^\star(\gamma) \left(\pi(i) - \sum_{a \in \anc(i)} y^\star_a \pi(i)/\pi(a)\right)}_{\text{Term 2}}.
\end{align*}
If we establish Terms 1 and 2 are equal to zero, this implies $\sum_{i \in \set{S}} \pi(i)\beta_i^\star(\gamma) = \sum_{a \in \set{S}} (c^f(a) - \gamma)y^\star_a$ and thus $\bolds{\beta}^\star(\gamma)$ is optimal. We prove them as follows:
\begin{itemize}
\item For Term 1, Lemma~\ref{lem:sub-tree-beta} shows $c^f(a) - \sum_{i\in\sub(a)}\beta_i^\star(\gamma)\pi(i)/\pi(a)= V(\gamma,a)$. As a result, $\text{Term 1} = \sum_{a \in \set{S}} y_a^\star (V(\gamma,a)-\gamma) = 0$ because for any $a$, if $y_a^\star > 0$, \eqref{eq:def-y} requires $c(a) + V^f(\gamma,a) > \gamma$ and thus $V(\gamma,a) = \min(c(a)+V^{f}(\gamma,a),\gamma) = \gamma.$
\item For Term 2, if $\beta_i^\star(\gamma) > 0$, then $c(i) + V^f(\gamma,i) > \gamma$. As a result, either $y_i^\star = \pi(i)$ and any ancestor $a$ has $y_a^\star = 0$, or there exists an ancestor $a$ with $y_a^\star = \pi(a)$ by \eqref{eq:def-y}. This implies $\pi(i) - \sum_{a \in \anc(i)} y_a^\star \pi(i) / \pi(a) = 0$. As a result, $\text{Term 2} = 0.$
\end{itemize}
Summarizing the above shows $f(\bolds{\beta}^\star(\gamma)) = \sum_{a \in \set{S}} (c^f(a) - \gamma)y_a^\star$ and thus $\bolds{\beta}^\star(\gamma)$ is optimal  to $D^\star(\gamma).$

For the second result of the lemma, since $\bolds{\beta}^\star(\gamma)$ is an optimal solution, 
\[
D^\star(\gamma) = \mu\cdot \gamma + \lambda f(\bolds{\beta}^\star(\gamma)) = \mu \cdot \gamma + \lambda \sum_{i \in \set{S}} \pi(i)\beta^\star_i(\gamma) = \mu \cdot \gamma + \lambda(c^f(\rootNode) - V(\gamma,\rootNode)),
\]
where, for the last equality, we use Lemma~\ref{lem:sub-tree-beta}.
\end{proof}

\subsection{Proof of Lemma~\ref{lem:aperiodic} (Section~\ref{sec:fluid-sol})}\label{app:lem-aperiodic}
\begin{proof}[Proof of Lemma~\ref{lem:aperiodic}]
Recall that the system state is $(\bolds{Q}(t),R(t))$ and initially $\bolds{Q}(1) = \bolds{0}.$ Define the system state space by the set of states reachable from $(\bolds{0},n)$ for some $n \in \{0,\ldots,N\}.$ To show irreducibility, it is sufficient to show that there are paths connecting these states to state $(\bolds{0},0)$ as state $(\bolds{0},0)$ reaches state $(\bolds{0},n)$ in one step for any $n$. Fix any state $(\bolds{q},n)$ in the system state space and suppose $\bolds{Q}(t) = \bolds{q}$ and $R(t) = n$. Since $\lambda < 1$, with non-zero probability there is no arrival for this period. Moreover, given that the minimum abandonment probability $\theta > 0$, every job in the queue has non-zero probability to abandon. Moreover, $R(t+1) = 0$ with non-zero probability since $\mu < 1$.  As a result, $\bolds{Q}(t+1) = \bolds{0}$ and $R(t+1) = 0$ with non-zero probability, which establishes that the Markov chain is irreducible. For aperiodicity, since state $(\bolds{0},0)$ can transition to itself in one step, the Markov chain contains a self-loop, and is thus aperiodic. Finally, the Markov chain has a finite state space because (i) $Q_i(t), R(t)$ are non-negative integers for any $i$; (ii) $R(t) \leq N$; (iii) $\sum_{i \in \set{S}} Q_i(t) \leq \sum_{\ell = 0}^{L-1} A(t - \ell - 1) \leq L\cdot N$ because a job with state $i \in \set{S}_{\ell}$ in period $t$ must have arrived in period $t - \ell - 1$.     
\end{proof}

\subsection{Proof of Lemma~\ref{lem:feasible-nu} (Section~\ref{sec:fluid-sol})}\label{app:lem-feasible-nu}
\begin{proof}[Proof of Lemma~\ref{lem:feasible-nu}]
Since $\bolds{q}^{\bo}$ and $\bolds{\nu}^{\bo}$ satisfy \eqref{eq:def-fluid-q}, it suffices to verify the two constraints of \eqref{eq:simfluid}, which imply that $(\bolds{q}^{\bo}, \bolds{\nu}^{\bo})$ is feasible to \eqref{eq:orifluid}. For the first constraint of \eqref{eq:simfluid}, recall from \eqref{eq:equiv-constraint} that it is equivalent to having $q_i^{\bo} \geq \nu_i^{\bo}$ for any state $i$. This inequality holds by definition except for the partially-served state $\partialNode$ \ref{item:partially-served}, for which we verify that, if $\partialNode \neq \perp$, then
\begin{align*}
q_{\partialNode}^{\bo} - \nu_{\partialNode}^{\bo} &= \lambda \pi(\partialNode) - \frac{\mu-\lambda\sum_{i \in \TopS(\bo_{[m]})} \pi(i)}{\kappa} &= \frac{\lambda \pi(\partialNode) - \mu + \lambda\sum_{i \in \TopS(\bo_{[m]}) \setminus \sub(\partialNode)} \pi(i)}{\kappa} \geq 0,
\end{align*}
where the second equality uses the definition of $\kappa$ in \eqref{eq:def-degeneracy}. To see the last inequality, if $\partialNode = \perp$, then $\kappa = +\infty$ and thus the fraction becomes zero. When $\partialNode \neq \perp$, the partially-served state $\partialNode$ is equal to $\bo_{m+1}$ and $\TopS(\bo_{[m+1]}) = \{\partialNode\} \cup \TopS(\bo_{[m]}) \setminus \sub(\partialNode).$ If the last inequality is not true, then it contradicts the definition of $m$ which defines $m$ as the maximum position with $\sum_{i \in \TopS(\bo_{[m]})} \lambda \pi(i) \leq \mu$. 

For the second constraint of \eqref{eq:simfluid}, we verify that $\sum_{i \in \set{S}} \nu_i^{\bo} = \sum_{i \in \TopS(\bo_{[m]})} \lambda \pi(i) \leq \mu$ when $\partialNode = \perp$, and when $\partialNode \neq \perp$, it is
\begin{align}
 \sum_{i \in \set{S}} \nu_i^{\bo}&= \sum_{i \in \TopS(\bo_{[m]}) \setminus \sub(\partialNode)} \lambda \pi(i) + \nu_{\partialNode}^{\bo} + \sum_{i \in \TopS(\bo_{[m]}) \cap \sub(\partialNode)} \left(\lambda \pi(i) - \frac{\nu_{\partialNode}^{\bo}\pi(i)}{ \pi(\partialNode)}\right) \nonumber\\
 &= \sum_{i \in \TopS(\bo_{[m]})} \lambda \pi(i) + \nu_{\partialNode}^{\bo}\left(1 - \sum_{i \in \TopS(\bo_{[m]}) \cap \sub(\partialNode)} \pi(i) / \pi(\partialNode)\right) \nonumber\\
 &= \sum_{i \in \TopS(\bo_{[m]})} \lambda \pi(i) + \nu_{\partialNode}^{\bo}\kappa  = \mu, \nonumber
\end{align}
where the second equality uses the definition of $\kappa$ in \eqref{eq:def-degeneracy} and the last equality uses the definition of $\nu_{\partialNode}^{\bo}$ in \ref{item:partially-served}. We have thus established that $\bolds{\nu}^{\bo}$ is feasible to \eqref{eq:simfluid}.
\end{proof}

\subsection{Proof of Lemma~\ref{lem:lower-bound} (Section~\ref{sec:proof-pafou})}\label{app:lem-lower-bound}
Throughout this proof, we assume a fixed system size $N$ and omit it in the notation. Fixing a feasible algorithm $\nalg$, we denote the number of state-$i$ jobs in the queue at the beginning of period $t$ by $Q_i(t,\nalg)$ and the number of served state-$i$ jobs by $R_i(t,\nalg).$ To establish the lower bound, for any horizon $T$, we define a feasible solution, $\bolds{q}^T = (q_i^T), \bolds{\nu}^T = (\nu_i^T),$ to \eqref{eq:orifluid} from the stochastic system by setting
\begin{align}
q^T_i &= \frac{1}{NT}\expect{\sum_{t=1}^T Q_i(t,\nalg)} + \frac{1}{NT}\expect{\sum_{k \in \anc(i) \setminus \{i\}} Z_k(T,\nalg)\frac{\pi(i)}{\pi(k)}}+\frac{\lambda\pi(i)}{T} \label{eq:def-sol-q}\\
\nu^T_i &= \frac{1}{NT}\expect{\sum_{t=1}^T R_i(t,\nalg)}, \label{eq:def-sol-nu}
\end{align}
where we define that $Z_i(t,\nalg) = Q_i(t,\nalg) - R_i(t,\nalg)$ and $\pi(i) = \prod_{k \in \anc(i)} P(\pa(k),k).$
\begin{lemma}\label{lem:feasible-sol}
The constructed $(\bolds{q}^T,\bolds{\nu}^T)$ is a feasible solution to \eqref{eq:orifluid}.
\end{lemma}
\begin{proof}
We verify the constraints in \eqref{eq:orifluid}:
\begin{enumerate}
\item [(i)] $q^T_i = (q^T_{\pa(i)} - \nu^T_{\pa(i)})P(\pa(i),i)$ for any $i \in \set{S} \setminus \{\rootNode\}$: this is because 
\begin{align}
 \frac{1}{NT}\expect{\sum_{t=1}^T Q_i(t,\nalg)} &= \frac{1}{NT}\expect{\sum_{t=1}^T Z_{\pa(i)}(t-1,\nalg)P(\pa(i),i)} \nonumber\\
 &= P(\pa(i),i)\frac{1}{NT}\expect{\sum_{t=1}^{T-1} (Q_{\pa(i)}(t) - R_{\pa(i)}(t))} \nonumber\\
 &\hspace{-1.5in}= P(\pa(i),i)\left(q^T_{\pa(i)} - \nu^T_{\pa(i)} - \frac{1}{NT}\expect{\sum_{k \in \anc(\pa(i))} Z_k(T,\nalg)\frac{\pi(\pa(i))}{\pi(k)}}-\frac{\lambda\pi(\pa(i))}{T}\right) \nonumber\\
 &\hspace{-1.5in}= P(\pa(i),i)(q^T_{\pa(i)} - \nu^T_{\pa(i)}) - \frac{1}{NT}\expect{\sum_{k \in \anc(i) \setminus \{i\}} Z_k(T,\nalg)\frac{\pi(i)}{\pi(k)}} - \frac{\lambda\pi(i)}{T},\label{eq:sol-q-simp}
\end{align}
where the first equality is because each remaining job with state $\pa(i)$ has probability $P(\pa(i),i)$ to become a state-$i$ job; the second uses the definition of $Z$ and shifts $t$ by one period; the third uses the definitions in \eqref{eq:def-sol-q} and \eqref{eq:def-sol-nu}; the last one uses $\pi(i) = \pi(\pa(i))P(\pa(i),i)$ and $\anc(\pa(i)) = \anc(i) \setminus \{i\}$. Plugging \eqref{eq:sol-q-simp} into \eqref{eq:def-sol-q} gives
\begin{align*}
q^T_i &= P(\pa(i),i)(q^T_{\pa(i)} - \nu^T_{\pa(i)}) - \cancel{\frac{1}{NT}\expect{\sum_{k \in \anc(i) \setminus \{i\}} Z_k(T,\nalg)\frac{\pi(i)}{\pi(k)}}} - \cancel{\frac{\lambda\pi(i)}{T}} \\
&\hspace{0.5in}+ \cancel{\frac{1}{NT}\expect{\sum_{k \in \anc(i) \setminus \{i\}} Z_k(T,\nalg)\frac{\pi(i)}{\pi(k)}}}+\cancel{\frac{\lambda\pi(i)}{T}} \\
&= (q^T_{\pa(i)}-\nu^T_{\pa(i)})P(\pa(i),i).
\end{align*}
\item [(ii)] $q^T_{\rootNode} = \lambda$ because \eqref{eq:def-sol-q} and the fact that $\rootNode$ does not have ancestors implies 
\[
q^T_{\rootNode} = \frac{1}{NT}\expect{\sum_{t=1}^T Q_{\rootNode}(t,\nalg)}+\frac{\lambda\pi(\rootNode)}{T} = \frac{1}{NT}\expect{\sum_{t=1}^{T-1} A(t)}+\frac{\lambda}{T} = \frac{\lambda (T-1)}{T}+\frac{\lambda}{T} = \lambda.
\]
\item [(iii)] $\nu^T_i \leq q^T_i$ for any $i \in \set{S}$ because $R_i(t,\nalg) \leq Q_i(t,\nalg)$ for any period $t$.
\item [(iv)] $\sum_{i \in \set{S}} \nu^T_i \leq \mu$ because
\[
\sum_{i \in \set{S}} \nu^T_i = \frac{1}{NT}\sum_{i \in \set{S}}\sum_{t=1}^T \expect{R_i(t,\nalg)} = \frac{1}{NT}\sum_{t=1}^T \expect{\sum_{i \in \set{S}} R_i(t,\nalg)} \leq \frac{1}{NT}\sum_{t=1}^T \expect{R(t)} = \mu,
\]
where $R(t)$ is the number of available servers for period $t$ and has expectation $N\mu.$
\end{enumerate}
\end{proof}

\begin{proof}[Proof of Lemma~\ref{lem:lower-bound}]
For any finite horizon $T$, define the finite-horizon average holding cost 
\[
C_T(\nalg) = \frac{1}{T}\expect{\sum_{t=1}^T \sum_{i \in \set{S}} c(i)(Q_i(t,\nalg) - R_i(t,\nalg))}.
\] 
The long-run average holding cost of $\nalg$ in \eqref{eq:def-cost} satisfies 
$C(\nalg) = \lim\sup_{T \to \infty} C_T(\nalg).$ 

Recall the definition of $\bolds{q}^T$ and $\bolds{\nu}^T$ in \eqref{eq:def-sol-q} and \eqref{eq:def-sol-nu}. We lower bound $C_T(\nalg)$ by
\begin{align}
C_T(\nalg) &= \frac{1}{T}\expect{\sum_{t=1}^T \sum_{i \in \set{S}} c(i)(Q_i(t,\nalg) - R_i(t,\nalg))} \nonumber\\
&= N\sum_{i \in \set{S}} c(i)(q_i^T - \nu_i^T) - \frac{1}{T}\sum_{i \in \set{S}}c(i)\left(\expect{\sum_{k \in \anc(i)\setminus \{i\}} Z_k(T,\nalg)\frac{\pi(i)}{\pi(k)}} + \lambda N\pi(i)\right) \nonumber\\
&\geq NC^\star - \underbrace{\frac{1}{T}\sum_{i \in \set{S}}c(i)\left(\expect{\sum_{k \in \anc(i)\setminus \{i\}} Z_k(T,\nalg)\frac{\pi(i)}{\pi(k)}} + \lambda N\pi(i)\right)}_{(*)},\label{eq:bound-avg-cost}
\end{align}
where the last inequality is because $(\bolds{q}^T,\bolds{\nu}^T)$ is feasible to \eqref{eq:orifluid} by Lemma~\ref{lem:feasible-sol} and $C^\star$ is the optimal value to \eqref{eq:orifluid}. 
To lower bound $(*)$, note that for any state $i$ with level $\ell$ and period $t$, the expected value of $Z_i(t)$ is at most 
\begin{align}
\expect{Z_i(t,\nalg)} &\leq \expect{Z_{\pa(i)}(t-1,\nalg)}P(\pa(i),i) \nonumber\\
&\leq \expect{Z_{\pa(\pa(i))}(t-2,\nalg)}P(\pa(\pa(i)),\pa(i))P(\pa(i),i) \nonumber\\
&\leq \ldots \leq \expect{A(t - \ell - 1)}\pi(i) \leq N\lambda \pi(i). \label{eq:bound-z-iter}
\end{align}
Recalling $c_{\max} = \max_{i \in \set{S}} c(i)$, we upper  bound $(*)$ by:
\begin{align*}
(*) \leq \frac{c_{\max}}{T}\sum_{i \in \set{S}}\left(\sum_{k \in \anc(i) \setminus \{i\}} N\lambda \pi(k)\cdot \frac{\pi(i)}{\pi(k)} + \lambda N \pi(i)\right) &= \frac{c_{\max}}{T}\sum_{i \in \set{S}}\sum_{k \in \anc(i)} N\lambda \pi(i) \\
&\leq \frac{N\lambda c_{\max}}{T}\sum_{i \in \set{S}} L\pi(i) \leq \frac{N\lambda c_{\max}L^2}{T},
\end{align*}
where the first inequality uses \eqref{eq:bound-z-iter}, the equality merges terms; the second inequality uses the fact that a state has at most $L$ ancestors; the last inequality uses $\sum_{i \in \set{S}} \pi(i) \leq L$. The latter is because 
\[
\sum_{i \in \set{S}} \pi(i) = \sum_{\ell=0}^{L-1}\sum_{i \in \set{S}_{\ell}} \pi(i) \leq \sum_{\ell=0}^{L-1}\pi(\rootNode) = L,\quad \text{by applying Lemma~\ref{lem:top-prob} with } \set{Y} = \set{S}_{\ell} \text{ and }\set{X} = \set{S}
\]
and because $\TopS(\set{S}_{\ell}) = \set{S}_{\ell}$, $\TopS(\set{S}) = \{\rootNode\}$.
As a result, \eqref{eq:bound-avg-cost} gives $C_T(\nalg) \geq NC^\star - \frac{N\lambda c_{\max}L^2}{T}$ and  $C(\nalg) = \limsup_{T \to \infty} C_T(\nalg) \geq NC^\star$ for any feasible algorithm $\nalg.$
\end{proof}

\subsection{Proof of Lemma~\ref{lem:prio-optimal} (Section~\ref{sec:optimal-fluid})}\label{app:prio-optimal}
A result we need is that the probability of a job being in a top state of a set $\set{X}$ is larger than that of a subset $\set{Y}$. 
Although this is direct when $\TopS(\set{Y}) \subseteq \TopS(\set{X})$, this condition is in general not true. For example, consider $\set{X} = \{2,3,4,5\}$ and $\set{Y} = \{3,5\}$ in Figure~\ref{fig:water-filling}. Then $\TopS(\set{X}) = \{2,4\}$ while $\TopS(\set{Y}) = \{3,5\}$. The proof of Lemma~\ref{lem:top-prob} is based on induction argument.
\begin{lemma}\label{lem:top-prob}
If $\set{Y} \subseteq \set{X} \subseteq \set{S}$, then $\sum_{i \in \TopS(\set{Y})} \pi(i) \leq \sum_{i \in \TopS(\set{X})} \pi(i).$
\end{lemma}
\begin{proof}
Fix any $\set{X}, \set{Y} \subseteq \set{S}$ such that $\set{Y} \subseteq \set{X}.$ We denote the top set of $\set{Y}$ by $\set{V} = \TopS(\set{Y}).$ For any state~$i$, we define a function $f(i) = \pi(i) - \sum_{k \in \sub(i) \cap \set{V}} \pi(k),$ which is the difference between $\pi(i)$ and the sum of $\pi(k)$ across any $k$ that is in both the top set $\set{V}$ \emph{and} the subtree of $i$.  Since $\set{Y} \subseteq \set{X},$ its top set $\set{V}$ is in the subtree of $\TopS(\set{X})$, and thus has a partition $\{\set{V}_i\}_{i \in \TopS(\set{X})}$ with $\set{V}_i = \set{V} \cap \sub(i)$ for $i \in \TopS(\set{X})$ such that $\set{V} = \bigcup_{i \in \TopS(\set{X})} \set{V}_i.$ As a result,
\begin{equation}\label{eq:decomp-diff}
\sum_{i \in \TopS(\set{X})} \pi(i) - \sum_{i \in \TopS(\set{Y})} \pi(i) = \sum_{i \in \TopS(\set{X})} \pi(i) - \sum_{i \in \TopS(\set{X})}\sum_{k \in \set{V}_i} \pi(k) = \sum_{i \in \TopS(\set{X})} f(i).
\end{equation}
To prove the lemma, it is sufficient to show $f(i) \geq 0$ for any state $i$. Recall that states have a partition based on their levels $\set{S} = \set{S}_0 \cup \cdots \cup \set{S}_{L-1}.$ We show $f(i) \geq 0$ by induction from $\set{S}_{L-1}$ to $\set{S}_0$:
\begin{itemize}
\item base case: for any $i \in \set{S}_{L-1},$ since it has no children, $f(i) = \pi(i) - \pi(i) \indic{i \in \set{V}} \geq 0.$
\item induction step: suppose $f(k) \geq 0$ for any $k \in \set{S}_{\ell+1}.$ Given a state $i \in \set{S}_{\ell},$ if $i \in \set{V}$, this implies there is no other state in its subtree inside $\set{V}$; otherwise $\set{V}$ is not a top set. As a result, $f(i) = \pi(i) - \pi(i) = 0.$ If $i \not \in \set{V},$ 
\[
f(i) = \pi(i) - \sum_{a \in \child(i)} \sum_{k \in \sub(a) \cap \set{V}} \pi(k) \geq \pi(i) - \sum_{a \in \child(i)} \pi(a) = \pi(i) - \pi(i)\sum_{a \in \child(i)} P(i,a) \geq 0,
\]
where the first inequality is by the induction hypothesis that $f(a) = \pi(a) - \sum_{k \in \sub(a) \cap \set{V}} \pi(k) \geq 0$ since $a \in \set{S}_{\ell+1}.$
\end{itemize}
The above shows that $f(i) \geq 0$ for any state $i$. Using \eqref{eq:decomp-diff} gives $\sum_{i \in \TopS(\set{X})} \pi(i) \geq \sum_{i \in \TopS(\set{Y})} \pi(i).$
\end{proof}
Recall that $\gamma^\star$ is the optimal capacity dual with the lowest $D^\star(\gamma)$ in \eqref{eq:dual}. Complementary slackness characterizes whether a feasible service decision vector $\bolds{\nu}$ is optimal for \eqref{eq:simfluid} in the below lemma. 

\begin{lemma}\label{lem:comple-slack}
Given a service decision vector $\bolds{\nu}$ that is feasible to \eqref{eq:simfluid}, it is optimal if and only if the following three conditions hold:
\begin{enumerate}[label=\textnormal{(C-\arabic*)}]
\item $\sum_{i \in \set{S}} \nu_i = \mu$ when $\gamma^\star > 0$;  \label{item:c-1}
\item if $c(i) + V^f(\gamma^\star,i) > \gamma^\star$, then $q_i = \nu_i$ where $q_i$ is defined by \eqref{eq:def-fluid-q}; \label{item:c-2}
\item if $c(i) + V^f(\gamma^\star,i) < \gamma^\star$ for state $i$, then $\nu_i = 0$. \label{item:c-3}
\end{enumerate}
\end{lemma}
\begin{proof}
Recall that $\gamma^\star$ is the optimal capacity dual and $\bolds{\beta}^\star(\gamma^\star)$ is an optimal solution to \eqref{eq:dual} by Lemma~\ref{lem:dual-structure}. To ease the notation, let $\beta_i^\star = \beta_i^\star(\gamma^\star) = \max(0, c(i) + V^f(\gamma^\star,i)-\gamma^\star)$ for any state $i$. For any feasible service decision $\bolds{\nu}=(\nu_i)_{i \in \set{S}}$ to \eqref{eq:simfluid}, the following equation holds true by rearranging terms: 
\begin{align*}
\left(\lambda\sum_{i \in \set{S}} \pi(i)\beta_i^\star + \mu \gamma^\star\right) - \sum_{a \in \set{S}} \nu_a c^f(a) &= 
\underbrace{\gamma^\star \left(\mu - \sum_{i \in \set{S}} \nu_i\right)}_{\text{Term 1}}+\underbrace{\sum_{i \in \set{S}} \beta_i^\star\left(\lambda \pi(i) - \sum_{a \in \anc(i)} \nu_a\frac{\pi(i)}{\pi(a)}\right) + }_{\text{Term 2}} \\
&\hspace{0.2in}+\underbrace{\sum_{a \in \set{S}} \nu_a\left(\gamma^{\star}+\sum_{i \in \sub(a)} \beta^\star_i \frac{\pi(i)}{\pi(a)} - c^f(a)\right)}_{\text{Term 3}}.
\end{align*}
By strong duality, $\bolds{\nu}$ is optimal to \eqref{eq:simfluid} if and only if the left hand side of the above equation is equal to zero, which is equivalent to having Terms 1 to 3 equal to zero because each of these terms is non-negative by the feasibility of $\bolds{\nu}$ to \eqref{eq:simfluid} and the feasibility of $(\gamma^\star,\bolds{\beta}^\star)$ to \eqref{eq:dual}. 

We prove the lemma by showing that Terms 1 to 3 are equivalent to Conditions \ref{item:c-1} to \ref{item:c-3}:
\begin{itemize}
\item Term 1 and Condition \ref{item:c-1} are equivalent because (i) if \ref{item:c-1} is satisfied, then Term 1 evaluates to zero because $\gamma^\star \geq 0$; (ii) if Term 1 is zero, then either $\gamma^\star$ or $\mu - \sum_{i \in \set{S}} \nu_i$ is zero, satisfying \ref{item:c-1}.
\item Term 2 and Condition \ref{item:c-2} are equivalent. If \ref{item:c-2} is satisfied, it implies that for any state $i$ with $\beta_i^\star > 0$, we have $q_i$, as defined in \eqref{eq:def-fluid-q}, equal to $\nu_i$, implying that $\lambda \pi(i) = \sum_{a \in \anc(i)} \nu_a \frac{\pi(i)}{\pi(a)}$. Since $\beta_i^\star \geq 0$ for any state $i$, this implies Term 2 evaluates to zero. In addition, if Term 2 is equal to zero, it requires that for any state $i$, either $\beta_i^\star = 0$ or $\lambda \pi(i) = \sum_{a \in \anc(i)} \nu_a \frac{\pi(i)}{\pi(a)}$. Therefore, if $c(i) + V^f(\gamma^\star,i) > \gamma$, i.e., $\beta_i^\star > 0$, \eqref{eq:def-fluid-q} requires $q_i = \lambda \pi(i) - \sum_{a \in \anc(i) \setminus \{i\}} \nu_a \frac{\pi(i)}{\pi(a)} = \nu_i,$ satisfying \ref{item:c-2}.  
\item Term 3 and Condition \ref{item:c-3} are equivalent. To see this, Lemma~\ref{lem:sub-tree-beta} shows that $V(\gamma^\star,a) = c^f(a) - \sum_{i \in \sub(a)}\beta_i^\star \pi(i) / \pi(a)$ and thus Term 3 simplifies to $\sum_{a \in \set{S}} \nu_a (\gamma - V(\gamma^\star,a))$ where $V(\gamma^\star,a) = \min(\gamma^\star, c(a) + V^f(\gamma^\star,a)).$ If \ref{item:c-3} holds true, it implies that for any state $a$ with $V(\gamma^\star,a) < \gamma^{\star}$, $\nu_a = 0$ and thus Term 3 is equal to zero. Moreover, if Term 3 is equal to zero, either $\nu_a = 0$ or $V(\gamma^\star,a) = \gamma^{\star}$ for any state $a$, which implies \ref{item:c-3}.
\end{itemize}
Summarizing the above, $\bolds{\nu}$ is optimal to \eqref{eq:simfluid} if and only if Terms 1 to 3 are equal to zero, which we established is equivalent to having Conditions \ref{item:c-1} to \ref{item:c-3}.
\end{proof}

\begin{proof}[Proof of Lemma~\ref{lem:prio-optimal}]
We show that $\bolds{\nu}^{\bo}$ is optimal to \eqref{eq:simfluid} using Lemma~\ref{lem:comple-slack}, which implies the optimality of $(\bolds{q}^{\bo}, \bolds{\nu}^{\bo})$ to \eqref{eq:orifluid} by the equivalence between \eqref{eq:orifluid} and \eqref{eq:simfluid}. 

To do so, recall that the priority ordering $\bo$ ranks states in $\set{S}_{\textsc{Hi}}$ before states in $\set{S}_{\textsc{Eq}}$ and those in $\set{S}_{\textsc{Eq}}$ before $\set{S}_{\textsc{Lo}}$. Accordingly, there is a \emph{high position} $\mhigh$ and a \emph{low position} $\mlow$ such that $c(o_h) + V^f(\gamma^\star,o_h) > \gamma^\star$ for any $h \leq \mhigh$, $c(o_h) + V^f(\gamma^\star, o_h) = \gamma^\star$ for any $h \in (\mhigh,\mlow]$, and $c(o_h) + V^f(\gamma^\star, o_h) < \gamma^\star$ for any $h > \mlow$. Since Lemma~\ref{lem:feasible-nu} implies that $\bolds{\nu}^{\bo}$ is feasible to \eqref{eq:simfluid}, it remains to verify conditions \ref{item:c-1} to \ref{item:c-3} in Lemma~\ref{lem:comple-slack} for $\bolds{\nu}^{\bo}.$ Before that, recall for the priority ordering $\bo$, $m$ is the maximum position such that $\sum_{i \in \TopS(\bo_{[m]})} \lambda \pi(i) \leq \mu$. If the inequality is strict and $m < |\set{S}|$, the partially-served state $\partialNode$ is $o_{m+1}$; otherwise it is $\perp$.

Condition \ref{item:c-1} requires that $\sum_{i \in \set{S}} \nu_i^{\bo} = \mu$ when $\gamma^\star > 0$. By Lemma~\ref{lem:feasible-nu}, if the partially-served state $\partialNode \neq \perp$, then $\sum_{i \in \set{S}} \nu_i^{\bo} = \mu$. Thus we focus on the case when $\partialNode = \perp.$ In this case, $\sum_{i \in \set{S}} \nu_i^{\bo} = \sum_{i \in \TopS(\bo_{[m]})} \lambda \pi(i)$. Moreover, by our construction, the case $\partialNode = \perp$ only happens when $\sum_{i \in \TopS(\bo_{[m]})} \lambda \pi(i) = \mu$, which implies $\sum_{i \in \set{S}} \nu_i^{\bo} = \mu$, or when $m = |\set{S}|$ and $\sum_{i \in \TopS(\bo_{[m]})} \lambda \pi(i) < \mu$. For the latter scenario, note that when $m = |\set{S}|$, the top set $\TopS(\bo_{[m]}) = \TopS(\set{S})$ is equal to the root node $\rootNode.$ As a result, in this scenario $\lambda \pi(\rootNode) = \lambda < \mu.$ Recall that $\gamma^\star$ minimizes $g(\gamma) \coloneqq \mu \cdot \gamma - \lambda V(\gamma,\rootNode)$ among $\gamma \geq 0$ by Lemma~\ref{lem:dual-structure}. Since now $\lambda < \mu$, this implies 
\begin{equation}\label{eq:g-gamma}
\forall \gamma > 0,~g(\gamma)=\mu \cdot \gamma - \lambda V(\gamma,\rootNode) \geq \mu \cdot \gamma - \lambda \cdot \gamma = (\mu-\lambda)\gamma > 0 = g(0).
\end{equation}
Therefore, $\gamma^\star = 0$ when $\lambda < \mu$, satisfying \ref{item:c-1}. We thus verify this condition for $\bolds{\nu}^{\bo}$.

Condition \ref{item:c-2} requires $q_i^{\bo} = \nu_i^{\bo}$ for any $i \in \bo_{[\mhigh]}.$ To show this, it is sufficient to show that $\sum_{i \in \TopS(\bo_{[\mhigh]})} \lambda \pi(i) \leq \mu$ as it implies any state $i \in \bo_{[\mhigh]}$ is fully-blocked \ref{item:fully-blocked}, partially-blocked \ref{item:partially-blocked} or an empty state \ref{item:empty}, which gives $q_i^{\bo} = \nu_i^{\bo}$ by our construction. To show $\sum_{i \in \TopS(\bo_{[\mhigh]})} \lambda \pi(i) \leq \mu$, by \ref{item:c-2} of Lemma~\ref{lem:comple-slack} an optimal solution $(\bolds{q}^\star,\bolds{\nu}^\star)$ to \eqref{eq:orifluid} satisfies $q_i^{\star} = \nu_i^{\star}$ for $i \in \bo_{[\mhigh]}$ and thus
\begin{align}
\sum_{i \in \TopS(\bo_{[\mhigh]})} \lambda \pi(i) &\overset{\eqref{eq:def-fluid-q}}{=} \sum_{i \in \TopS(\bo_{[\mhigh]})} \left(q^\star_i + \sum_{a \in \anc(i) \setminus \{i\}} \nu_a^\star \cdot \frac{\pi(i)}{\pi(a)}\right) \overset{\text{Lemma~\ref{lem:comple-slack}}}{=} \sum_{i \in \TopS(\bo_{[\mhigh]})} \sum_{a \in \anc(i)} \nu_a^\star \cdot \frac{\pi(i)}{\pi(a)} \nonumber\\
&\hspace{0.2in}= \sum_{a \in \set{S}} \nu_a^\star \left(\sum_{i \in \TopS(\bo_{[\mhigh]}) \cap \sub(a)} \pi(i)\right) / \pi(a) \overset{(*)}{\leq} \sum_{a \in \set{S}} \nu_a^\star \pi(a) / \pi(a) = \sum_{a \in \set{S}} \nu_a^\star \leq \mu, \label{eq:shi-prob}
\end{align}
where Inequality (*) applies Lemma~\ref{lem:top-prob} with $\set{X} = \sub(a)$ and $\set{Y} = \TopS(\bo_{[\mhigh]}) \cap \sub(a)$ and noting that $\TopS(\set{X}) = \{a\}, \TopS(\set{Y}) = \set{Y}, \set{Y} \subseteq \set{X}.$ We thus show $\sum_{i \in \TopS(\bo_{[\mhigh]})} \lambda \pi(i) \leq \mu$, which implies Condition \ref{item:c-2} as we showed at the beginning of this part.

Condition \ref{item:c-3} requires $\nu_i^{\bo} = 0$ for any $i \not \in \bo_{[\mlow]}.$ We show this by a case discussion:
\begin{itemize}
\item $m + 1 \leq \mlow$: since our construction  only allows state $i \in \bo_{[m+1]}$ to have non-zero $\nu_i^{\bo}$, this gives $\nu_i^{\bo} = 0$ for any $i \not \in \bo_{[\mlow]}$ since $\bo_{[m+1]} \subseteq \bo_{[\mlow]}$ when $m + 1 \leq \mlow$.
\item $m+1 > \mlow$: By \ref{item:c-3} of Lemma~\ref{lem:comple-slack}, an optimal solution $(\bolds{q}^\star,\bolds{\nu}^\star)$ to \eqref{eq:orifluid} satisfies $\nu_i^{\star} = 0$ for $i \not \in \bo_{[\mlow]}$ and thus
\begin{align}
C^\star &= \sum_{i \in \set{S}} c(i)(q_i^\star - \nu_i^\star) \overset{\eqref{eq:sim-obj}}{=} \lambda \sum_{i \in \set{S}} \pi(i)c(i) - \sum_{a \in \set{S}} \nu^\star_a \sum_{i \in \sub(a)} c(i)\pi(i) / \pi(a) \nonumber\\
&\hspace{0.2in}\overset{\text{Lemma~\ref{lem:comple-slack}}}{=} \lambda \sum_{i \in \set{S}} \pi(i)c(i) - \sum_{a \in \bo_{[\mlow]}} \nu^\star_a \sum_{i \in \sub(a)} c(i)\pi(i) / \pi(a) \nonumber\\
&= \lambda \sum_{i \in \set{S}} \pi(i)c(i) - \sum_{i \in \sub(\bo_{[\mlow]})} c(i)\sum_{a \in \anc(i)} \nu_a^\star \pi(i) / \pi(a) \nonumber\\
&\geq \lambda \sum_{i \in \set{S}} \pi(i)c(i) - \sum_{i \in \sub(\bo_{[\mlow]})} \lambda c(i)\pi(i) = \lambda \sum_{i \not \in \sub(\bo_{[\mlow]})} \pi(i)c(i),\label{eq:bound-opt}
\end{align}
where the inequality is by the first constraint of \eqref{eq:simfluid}. As a result,
\[
\sum_{i \in \set{S}} c(i)(q_i^{\bo} - \nu_i^{\bo}) = \sum_{i \in \set{S} \setminus \sub(\bo_{[\mlow]})} c(i)(q_i^{\bo} - \nu_i^{\bo}) \leq \lambda \sum_{i \in \set{S} \setminus \sub(\bo_{[\mlow]})} c(i) \pi(i) \leq C^\star,
\]
where the first equality is because $m + 1 > \mlow$ and thus $q_i^{\bo} = \nu_i^{\bo}$ for $i \in \sub(\bo_{[\mlow]})$ by our construction in \ref{item:fully-blocked}, \ref{item:empty} and \ref{item:partially-blocked}; the first inequality is because $q_i^{\bo} \leq \lambda \pi(i)$ for any $i$; the last inequality is by \eqref{eq:bound-opt}. Hence, if $m + 1 > \mlow$, we directly show that $(\bolds{q}^{\bo}, \bolds{\nu}^{\bo})$ is optimal to \eqref{eq:orifluid}, the desired result of this lemma, which also implies \ref{item:c-3} by Lemma~\ref{lem:comple-slack}.
\end{itemize}
Summarizing the above, Conditions \ref{item:c-1}, \ref{item:c-2}, \ref{item:c-3} hold for $\bolds{\nu}^{\bo}$. Lemma~\ref{lem:comple-slack} thus shows it is optimal to \eqref{eq:simfluid}, implying that $(\bolds{q}^{\bo}, \bolds{\nu}^{\bo})$ is optimal to \eqref{eq:orifluid}.
\end{proof}

\subsection{Proof of Lemma~\ref{lem:queue-upper-bound} (Section~\ref{sec:stochastic})}\label{app:lem-queue-upper-bound}
\begin{proof}
The challenge in directly proving this result lies in the following two issues: 
\begin{enumerate}
    \item[(a)] the number of arrivals in each period is random; 
    \item[(b)] the correlation between jobs due to capacity constraint because a job's presence in the queue affects another job's chance of getting service.
\end{enumerate}
We circumvent these two issues by expanding the original system with \emph{fictitious} jobs such that (a) in each period there are always $N$ new (fictitious) jobs and (b) even if a (fictitious) job leaves because of service, we keep it in the system and thus remove correlation between jobs' presence.

In particular, in each period there are always $N$ new (fictitious) jobs. A new job $j$ has probability $\lambda$ to have a state $S_j(t+1) = \rootNode$, while with probability $1-\lambda$, the job has a state $S_j(t+1) = \perp$, i.e., it immediately abandons upon its arrival. The original set of new jobs, $\set{A}(t)$, corresponds to the fictitious jobs with $S_j(t+1) = \rootNode.$  Moreover, for the jobs that obtain service and leave, i.e., for $j \in \set{R}(1) \cup \cdots \cup \set{R}(t)$, we assume that their states still transition according to the Markov chain. For a period $t$, we define the set of fictitious job arrivals in this period by $\set{J}(t),$ which deterministically contains $N$ jobs. The above expansion does not give fewer jobs than what exist in the original system. Therefore, for a set of states $\set{X} \subseteq \set{S}$, the number of waiting jobs with states in $\set{X}$ is 
\[
\sum_{i \in \set{X}} Q_i(t) \leq \sum_{\tau=1}^{t-1}\sum_{j \in \set{J}(\tau)} \indic{S_j(t) \in \set{X}} \leq N\min(L,|\set{X}|),
\]
where the second inequality is because for a job $j \in \set{J}(\tau)$ to have a state $S_j(t)$ with level $\ell$, it must arrive in period $t - \ell-1$, and the states in $\set{X}$ contain at most $\min(L,|\set{X}|)$ different levels. 

Since job states transition independently, $X \coloneqq \sum_{\tau=1}^{t-1}\sum_{j \in \set{J}(\tau)} \indic{S_j(t) \in \set{X}}$ is the sum of $N(t-1)$ independent binary random variables $\{X_j\}_{j \in \set{J}(1)\cup \ldots \cup \set{J}(t-1)}$ where $X_j = \indic{S_j(t) \in \set{X}}.$ Moreover,
\begin{align*}
\expect{X} = \expect{\sum_{\tau=1}^{t-1} \sum_{j \in \set{J}(\tau)} \indic{S_j(t) \in \set{X}}} &= \sum_{\ell=0}^{L-1} \sum_{i \in \set{X} \cap \set{S}_{\ell}}\sum_{j \in \set{J}(t - \ell - 1)}\Pr\{S_j(t) = i\}\\ &\leq \sum_{\ell=0}^{L-1}\sum_{i \in \set{X} \cap \set{S}_{\ell}} N\lambda\pi(i) = N\lambda\sum_{i \in \set{X}} \pi(i),
\end{align*}
where the first equality is by partitioning states in $\set{X}$ according to their levels and noting that only jobs arriving in period $t - \ell-1$ will have states of level $\ell$ in period $t$; the inequality is because (i) the size of $\set{J}(t-\ell-1)$ is $N$ when $t > \ell + 1$ and equal to zero otherwise; (ii) the probability that a new fictitious job transitions to state $i \in \set{S}_{\ell}$ in $\ell$ periods is $\lambda \pi(i)$. Moreover, the abandonment property gives that for any $\ell > 0$, by induction,
\[
\sum_{i \in \set{S}_{\ell}} \pi(i) = \sum_{i' \in \set{S}_{\ell-1}} \pi(i')\sum_{i \in \child(i')} P(i',i) \leq \sum_{i' \in \set{S}_{\ell-1}} \pi(i')(1-\theta) \leq (1-\theta)^2\sum_{i' \in \set{S}_{\ell-2}} \pi(i') \leq \cdots \leq (1-\theta)^{\ell}.
\]
Therefore, $\sum_{i \in \set{X}} \pi(i) \leq \sum_{\ell=0}^{L-1} (1-\theta)^{\ell} \leq \theta^{-1}$ and thus $\expect{X} \leq N/\theta$.

To give a high probability bound of $\sum_{i \in \set{X}} Q_i(t)$, consider two cases. First, if $\expect{X} = 0$, it implies $X = 0$ since $X \geq 0$ and thus $\sum_{i \in \set{X}} Q_i(t) = 0$. Second, if $\expect{X} > 0$, we finish the proof by
\begin{align*}
\Pr\left\{\sum_{i \in \set{X}} Q_i(t) \geq N\lambda \sum_{i \in \set{X}} \pi(i) + 3\sqrt{N/\theta}\ln \frac{1}{\delta}\right\} &\leq \Pr\left\{X \geq N\lambda \sum_{i \in \set{X}} \pi(i) + 3\sqrt{N/\theta}\ln \frac{1}{\delta}\right\} \\
&\hspace{-2in} \leq \Pr\left\{X \geq \expect{X}\left(1 + \frac{3\sqrt{N/\theta}\ln\frac{1}{\delta}}{\expect{X}}\right)\right\} \overset{(a)}{\leq} \exp\left(\frac{-9(N/\theta)\ln^2\frac{1}{\delta}}{2\expect{X} + 2\sqrt{N/\theta}\ln\frac{1}{\delta}}\right) \\
&\overset{(b)}{\leq} \exp\left(\frac{-9(N/\theta)\ln^2\frac{1}{\delta}}{4(N/\theta)\ln\frac{1}{\delta}}\right) \leq \delta^2,
\end{align*}
where Inequality (a) is by noting $X$ is the sum of independent binary random variables and applying Fact~\ref{fact:chernoff}; Inequality (b) is by $\expect{X} \leq N/\theta$ and the assumption that $\delta \in (0,1/e]$.
\end{proof}

\subsection{Proof of Lemma~\ref{lem:tail-bound} (Section~\ref{sec:stochastic})}\label{app:lem-tail-bound}
The proof follows a similar structure to that of theorem 1 in \cite{bertsimas2001performance}. We first establish the following result.
\begin{lemma}\label{lem:itera-bound}
Given the same conditions of Lemma~\ref{lem:tail-bound} and that $v_{\whp} \leq v_{\max}$, then $\forall a \geq B - v_{\whp},$ \[\Pr\{\Phi(\bolds{X}(\infty)) > a + v_{\whp}\} \leq \frac{v_{\whp}}{v_{\whp}+\Delta}\Pr\{\Phi(\bolds{X}(\infty)) > a - v_{\whp}\} + \frac{6\varepsilon v_{\max}}{v_{\whp}}.\]
\end{lemma}
\begin{proof}
Define a stochastic process $\{\bolds{Y}(t)\}$ such that $\bolds{Y}(t) = \bolds{X}(t+1).$ Let $\hat{\Phi}(\bolds{x}) = \max(a, \Phi(\bolds{x}))$. Condition (iii) of the lemma also implies an almost sure upper bound on $\hat{\Phi}(\bolds{X}(t+1)) - \hat{\Phi}(\bolds{X}(t))$ since if $\Phi(\bolds{X}(t+1)) - \Phi(\bolds{X}(t)) \leq v_{\max}$ {holds} almost surely, it implies almost surely that
\begin{align}
\hat{\Phi}(\bolds{X}(t+1)) - \hat{\Phi}(\bolds{X}(t)) &= \max(a, \Phi(\bolds{X}(t+1))) - \max(a,\Phi(\bolds{X}(t))) \nonumber\\
&\leq \max(a - a, \Phi(\bolds{X}(t+1)) - \Phi(\bolds{X}(t))) \leq v_{\max}. \label{eq:as-bound}
\end{align}
The assumption that $\expect{\Phi(\bolds{X}(\infty))} < \infty$ implies $\expect{\hat{\Phi}(\bolds{X}(\infty))} < \infty.$ Since $\bolds{X}(t)$ has a limiting distribution, both $\bolds{X}(\infty)$ and $\bolds{Y}(\infty)$ exist and are identically distributed (though not independent). As a result, recalling the state space is $\set{X}$, the following equation holds:
\begin{align}
&\expect{\hat{\Phi}(\bolds{Y}(\infty)) - \hat{\Phi}(\bolds{X}(\infty))} = 0, \text{ and thus } \nonumber \\ 
&\sum_{\bolds{x} \in \set{X}} \Pr\{\bolds{X}(\infty) =\bolds{x}\}\underbrace{\expect{\hat{\Phi}(\bolds{Y}(\infty)) - \hat{\Phi}(\bolds{X}(\infty)) \mid \bolds{X}(\infty) = \bolds{x}}}_{\text{drift } D(\bolds{x})} = 0. \label{eq:stationary}
\end{align}
The proof proceeds by bounding the drift term $D(\bolds{x})$ considering three types of state $\bolds{x} \in \set{G}$:
\begin{itemize}
\item  $\Phi(\bolds{x}) \leq a - v_{\whp}$: in this case $\hat{\Phi}(\bolds{x}) = a$. Moreover, by Condition (i) of the lemma, for any period $t$, $\Pr\{\Phi(\bolds{X}(t+1)) - \Phi(\bolds{X}(t)) > v_{\whp} | \bolds{X}(t) = \bolds{x}\} \leq \varepsilon$. That is, with probability at least $1 - \varepsilon$, $\Phi(\bolds{X}(t+1)) \leq \Phi(\bolds{x}) + v_{\whp} \leq a$  and $\hat{\Phi}(\bolds{X}(t+1)) = a$ conditioning on $\bolds{X}(t) = \bolds{x}.$ Moreover, by \eqref{eq:as-bound}, $\hat{\Phi}(\bolds{X}(t+1)) - \hat{\Phi}(\bolds{X}(t)) \leq v_{\max}$ almost surely. As a result, recalling $\bolds{Y}(t) = \bolds{X}(t+1),$ 
$\expect{\hat{\Phi}(\bolds{Y}(t)) - \hat{\Phi}(\bolds{X}(t)) \mid \bolds{X}(t) = \bolds{x}} \leq (1 - \varepsilon)(a - a) + \varepsilon\cdot v_{\max} = \varepsilon v_{\max}.$
Since this bound holds for any $t$, it is also true for the limiting distribution. We thus conclude the drift $D(\bolds{x}) \leq \varepsilon v_{\max}$ when $\Phi(\bolds{x}) \leq a - v_{\whp}$ and $\bolds{x} \in \set{G}.$
\item  $\Phi(\bolds{x}) > a + v_{\whp}$: since $a \geq B - v_{\whp},$ $\Phi(\bolds{x}) \geq B$. By Condition (ii) of the lemma, conditioning on $\bolds{X}(t) = \bolds{x}$, with probability at least $1 - \varepsilon$, $\Phi(\bolds{X}(t+1)) \leq \Phi(\bolds{x}) - \Delta$, which gives 
\begin{align*}
\hat{\Phi}(\bolds{X}(t+1)) - \hat{\Phi}(\bolds{X}(t)) &= \max(a, \Phi(\bolds{X}(t+1))) - \max(a,\Phi(\bolds{x})) \\
&\leq \max(a, \Phi(\bolds{x}) - \Delta) - \Phi(\bolds{x}) \leq -\Delta,
\end{align*}
where the last inequality is because $\Phi(\bolds{x}) \geq a + v_{\whp}$ and $v_{\whp} \geq \Delta$ by assumption. By \eqref{eq:as-bound}$, \hat{\Phi}(\bolds{X}(t+1)) - \hat{\Phi}(\bolds{X}(t)) \leq v_{\max}$ almost surely. Therefore, \[\expect{\hat{\Phi}(\bolds{Y}(t)) - \hat{\Phi}(\bolds{X}(t)) \mid \bolds{X}(t) = \bolds{x}} \leq (1 - \varepsilon)(-\Delta) + \varepsilon\cdot v_{\max} = -\Delta + \varepsilon\Delta + \varepsilon v_{\max} \leq -\Delta + 2\varepsilon v_{\max},\]
where the last inequality is by the assumption that $\Delta \leq v_{\whp} \leq v_{\max}.$ Taking $t \to \infty$ shows that the drift $D(\bolds{x}) \leq -\Delta + 2\varepsilon v_{\max}$ when $\Phi(\bolds{x}) \geq a + v_{\whp}$ and $\bolds{x} \in \set{G}.$
\item $\Phi(\bolds{x}) \in (a - v_{\whp}, a + v_{\whp}]:$ for this case, recall from \eqref{eq:as-bound} that $\hat{\Phi}(\bolds{X}(t+1)) - \hat{\Phi}(\bolds{X}(t)) \leq \max(0, \Phi(\bolds{X}(t+1)) - \Phi(\bolds{X}(t))).$ By Condition (i) of the lemma, conditioning on $\bolds{X}(t) = \bolds{x} \in \set{G}$, with probability at least $1-\varepsilon$, the difference $\Phi(\bolds{X}(t+1)) - \Phi(\bolds{X}(t))$ is at most $v_{\whp}.$ Therefore, 
\[\expect{\hat{\Phi}(\bolds{Y}(t)) - \hat{\Phi}(\bolds{X}(t)) \mid \bolds{X}(t) = \bolds{x}} \leq (1 - \varepsilon)v_{\whp} + \varepsilon\cdot v_{\max} \leq v_{\whp} + \varepsilon v_{\max},\]
which implies $D(\bolds{x}) \leq v_{\whp} + \varepsilon v_{\max}$ when $\Phi(\bolds{x}) \in (a-v_{\whp},a+v_{\whp}]$ and $\bolds{x} \in \set{G}.$
\end{itemize}
Applying the above three cases to \eqref{eq:stationary} gives
\begin{align*}
0 &\leq \Pr\left\{\Phi(\bolds{X}(\infty)) \leq a - v_{\whp}\right\}(\varepsilon v_{\max}) + \Pr\left\{\Phi(\bolds{X}(\infty)) > a + v_{\whp}, \bolds{X}(\infty) \in \set{G}\right\}(-\Delta + 2\varepsilon v_{\max}) \\
&\hspace{0.2in}+ \Pr\left\{\Phi(\bolds{X}(\infty)) \in (a - v_{\whp}, a + v_{\whp}]\right\}(v_{\whp} + \varepsilon v_{\max}) + \Pr\left\{\bolds{X}(\infty) \not \in \set{G}\right\} v_{\max}. 
\end{align*}
Since $\Pr\left\{\Phi(\bolds{X}(\infty)) > a + v_{\whp}, \bolds{X}(\infty) \in \set{G}\right\} \geq \Pr\left\{\Phi(\bolds{X}(\infty)) > a + v_{\whp}\right\} - \Pr\{\bolds{X}(\infty) \not \in \set{G}\}$ by the union bound and $\Pr\{\bolds{X}(\infty) \not \in \set{G}\} \leq \varepsilon$ by the lemma condition, the above inequality simplifies to 
\begin{align*}
0 &\leq -\Delta\Pr\left\{\Phi(\bolds{X}(\infty)) > a + v_{\whp}\right\} + v_{\whp}\Pr\left\{\Phi(\bolds{X}(\infty)) \in (a - v_{\whp}, a + v_{\whp}]\right\}+ 6\varepsilon v_{\max} \\
&= -\Delta\Pr\left\{\Phi(\bolds{X}(\infty)) > a + v_{\whp}\right\} + v_{\whp}\left(\Pr\left\{\Phi(\bolds{X}(\infty)) > a - v_{\whp}\right\} - \Pr\left\{\Phi(\bolds{X}(\infty)) > a + v_{\whp}\right\}\right)\\
&\hspace{0.2in} + 6\varepsilon v_{\max}.
\end{align*}
Therefore, by rearranging terms, we obtain
\[
\Pr\left\{\Phi(\bolds{X}(\infty)) > a + v_{\whp}\right\} \leq \frac{v_{\whp}}{v_{\whp} + \Delta}\Pr\left\{\Phi(\bolds{X}(\infty)) > a - v_{\whp}\right\} + \frac{6\varepsilon v_{\max}}{v_{\whp}}.
\]
\end{proof}
\begin{proof}[Proof of Lemma~\ref{lem:tail-bound}]
For any $\ell \geq 0$, applying Lemma~\ref{lem:itera-bound}  iteratively with $a = B + 2v_{\whp}\ell - v_{\whp}, B + 2v_{\whp}(\ell-1) - v_{\whp},\ldots,B-v_{\whp}$  gives
\begin{align*}
\Pr\{\Phi(\bolds{X}(\infty)) > B + 2v_{\whp} \ell\} &\leq \frac{v_{\whp}}{v_{\whp}+\Delta}\Pr\{\Phi(\bolds{X}(\infty)) > B + 2v_{\whp}(\ell-1)\} + \frac{6\varepsilon v_{\max}}{v_{\whp}} \\
&\hspace{-2in}\leq \left(\frac{v_{\whp}}{v_{\whp}+\Delta}\right)^2 \Pr\{\Phi(\bolds{X}(\infty)) > B + 2v_{\whp}(\ell-2)\} + \left(1 + \frac{v_{\whp}}{v_{\whp}+\Delta}\right)\frac{6\varepsilon v_{\max}}{v_{\whp}} \\
&\hspace{-2in}\leq \left(\frac{v_{\whp}}{v_{\whp}+\Delta}\right)^2 \Pr\{\Phi(\bolds{X}(\infty)) > B + 2v_{\whp}(\ell-2)\} + 2\frac{6\varepsilon v_{\max}}{v_{\whp}} \\
&\hspace{-2in}\leq  \cdots \\
&\hspace{-2in}\leq  \left(\frac{v_{\whp}}{v_{\whp}+\Delta}\right)^{\ell+1} \Pr\{\Phi(\bolds{X}(\infty)) > B - 2v_{\whp}\} + \frac{6(\ell+1)\varepsilon v_{\max}}{v_{\whp}}.
\end{align*}
\end{proof}

\subsection{Existing Tail Bounds based on Moment Generating Function Do Not Supersede Lemma~\ref{lem:tail-bound} (Section~\ref{sec:stochastic})}\label{app:comp-hajek}
This section illustrates why the bound in \cite{hajek1982hitting} can still give a suboptimal bound based on the conditions in Lemma~\ref{lem:tail-bound}. To set the stage, fix a real-value Markov process $\{\bolds{Y}(t)\}_{t \geq 1}$ (view it as the process of the Lyapunov function). Recall that Lemma~\ref{lem:tail-bound} essentially requires three conditions for the drift $Y(t+1)-Y(t)$:
\begin{itemize}
\item[(i)] with probability $1-\varepsilon$, it changes by at most $v_{\whp}$; 
\item [(ii)] if $Y(t) \geq B$, then with probability $1 - \varepsilon$ the drift is at most $-\Delta$; and 
\item [(iii)] the drift is almost surely upper bounded by $v_{\max}$.
\end{itemize}
With these conditions, Lemma~\ref{lem:tail-bound} shows that 
\[
\Pr\{Y(\infty) > B + 2v_{\whp}\ell\} \leq \left(\frac{v_{\whp}}{v_{\whp}+\Delta}\right)^{\ell+1} + \frac{6(\ell+1)\varepsilon v_{\max}}{v_{\whp}}.
\]
In our setting we take $\varepsilon\approx 1/N^2$; $v_{\whp}, \Delta, B \approx \sqrt{N}; \ell \approx \ln N$; and $v_{\max} \approx N$ and show that the considered Lyapunov function is greater than $\tilde{O}(\sqrt{N})$ with probability at most $\tilde{O}(1/N)$.

We have argued why a tail bound that only relies on an absolute drift bound $v_{\max}$ (like the one in \cite{bertsimas2001performance}) leads to suboptimal guarantee (it would conclude that the Lyapunov function is greater than $O(v_{\max})=O(N)$ with low probability). However, there are bounds that focus on the moment generating function of the drift, such as the below result from Theorem 2.3 of \cite{hajek1982hitting}. It is a priori unclear whether this result can also give the desired $\tilde{O}(\sqrt{N})$ bound.
\begin{fact}\label{fact:hajek}
Suppose there exist $a,\eta,\rho,D$ such that (1) $\expect{e^{\eta (Y(t+1)-Y(t))}; Y(t) > a \mid Y(t)} \leq \rho$ and (2) $\expect{e^{\eta (Y(t+1)-a)}; Y(t) \leq a \mid Y(t)} \leq D$. Then $\Pr\{Y(t) \geq b\} \leq \rho^t e^{\eta(Y(0)-b)}+\frac{1-\rho^t}{1-\rho}De^{\eta(a-b)}$ for any $b$.
\end{fact}
To see why the above prior result does not give the desired bound, recall that we aim to show that the Lyapunov function is less than $\tilde{O}(\sqrt{N})$ with high probability in steady state. Therefore, in Fact~\ref{fact:hajek}, we need to take $a \approx \sqrt{N}, b \approx \sqrt{N}$ and let $t \to \infty$. Under the above conditions (i)-(iii), in the worst case $\expect{e^{\eta (Y(t+1)-Y(t))}; Y(t) > a \mid Y(t)}$ can be as large as $\varepsilon e^{\eta v_{\max}}$. For the bound in the fact to be meaningful, the constant $\rho$ must be less than one, which implies
\[
1 > \rho \geq \expect{e^{\eta (Y(t+1)-Y(t))}; Y(t) > a \mid Y(t)} \geq \varepsilon e^{\eta v_{\max}}.
\]
Since we want $\varepsilon = \tilde{O}(1/N)$, the constant $\eta$ should be set to at most $\ln(1/\varepsilon)/v_{\max} \approx \ln(N) / N$. In this case, Fact~\ref{fact:hajek} can at best give an upper bound
\begin{equation}
\Pr\{Y(\infty) \geq \sqrt{N}\ln N\} \leq \frac{D}{1-\rho}e^{-\eta \sqrt{N}} \lesssim e^{-\ln(N)\sqrt{N} / N} = \Theta(1),
\end{equation}
which is much looser than the $\tilde{O}(1/N)$ bound implied by Lemma~\ref{lem:tail-bound}.

\subsection{Proof of Lemma~\ref{lem:drift-bound} (Section~\ref{sec:stochastic})}\label{app:lem-drift-bound}
The proof requires an additional lemma that lower bounds the number of jobs that leave the system without getting service in each period due to abandonment. Specifically, for a state $i$, let $F_i(t)=Z_i(t) - \sum_{k \in \child(i)} Q_k(t+1)$ be the number of state-$i$ jobs transitioning into the $\perp$ state at the end of period $t$  (abandoning). Given a minimum abandonment probability $\theta > 0$ and a set of states $\set{X}$, the below concentration result shows that $\sum_{i \in \set{X}} F_i(t)$ is approximately lower bounded by $\theta \sum_{i \in \set{X}} Z_i(t).$
\begin{lemma}\label{lem:bound-abandon}
For any set of states $\set{X}$, period $t$ and $\delta \in (0,e^{-1}]$, 
\[
\Pr\left\{\sum_{i \in \set{X}} F_i(t) \geq \theta\sum_{i \in \set{X}} Z_i(t) - 
2\sqrt{N/\theta}\ln \frac{1}{\delta}\right\} \geq 1 - 2\delta^2.
\]
\end{lemma}
\begin{proof}
Letting $\theta_i = 1 - \sum_{k \in \child(i)} P(i,k)$ be the abandonment probability for a state-$i$ job, $F_i(t)$ follows a Binomial distribution with $Z_i(t)$ trials and success probability $\theta_i$. Moreover, $F_i(t), i\in \set{X},$ are independent because each job transitions according to the Markov chain independently. As a result, the sum $\sum_{i \in \set{X}} F_i(t)$ is the sum of $\sum_{i \in \set{X}} Z_i(t)$ independent random variables taking value in $\{0,1\}$. Conditioning on $\sum_{i \in \set{X}} Z_i(t) = n$, the conditional expectation of $\sum_{i \in \set{X}} F_i(t)$ is at least $n\theta$. Applying Hoeffding's Inequality (Fact~\ref{fact:hoeffding}) gives
\begin{equation}\label{eq:bound-abandon-cond}
\Pr\left\{\sum_{i \in \set{X}} F_i(t) \geq n\theta - \sqrt{n\ln \frac{1}{\delta}}\mid \sum_{i \in \set{X}} Z_i(t) = n\right\} \geq 1 - \exp\left(\frac{-2n\ln \frac{1}{\delta}}{n}\right) = 1 - \delta^2.
\end{equation}
Moreover, letting $M = \left\lfloor 4(N/\theta)\ln\frac{1}{\delta}\right\rfloor$, the sum $\sum_{i \in \set{X}} Z_i(t)$ is no larger than $M$ with high probability: 
\begin{align*}
\Pr\left\{\sum_{i \in \set{X}} Z_i(t) \leq M\right\} &\geq \Pr\left\{\sum_{i \in \set{X}} Q_i(t) \leq M\right\} \tag{by $Z_i(t) \leq Q_i(t)$}\\
&=  \Pr\left\{\sum_{i \in \set{X}} Q_i(t) \leq 4(N/\theta)\ln\frac{1}{\delta}\right\} \tag{by the fact that $Q_i(t)$ is an integer}\\
&\geq \Pr\left\{\sum_{i \in \set{X}} Q_i(t) \leq (N/\theta) + 3\sqrt{N/\theta}\ln\frac{1}{\delta}\right\} \geq 1 - \delta^2 \tag{by Lemma~\ref{lem:queue-upper-bound}}.
\end{align*}
Combining the above inequality with the law of total probability,
\begin{align*}
&\Pr\left\{\sum_{i \in \set{X}} F_i(t) \geq \theta\sum_{i \in \set{X}} Z_i(t) - 2\sqrt{N/\theta}\ln \frac{1}{\delta}\right\} \\
&\geq \sum_{n=0}^{M}\Pr\left\{\sum_{i \in \set{X}} F_i(t) \geq \theta\sum_{i \in \set{X}} Z_i(t) - 2\sqrt{N/\theta}\ln \frac{1}{\delta} \mid \sum_{i \in \set{X}} Z_i(t) = n\right\} \Pr\left\{\sum_{i \in \set{X}} Z_i(t) = n\right\} \\
&\geq \sum_{n=0}^{M} \Pr\left\{\sum_{i \in \set{X}} F_i(t) \geq n\theta - \sqrt{n\ln \frac{1}{\delta}} \mid \sum_{i \in \set{X}} Z_i(t) = n\right\} \Pr\left\{\sum_{i \in \set{X}} Z_i(t) = n\right\} \\
&\geq (1-\delta^2)\sum_{n=0}^{M} \Pr\left\{\sum_{i \in \set{X}} Z_i(t) = n\right\} \geq (1-\delta^2)(1-\delta^2) \geq 1-2\delta^2.
\end{align*}
\end{proof}

\begin{proof}[Proof of Lemma~\ref{lem:drift-bound}]
    Recall that $\set{T} = \TopS(\bo_{[m]})$ and thus $\sub(\set{T}) = \sub(\bo_{[m]}).$ We define the following good event $\set{G}$ of system states such that if $\bolds{X}(t) = (\bolds{Q}(t), R(t),\bolds{Q}(t+1)) \in \set{G}$, then
\begin{enumerate}[label=\textnormal{(G-\arabic*)}]
\item $\sum_{i \in \set{T}} Q_i(t+1) \leq \lambda N \sum_{i \in \set{T}} \pi(i) + 3\sqrt{N/\theta}\ln\frac{1}{\delta}$;\label{item:G-p2}
\item $\sum_{i \in \sub(\set{T})} f_i(\bolds{X}(t)) \geq \theta \sum_{i \in \sub(\set{T})} z_i(\bolds{X}(t)) - 2\sqrt{N/\theta}\ln\frac{1}{\delta}$, where we denote the number of abandoning state-$i$ jobs by $f_i(\bolds{X}(t)) = z_i(\bolds{X}(t)) - \sum_{k \in \child(i)} Q_k(t+1)$. \label{item:G-p3}
\end{enumerate}
We note that $\Pr\{\bolds{X}(\infty) \in \set{G}\} \geq 1 -3\delta^2.$ To see this, $\bolds{X}(\infty)$ satisfies \ref{item:G-p2} with probability at least $1 - \delta^2$ by Lemma~\ref{lem:queue-upper-bound} and it satisfies \ref{item:G-p3} with probability at least $1 - 2\delta^2$ by Lemma~\ref{lem:bound-abandon}. Using a union bound gives $\Pr\{\bolds{X}(\infty) \in \set{G}\} \geq 1 - 3\delta^2.$ We thus showed (i) of the lemma.

To show (ii), for any $t$, condition on $\set{X}(t)$ and define an event $\set{G}^{\text{service}}$ where the number of available services is not far below its mean: 
\begin{equation}\label{eq:def-Gr}
\set{G}^{\text{service}} = \left\{R(t+1) \geq \mu \cdot N - \sqrt{N\ln\frac{1}{\delta}}\right\}. 
\end{equation}
By Hoeffding's Inequality (Fact~\ref{fact:hoeffding}), $\Pr\{\set{G}^{\text{service}} | \bolds{X}(t)\} \geq 1 - \delta^2$ since $R(t+1)$ follows a Binomial distribution with mean $N\mu$ and is independent from $\bolds{X}(t)$.

Given $\bolds{X}(t) \in \set{G}$ and $\set{G}^{\text{service}}$, we show \eqref{eq:self-reflected} in Lemma~\ref{lem:drift-bound} by expanding the drift:
\begin{align}
\Phi(\bolds{X}(t+1)) - \Phi(\bolds{X}(t)) &= \sum_{i \in \sub(\set{T})} Z_i(t+1) - \sum_{i \in \sub(\set{T})} Z_i(t) \nonumber\\
&= \sum_{i \in \sub(\set{T})} (Q_i(t+1)-R_i(t+1))-\sum_{i \in \sub(\set{T})} Z_i(t) \nonumber\\
&\hspace{-0.5in}= \sum_{i \in \set{T}} Q_i(t+1) +\sum_{i \in \sub(\set{T}) \setminus \set{T}} Q_i(t+1)  - \sum_{i \in \sub(\set{T})} R_i(t+1) -  \sum_{i \in \sub(\set{T})} Z_i(t). \label{eq:drift-decompose}
\end{align}
By the rule of $\prio(\bo)$ in \eqref{eq:priority-service}, the $R(t+1)$ available services must first serve jobs with states in $\bo_{[m]}$ before they are used to serve jobs in $\set{S} \setminus \bo_{[m]}.$ As a result, there are two cases:
\begin{itemize}
\item all jobs in $\bo_{[m]}$ are served in period $t + 1$: in this case, $\sum_{i \in \set{T}} Q_i(t+1) \leq \sum_{i \in \sub(\set{T})} R_i(t+1)$ and thus \eqref{eq:drift-decompose} gives
\[
\Phi(\bolds{X}(t+1)) - \Phi(\bolds{X}(t)) \leq \sum_{i \in \sub(\set{T}) \setminus \set{T}} Q_i(t+1) -\sum_{i \in \sub(\set{T})} Z_i(t).
\]

\item some jobs in $\bo_{[m]}$ are not served in $t+1$: in this case $\sum_{i \in \sub(\set{T})} R_i(t+1) = R(t+1)$. This is because the jobs with states in $\sub(\set{T})$ must take all $R(t+1)$ available services (otherwise we can increase the $R_i(t+1)$ for some $i \in \bo_{[m]}$ given the priority ordering $\bo$). Putting this into \eqref{eq:drift-decompose} gives
\begin{align}
&\Phi(\bolds{X}(t+1)) - \Phi(\bolds{X}(t))  \leq \sum_{i \in \set{T}} Q_i(t+1) + \sum_{i \in \sub(\set{T}) \setminus \set{T}} Q_i(t+1) - R(t+1) - \sum_{i \in \sub(\set{T})} Z_i(t) \label{eq:drift-decompose-2}
\end{align}
Recall that, by the definition of $m$, $\lambda \sum_{i \in \set{T}} \pi(i) \leq \mu.$ Combining this bound with  Property \ref{item:G-p2} of event $\set{G}$ and the definition of event $\set{G}^{\text{service}}$ shows
\begin{equation}\label{eq:bound-q-r}
\sum_{i \in \set{T}} Q_i(t+1) - R(t+1) \leq N\lambda\sum_{i \in \set{T}} \pi(i) - N\mu + 3\sqrt{N/\theta}\ln\frac{1}{\delta} + \sqrt{N\ln\frac{1}{\delta}} \leq 4\sqrt{N/\theta}\ln\frac{1}{\delta}. 
\end{equation}
Putting \eqref{eq:bound-q-r} into \eqref{eq:drift-decompose-2} shows
\begin{equation}\label{eq:drift-dec-3}
\Phi(\bolds{X}(t+1))-\Phi(\bolds{X}(t)) \leq \sum_{i \in \sub(\set{T}) \setminus \set{T}} Q_i(t+1) - \sum_{i \in \sub(\set{T})} Z_i(t) + 4\sqrt{N/\theta}\ln\frac{1}{\delta}. 
\end{equation}
\end{itemize}
Summarizing the above two cases, \eqref{eq:drift-dec-3} always holds, given $\bolds{X}(t) \in \set{G}$ and $\set{G}^{\text{service}}$.
To upper bound the right hand side of \eqref{eq:drift-dec-3}, the set of parents of $\sub(\set{T}) \setminus \set{T}$ is a subset of $\sub(\set{T})$. Therefore,
\begin{align*}
\Phi(\bolds{X}(t+1))-\Phi(\bolds{X}(t)) &\leq \sum_{i \in \sub(\set{T}) \setminus \set{T}} Q_i(t+1)  - \sum_{i \in \sub(\set{T})} Z_i(t) + 4\sqrt{N/\theta}\ln\frac{1}{\delta} \\
&= \sum_{i \in \sub(\set{T})} \sum_{k \in \child(i)} Q_k(t+1) - \sum_{i \in \sub(\set{T})} Z_i(t) + 4\sqrt{N/\theta}\ln\frac{1}{\delta} \\
&\overset{(a)}{=} -\sum_{i \in \sub(\set{T})} F_i(t) + 4\sqrt{N/\theta}\ln\frac{1}{\delta}  \\
&\hspace{-1in}\overset{(b)}{\leq} -\theta \sum_{i \in \sub(\set{T})} Z_i(t) + 6\sqrt{N/\theta}\ln\frac{1}{\delta} =-\theta \Phi(\bolds{X}(t))+ 6\sqrt{N/\theta}\ln\frac{1}{\delta}.
\end{align*}
where $(a)$ is by the definition of the number of abandoning state-$i$ jobs $F_i(t)$ and $(b)$ is by Property \ref{item:G-p3} of event $\set{G}$. As a result, \eqref{eq:self-reflected} holds true. 

Lastly, since the parents of all states in $\sub(\set{T}) \setminus \set{T}$ must be in $\sub(\set{T})$, $\sum_{i \in \sub(\set{T}) \setminus \set{T}} Q_i(t+1) \leq \sum_{i \in \sub(\set{T})} Z_i(t)$. As a result, $\Phi(\bolds{X}(t+1)) - \Phi(\bolds{X}(t)) \leq \sum_{i \in \set{T}} Q_{i}(t+1) \leq NL$ by \eqref{eq:drift-decompose} and Lemma~\ref{lem:queue-upper-bound}. Therefore, the drift has an almost sure upper bound $NL$.
\end{proof}

\subsection{Proof of Lemma~\ref{lem:iterative-ssc} (Section~\ref{sec:stochastic})}\label{app:lem-iterative-ssc}

\begin{proof}[Proof of Lemma~\ref{lem:iterative-ssc}]
The proof is by applying Lemmas~\ref{lem:tail-bound} and \ref{lem:drift-bound}. Lemma~\ref{lem:drift-bound} shows that there exists a good event $\set{G}$ with $\Pr\{\bolds{X}(\infty) \in \set{G}\} \geq 1 -3\delta^2$, such that conditioning on $\bolds{X}(t) \in \set{G}$, \eqref{eq:self-reflected} holds with probability at least $1 - \delta^2$. Take $\varepsilon = 3 \delta^2, \Delta = v_{\whp} = 6\sqrt{N/\theta}\ln\frac{1}{\delta}, B = \frac{2\Delta}{\theta}, v_{\max} = NL.$ If $v_{\whp} > v_{\max}$, which implies $36\theta^{-1}\ln^2\frac{1}{\delta} > NL^2$, the event in the lemma satisfies 
\[
\Pr\left\{\sum_{i \in \sub(o_{[m]})} Z_i(\infty) > 48\theta^{-1.5}\sqrt{N}\ln^2\frac{1}{\delta}\right\} \leq \Pr\left\{\sum_{i \in \sub(o_{[m]})} Z_i(\infty) > N^{1.5}L^2\right\} = 0,
\]
where the first inequality is by $48\theta^{-1.5}\ln^2\frac{1}{\delta} > 36\theta^{-1}\ln^2\frac{1}{\delta} > NL^2$ and the second inequality is because in every period, there are at most $NL$ jobs in the system. Therefore, when $v_{\whp} > v_{\max}$, Lemma~\ref{lem:iterative-ssc} holds true. The remainder of the proof thus assumes $v_{\whp} \leq v_{\max}$.

Conditions of Lemma~\ref{lem:tail-bound} are met because $v_{\whp} \leq v_{\max}$, $\Pr\{\bolds{X}(\infty) \in \set{G}\} \geq 1 - \varepsilon, \expect{\Phi(\bolds{X}(\infty))} < +\infty,$ and 
\begin{enumerate}
    \item [(i)] if $\bolds{X}(t) \in \set{G}$, with probability at least $1 - \delta^2 \geq 1 - \varepsilon$ conditioning on $\bolds{X}(t)$, \eqref{eq:self-reflected} holds and implies that $\Phi(\bolds{X}(t+1)) - \Phi(\bolds{X}(t)) \leq 6\sqrt{N/\theta}\ln\frac{1}{\delta} = v_{\whp}$;
    \item [(ii)] 
    if $\bolds{X}(t) \in \set{G}$ and $\Phi(\bolds{X}(t)) \geq B$, with probability at least $1 - \delta^2 \geq 1 - \varepsilon$ conditioning on $\bolds{X}(t)$, \eqref{eq:self-reflected} gives
        $\Phi(\bolds{X}(t+1)) - \Phi(\bolds{X}(t)) \leq -\theta \Phi(\bolds{X}(t)) + 6\sqrt{N/\theta}\ln \frac{1}{\delta} \leq -\theta \cdot B +  \Delta = -\Delta.$
    \item [(iii)] $\Phi(\bolds{X}(t+1)) - \Phi(\bolds{X}(t)) \leq v_{\max} = NL$ by the last result of Lemma~\ref{lem:drift-bound}.
\end{enumerate}
As a result, Lemma~\ref{lem:tail-bound} applies. Taking $\ell = \lfloor 3 \ln \frac{1}{\delta}\rfloor$ (thus $\ell + 1 \geq 3\ln \frac{1}{\delta}$) in Lemma~\ref{lem:tail-bound} gives
\begin{align*}
\Pr\left\{\Phi(\bolds{X}(\infty)) > B + 6v_{\whp}\ln\frac{1}{\delta}\right\} \leq 2^{-(\ell + 1)} + 6(\ell + 1)\varepsilon v_{\max} &\leq \delta^2 + 6\left(3\ln \frac{1}{\delta} + 1\right)3\delta^2NL\\
&\hspace{-0.8in}\leq (1 + 6\cdot (3 + 1) \cdot 3) \delta^2NL \ln\frac{1}{\delta} = 73NL\delta^2\ln\frac{1}{\delta},
\end{align*}
where we use the assumption that $\delta \leq e^{-1}.$ Moreover, 
\[
B + 6v_{\whp}\ln \frac{1}{\delta} \leq \left(\frac{2}{\theta} + 6\ln \frac{1}{\delta}\right)6\sqrt{N/\theta}\ln\frac{1}{\delta} \leq 48\theta^{-1.5}\sqrt{N}\ln^2\frac{1}{\delta},
\]
which finishes the proof.
\end{proof}

\subsection{Proof of Lemma~\ref{lem:connect-partial} (Section~\ref{sec:stochastic})}\label{app:lem-connect-partial}
Recall the degeneracy parameter $\kappa = 1 - \sum_{i \in \TopS(\bo_{[m]}) \cap \sub(\partialNode)} \pi(i) / \pi(\partialNode)$ in \eqref{eq:def-degeneracy}. We first lower bound it by the minimum abandonment probability $\theta$.
\begin{lemma}\label{lem:degen-aban}
The minimum abandonment probability $\theta$ lower bounds the degeneracy parameter $\kappa$.
\end{lemma}
\begin{proof}
Recall that $\set{T}_m = \TopS(\bo_{[m]})$ and thus $\kappa = 1 - \sum_{i \in \set{T}_m \cap \sub(\partialNode)} \pi(i) / \pi(\partialNode)$ from \eqref{eq:def-degeneracy}. Showing $\kappa \geq \theta$ is equivalent to showing $(1-\theta)\pi(\partialNode) \geq \sum_{i \in \set{T}_m \cap \sub(\partialNode)} \pi(i).$ 

To see why this is true, define the set $\set{X} = \sub(\child(\partialNode))$, which is the subtree of children of $\partialNode.$ Then $\TopS(\set{X})$ consists of the children of $\partialNode$, i.e., $\TopS(\set{X}) = \child(\partialNode).$ In addition, define $\set{Y} = \set{T}_m \cap \sub(\partialNode)$. The top set of $\set{Y}$ is itself since it is a subset of $\set{T}_m$, which is itself a top set of $\bo_{[m]}.$ As a result,
\[
\sum_{i \in \set{T}_m \cap \sub(\partialNode)} \pi(i) = \sum_{i \in \set{Y}} \pi(i) = \sum_{i \in \TopS(\set{Y})} \pi(i) \overset{(a)}{\leq} \sum_{i \in \TopS(\set{X})} \pi(i) = \sum_{i \in \child(\partialNode)} \pi(i) \overset{(b)}{\leq} (1-\theta)\pi(\partialNode),
\]
where inequality (a) is because $\set{Y}$ is a subset of $\set{X}$ and Lemma~\ref{lem:top-prob} applies; inequality (b) is because $\sum_{i \in \child(\partialNode)}\pi(i) = \sum_{i \in \child(\partialNode)} \pi(\partialNode)P(\partialNode,i) = \pi(\partialNode)\sum_{i \in \child(\partialNode)}P(\partialNode,i) \leq \pi(\partialNode)(1-\theta)$ as $\theta$ is the minimum abandonment probability. We thus finish the proof by establishing an equivalence of $\kappa \geq \theta.$
\end{proof}

\begin{proof}[Proof of Lemma~\ref{lem:connect-partial}.]
If $\partialNode = \perp$ the result trivially holds; we thus suppose $\partialNode \neq \perp$, which by definition is equal to $o_{m+1}.$ The limiting distribution $Z_{\partialNode}(\infty)$ satisfies $Z_{\partialNode}(\infty) = Q_{\partialNode}(\infty) - R_{\partialNode}(\infty)$ with 
\[
R_{\partialNode}(\infty) = \min\left\{Q_{\partialNode}(\infty), \left(R(\infty) - \sum_{i \in \bo_{[m]}} Q_i(\infty)\right)^+\right\}.
\]
Therefore, $Z_{\partialNode}(\infty) = \max\left(0, Q_{\partialNode}(\infty) - \left(R(\infty) - \sum_{i \in \bo_{[m]}} Q_i(\infty)\right)^+\right)$, implying
\begin{align}
Z_{\partialNode}(\infty) &\leq \max\left(0, Q_{\partialNode}(\infty) - R(\infty) + \sum_{i \in \bo_{[m]}} Q_i(\infty)\right) \nonumber \\
&\hspace{-0.7in}= \max\left(0, Q_{\partialNode}(\infty)  + \sum_{i \in \set{T}_m \setminus \sub(\partialNode)} Q_i(\infty) + \sum_{i \in \set{T}_m \cap \sub(\partialNode)} Q_i(\infty) + \sum_{i \in \bo_{[m]} \setminus \set{T}_m} Q_i(\infty) - R(\infty)\right).\label{eq:bound-zp}
\end{align}
Define a good event $\set{G}_z$ with the following three properties:
\begin{enumerate}[label=\textnormal{(Gz-\arabic*)}]
\item the available service is not far below its mean, i.e.,  $R(\infty) \geq N\mu - \sqrt{N\ln N}$; \label{item:gz-p1}
\item the number of  jobs with states in $\set{T}_m \setminus \sub(\partialNode)$ is not high, i.e., \[\sum_{i \in \set{T}_m \setminus \sub(\partialNode) \cup \{\partialNode\}} Q_i(\infty) \leq N\lambda \sum_{i \in \set{T}_m \setminus \sub(\partialNode) \cup \{\partialNode\}} \pi(i) + 3\sqrt{N/\theta}\ln N.\] \label{item:gz-p2}
\item there are few jobs with states in $\sub(\set{T}_m) \setminus \set{T}_m$, i.e., $\sum_{i \in \sub(\set{T}_m) \setminus \set{T}_m} Q_i(\infty) \leq \tilde{U}\sqrt{N}.$
\end{enumerate}
Each of the three properties happens with probability at least $1 - 1 / N$: the first is by Hoeffding's Inequality in Fact~\ref{fact:hoeffding}; the second is by Lemma~\ref{lem:queue-upper-bound}; and the third is by the lemma assumption and the fact that $\sum_{i \in \sub(\bo_{[m]})} Z_i(\infty)$  stochastically dominates $\sum_{i \in \sub(\set{T}_m) \setminus \set{T}_m} Q_i(\infty)$. Applying a union bound gives $\Pr\{\set{G}_z\} \geq 1 - 3 / N.$ Conditioning on $\set{G}_z$, \eqref{eq:bound-zp} gives
\begin{align*}
Z_{\partialNode}(\infty) &\leq \max\left(0, N\lambda \sum_{i \in \set{T}_m \setminus \sub(\partialNode)\cup \{\partialNode\}} \pi(i) + \sum_{i \in \set{T}_m \cap \sub(\partialNode)} Q_i(\infty)+4\sqrt{N/\theta}\ln N + \tilde{U}\sqrt{N} - N\mu\right) \\
&\leq N\lambda \sum_{i \in \set{T}_m \setminus \sub(\partialNode)\cup \{\partialNode\}} \pi(i) - N\mu + \sum_{i \in \set{T}_m \cap \sub(\partialNode)} Q_i(\infty)+4\sqrt{N/\theta}\ln N + \tilde{U}\sqrt{N}, 
\end{align*}
where the second inequality is because $\lambda \sum_{i \in \set{T}_m \setminus \sub(\partialNode)\cup \{\partialNode\}} \pi(i) > \mu$. Otherwise we can increase $m$ while keeping $\lambda \sum_{i \in \set{T}_m} \pi(i) \leq \mu$, which contradicts the definition that $m$ is the maximum such position. Since $Z_{\partialNode} \leq N$ almost surely and $\Pr\{\set{G}_z\} \geq 1 - 3 / N$, the above inequality implies
\begin{equation}\label{eq:bound-ex-z}
\expect{Z_{\partialNode}(\infty)} \leq N\lambda \sum_{i \in \set{T}_m \setminus \sub(\partialNode)\cup \{\partialNode\}} \pi(i) - N\mu + \expect{\sum_{i \in \set{T}_m \cap \sub(\partialNode)} Q_i(\infty)}+4\sqrt{N/\theta}\ln N + \tilde{U}\sqrt{N} + 3.
\end{equation}
For any state $i$ in the subtree of $\sub(\partialNode)$, its steady-state number of waiting jobs must satisfy $\expect{Q_i(\infty)} \leq \expect{Z_{\partialNode}(\infty)}\frac{\pi(i)}{\pi(\partialNode)}$ because (i) to be in state $i$ a job must first be in state $\partialNode$; (ii) the conditional probability of a state-$\partialNode$ job becoming a state-$i$ job is $\pi(i)/\pi(\partialNode)$ assuming the job does not leave because of getting service. As a result, $\expect{\sum_{i \in \set{T}_m \cap \sub(\partialNode)} Q_i(\infty)} \leq \expect{Z_{\partialNode}(\infty)}\sum_{i \in \set{T}_m \cap \sub(\partialNode)} \pi(i) / \pi(\partialNode).$ Putting it into \eqref{eq:bound-ex-z} gives
\[
\expect{Z_{\partialNode}(\infty)} \leq N\lambda \sum_{i \in \set{T}_m \setminus \sub(\partialNode)\cup \{\partialNode\}} \pi(i) - N\mu + \expect{Z_{\partialNode}(\infty)}\sum_{i \in \set{T}_m \cap \sub(\partialNode)} \pi(i) / \pi(\partialNode) +4\sqrt{N/\theta}\ln N + \tilde{U}\sqrt{N} + 3.
\]
Rearranging the term and recalling $\kappa = 1 - \sum_{i \in \set{T}_m \cap \sub(\partialNode)} \pi(i) / \pi(\partialNode)$, we obtain \[
\expect{Z_{\partialNode}(\infty)} \leq N\cdot \boxed{\frac{\lambda \sum_{i \in \set{T}_m \setminus \sub(\partialNode)\cup \{\partialNode\}} \pi(i) - \mu}{\kappa}} + \frac{4\sqrt{N/\theta}\ln N + \tilde{U}\sqrt{N} + 3}{\kappa}.
\]
For the boxed term, observe that 
\begin{align*}
z_{\partialNode}^{\bo} = \lambda \pi(\partialNode) - \nu_{\partialNode}^{\bo} = \lambda \pi(\partialNode) - \frac{\mu - \lambda \sum_{i \in \set{T}_m} \pi(i)}{\kappa} &= \frac{\lambda \pi(\partialNode) - \lambda \sum_{i \in \set{T}_m \cap \sub(\partialNode)}\pi(i)+\lambda\sum_{i \in \set{T}_m} \pi(i) - \mu}{\kappa} \\
&\hspace{-1in}= \frac{\lambda \pi(\partialNode) + \lambda \sum_{i \in \set{T}_m \setminus \sub(\set{\partialNode})} \pi(i)-\mu}{\kappa} = \boxed{\frac{\lambda \sum_{i \in \set{T}_m \setminus \sub(\set{\partialNode})\cup\{\partialNode\}} \pi(i)-\mu}{\kappa}}.
\end{align*}
We thus conclude $\expect{Z_{\partialNode}(\infty)} \leq Nz_{\partialNode}^{\bo} + \frac{4\sqrt{N/\theta}\ln N + \tilde{U}\sqrt{N} + 3}{\kappa} \leq Nz_{\partialNode}^{\bo} + \frac{4\sqrt{N/\theta}\ln N + \tilde{U}\sqrt{N} + 3}{\theta}$ by Lemma~\ref{lem:degen-aban}.
\end{proof}

\subsection{Proof of Lemma~\ref{lem:ssc-cost-diff} (Section~\ref{sec:stochastic})}\label{app:lem-ssc-cost-diff}

We first give a result that bounds the second moment of the number of jobs in the system.
\begin{lemma}\label{lem:2nd-moment}
For any period $t$, the second moment of $\sum_{i \in \set{S}} Q_i(t)$ is at most $N/\theta+(N/\theta)^2$.
\end{lemma}
\begin{proof}
For $0 \leq \ell < L$, define a binomial random variable $X_{\ell} = \textsc{Bin}(N,(1-\theta)^{\ell})$.  These random variables are independent. Note that a waiting job in period $t$ must be an arrival in the last $L$ periods (the Markov chain has a bounded depth $L$) and the probability that an arrival in period $t'$ does not abandon the system until period $t$ is at most $(1-\theta)^{t-t'-1}$. Moreover, the number of job arrivals in any period is at most $N$. Therefore, the total number of jobs in period $t$, $\sum_{i \in \set{S}} Q_i(t)$, is stochastically dominated by the sum $\sum_{\ell=0}^{L-1} X_{\ell}.$ The mean of the sum is at most $N / \theta$ and the variance of it is also at most $N / \theta$. As a result,
\[
\expect{\left(\sum_{i \in \set{S}} Q_i(t)\right)^2} \leq \expect{\left(\sum_{\ell=0}^{L-1} X_{\ell}\right)^2} = \expect{\sum_{\ell=0}^{L-1} X_{\ell}}^2 + \var\left(\sum_{\ell=0}^{L-1} X_{\ell}\right) \leq (N/\theta)^2 + N/\theta.
\]
\end{proof}

\begin{proof}[Proof of Lemma~\ref{lem:ssc-cost-diff}]
By its definition \eqref{eq:def-cost}, the long-run average holding cost of $\prio(\bo)$ is 
\begin{align*}
C(N,\prio(\bo)) &= \lim \sup_{T \to \infty} \frac{1}{T}\expect{\sum_{t=1}^T \sum_{j \in \set{Q}(t) \setminus \set{R}(t)} c(S_j(t))} = \lim \sup_{T \to \infty} \frac{1}{T} \expect{\sum_{t=1}^T \sum_{i \in \set{S}} c(i)Z_i(t)} \\
&\hspace{0.5in}= \sum_{i \in \set{S}} c(i)\left(\lim\sup_{T \to \infty}\expect{\frac{1}{T}\sum_{t=1}^T Z_i(t)}\right) = \sum_{i \in \set{S}} c(i)\expect{Z_i(\infty)}.
\end{align*}
The last equation is because $Z_i(t)$ converges in distribution to $Z_i(\infty)$ and is bounded ($0 \leq Z_i(t) \leq Q_i(t)$ and $Q_i(t) \leq N$ by Lemma~\ref{lem:queue-upper-bound}). Therefore, $C(N,\prio(\bo)) - N\sum_{i \in \set{S}} c(i) z_i^{\bo}$ is equal to 
\begin{align*}
C(N,\prio(\bo)) - N\sum_{i \in \set{S}} c(i)z_i^{\bo}&=\underbrace{\sum_{i \in \sub(\bo_{[m]})}c(i)\expect{Z_i(\infty) - Nz_i^{\bo}}}_{\text{fully served}} + \underbrace{\sum_{i \in \sub(\partialNode) \setminus \sub(\bo_{[m]})} c(i)\expect{Z_i(\infty)-Nz_i^{\bo}}}_{\text{partially served}} \\
&\hspace{0.5in}+ \underbrace{\sum_{i \in \set{S} \setminus \sub(\bo_{[m]}) \setminus \sub(\partialNode)} c(i)\expect{Z_i(\infty) - Nz_i^{\bo}}}_{\text{never served}}.
\end{align*}
The proof bounds each of the three terms. For the first term, since $z_i^{\bo} \geq 0$ by definition, it suffices to bound $\sum_{i \in \sub(\bo_{[m]})} \expect{Z_i(\infty)}$. To bound this quantity, denote $X \coloneqq \sum_{i \in \sub(\bo_{[m]})} Z_i(\infty)$ and define the good event $\set{G}$ in which $X \leq \tilde{U}\sqrt{N}$, where $\tilde{U}$ is the constant of the statement of Lemma~\ref{lem:ssc-cost-diff}. The lemma assumption gives $\Pr\{\set{G}\} \geq 1 - 1 / N$. Moreover, 
\[
\expect{X^2} \leq \expect{\left(\sum_{i \in \sub(\bo_{[m]})} Q_i(\infty)\right)^2} \leq \expect{\left(\sum_{i \in \set{S}} Q_i(\infty)\right)^2} \leq (N/\theta)^2 + N/\theta,
\]
where the first inequality is by $Z_i(\infty) \leq Q_i(\infty)$ for any state $i$; the second inequality is because $\set{S}$ is the set of all possible states; and the third inequality applies Lemma~\ref{lem:2nd-moment} for $t \to \infty$. Therefore,
\begin{align}
(\text{fully served}) \leq c_{\max} \expect{X} &= c_{\max}\left(\expect{X\indic{\set{G}}} + \expect{X\indic{\set{G}^c}}\leq \tilde{U}\sqrt{N} + \expect{X\indic{\set{G}^c}}\right) \nonumber\\
&\leq c_{\max}\left(\tilde{U}\sqrt{N} + \sqrt{\expect{X^2}\expect{\indic{\set{G}^c}^2}}\right) \nonumber\\
&= c_{\max}\left(\tilde{U}\sqrt{N}+\sqrt{\expect{X^2}\Pr\{\set{G}^c\}} \right) \leq c_{\max}\left(\tilde{U}\sqrt{N} + \sqrt{2N}/\theta\right), 
\end{align}
where the second inequality is because of Cauchy-Schwarz inequality.

For the second term, an induction on the tree shows that for any state $k$ and $i$ with $i \in \sub(k),$ the number of remaining jobs with state $i$ in period $t$ is at most $\expect{Z_i(t)} \leq \expect{Z_k(t - d(k,i))}\pi(i) / \pi(k)$, where $d(k,i)$ is the tree distance from state $k$ to state $i$. As a result, $\expect{Z_i(\infty)} \leq \expect{Z_k(\infty)}\pi(i) / \pi(k)$ by taking $t \to \infty.$ Moreover, as discussed in the proof sketch, any state $i \in \sub(\partialNode) \setminus \sub(\bo_{[m]})$ is either the partially-served state \ref{item:partially-served} or a partially-reduced state \ref{item:partially-reduced}. This implies $z_i^{\bo} = z_{\partialNode}^{\bo}\pi(i) / \pi(\partialNode)$ by our construction. Using this observation, the second term satisfies 
\begin{align*}
\sum_{i \in \sub(\partialNode) \setminus \sub(\bo_{[m]})} c(i)\expect{Z_i(\infty)-Nz_i^{\bo}} 
&= \sum_{i \in \sub(\partialNode) \setminus \sub(\bo_{[m]})} c(i)\left(\expect{Z_i(\infty)}-Nz_{\partialNode}^{\bo}\pi(i)/\pi(\partialNode)\right) \\
&\leq \sum_{i \in \sub(\partialNode) \setminus \sub(\bo_{[m]})} c(i)\left(\expect{Z_{\partialNode}(\infty)}\pi(i)/\pi(\partialNode)-Nz_{\partialNode}^{\bo}\pi(i)/\pi(\partialNode)\right) \\
&= (\expect{Z_{\partialNode}(\infty)} - Nz_{\partialNode}^{\bo})\sum_{i \in \sub(\partialNode) \setminus \sub(\bo_{[m]})}c(i)\pi(i) / \pi(\partialNode) \\
&\hspace{-1in}\overset{(a)}{\leq} (\expect{Z_{\partialNode}(\infty)} - Nz_{\partialNode}^{\bo})^+ c_{\max}/\theta \overset{(b)}{\leq} \frac{c_{\max}}{\theta^2}\left(4\sqrt{N/\theta}\ln N + \tilde{U}\sqrt{N}+ 3\right).
\end{align*}
where (a) uses that $\sum_{i \in \sub(\partialNode)} \pi(i) = \sum_{\ell=0}^{L-1}\sum_{i \in \sub(\partialNode) \cap \set{S}_{\ell}} \pi(i) \leq \sum_{\ell=0}^{L-1}\pi(\partialNode)(1-\theta)^{\ell} \leq \pi(\partialNode)/\theta$ and (b) uses Lemma~\ref{lem:connect-partial}.

For the third term, observe that for any state $i$ and period $t$, $\expect{Z_i(t)} \leq N\lambda \pi(i)$ by induction from its parent state. Therefore, $\expect{Z_i(\infty)} \leq N\lambda \pi(i)$ by taking $t \to \infty.$ Since any state $i \in \set{S} \setminus \sub(\bo_{[m]}) \setminus \sub(\partialNode)$ is an un-reduced state \ref{item:un-reduced} with $z_i^{\bo} = \lambda \pi(i)$, the third term is at most zero. 

Combining the above bounds for the three terms gives
\begin{align*}
C(N, \prio(\bo)) - N\sum_{i \in \set{S}} c(i)z_i^{\bo} &\leq c_{\max}\left(\tilde{U}\sqrt{N}+\sqrt{2N}/\theta+\frac{1}{\theta^2}\left(4\sqrt{N/\theta}\ln N + \tilde{U}\sqrt{N}+3\right)\right) \\
&\leq 10c_{\max} \theta^{-2}\left(\sqrt{N/\theta}\ln N+\tilde{U}\sqrt{N}\right).
\end{align*}
\end{proof}

\subsection{Proof of Lemma~\ref{lem:stochastic} (Sections~\ref{sec:proof-pafou} and \ref{sec:stochastic})}\label{app:lem-stochastic}

To give the proof of Lemma~\ref{lem:stochastic}, we first simplify the bound in Lemma~\ref{lem:iterative-ssc} by tuning the parameter~$\delta$. 
\begin{lemma}\label{lem:ssc-simplified}
With probability at least $1 - 1 / N$, the number of remaining jobs in $\sub(\bo_{[m]})$ is upper bounded by $192\theta^{-1.5}\sqrt{N}\ln^2(73NL).$
\end{lemma}
\begin{proof}
Lemma~\ref{lem:iterative-ssc} shows that for any $\delta \in (0,e^{-1}],$ $\Pr\{\sum_{i \in \sub(\bo_{[m]})} Z_i(\infty) > 48\theta^{-1.5}\sqrt{N}\ln^2\frac{1}{\delta}\} \leq 73NL\delta^2\ln\frac{1}{\delta}$. To highlight the dependence on $\delta$, we define 
\begin{equation}\label{eq:def-u-epsilon-delta}
U(\delta) =  48\theta^{-1.5}\ln^2\frac{1}{\delta}, \quad \varepsilon(\delta) = 73NL\delta^2\ln\frac{1}{\delta}.
\end{equation}
We aim to pick a $\delta^\star \in (0,e^{-1}]$ such that $\varepsilon(\delta^\star) \leq 1 / N$ and then bound $U(\delta^\star).$ Lemma~\ref{lem:iterative-ssc} then implies $\Pr\{\sum_{i \in \sub(\bo_{[m]})} Z_i(\infty) > U(\delta^\star)\sqrt{N}\} \leq 1 / N$.

Set $\delta^\star$ to be the maximum value in $(0,1/e]$ such that 
\begin{equation}\label{eq:def-delta}
\left(\ln(73NL)+\ln\ln\frac{1}{\delta^\star}\right) = \ln\frac{1}{\delta^\star}.
\end{equation}
Such a $\delta^\star$ must exist because (1) both the left and right hand sides of the equation are continuous function when $\delta \in (0,1)$; (2) when $\delta \to 0^+$, the right hand side grows to infinity faster than the left hand side; (3) when $\delta = 1/e$, the left hand side is $\ln(73NL)$, which is greater than the right hand side with value $1$.

We next show $\varepsilon(\delta^\star) \leq 1 / N.$ Taking exponentiation of both sides of \eqref{eq:def-delta} gives
\[
73NL\ln\frac{1}{\delta^\star}=\frac{1}{\delta^\star} \Rightarrow \varepsilon(\delta^\star) = \delta^\star \Rightarrow \varepsilon(\delta^\star) = \frac{1}{73NL\ln\frac{1}{\delta}^\star} \leq \frac{1}{N}. 
\]
Upper bounding $U(\delta^\star)$ is less immediate as it involves a term $\ln^2 \frac{1}{\delta^\star}.$ We define two auxiliary quantities, $\delta^0 = 1 / \exp(2\ln(73NL))$ and $\delta^1 = 1 / \exp(\exp(4))$ and consider three cases:
\begin{itemize}
\item if $\delta^\star \geq \delta^0$, then $\ln\frac{1}{\delta^\star} \leq 2\ln(73NL).$
\item if $\delta^\star \geq \delta^1$, then $\ln\ln\frac{1}{\delta^\star} \leq \ln\ln\frac{1}{\delta^1}$ and \eqref{eq:def-delta} gives $\ln\frac{1}{\delta^\star}  \leq \ln(73NL)+4 \leq 2\ln(73NL).$
\item It is impossible to see $\delta^\star < \min(\delta^0,\delta^1)$. This is because if $\delta^\star < \delta^1$, then $\ln \frac{1}{\delta^\star} > \exp(4)$ and thus $\ln\ln\frac{1}{\delta^\star} \leq \frac{1}{2}\ln\frac{1}{\delta^\star}$ by Fact~\ref{fact:ln-order}. As a result,
\[
\ln(73NL)+\ln\ln\frac{1}{\delta^\star}\leq \ln(73NL) + \frac{1}{2}\ln\frac{1}{\delta^\star} \overset{(a)}{<} \frac{1}{2}\ln\frac{1}{\delta^\star} + \frac{1}{2}\ln\frac{1}{\delta^\star} = \ln\frac{1}{\delta^\star},
\]
which contradicts \eqref{eq:def-delta}. The strict  inequality (a) is by the case assumption that $\delta^\star < \delta^0$.
\end{itemize}
Summarizing the above discussion gives $\ln\frac{1}{\delta^\star} \leq 2\ln(73NL)$ and thus using \eqref{eq:def-u-epsilon-delta},
\[
U(\delta^\star) \leq 192\theta^{-1.5}\ln^2(73NL),
\]
which finishes the proof.
\end{proof}
\begin{proof}[Proof of Lemma~\ref{lem:stochastic}]
By Lemma~\ref{lem:ssc-simplified}, \[\Pr\left\{\sum_{i \in \sub(\bo_{[m]})} Z_i(\infty) \leq 192\theta^{-1.5}\sqrt{N}\ln^2(73NL)\right\} \geq 1-1/N.\] As a result of Lemma~\ref{lem:ssc-cost-diff},
\begin{align*}
C(N,\prio(\bo)) - N\sum_{i \in \set{S}} c(i)z_i^{\bo} &\leq 10c_{\max}\theta^{-2}(\sqrt{N/\theta}\ln N + 192\theta^{-1.5}\sqrt{N}\ln^2(73NL))\\
&\leq 1930c_{\max}\theta^{-3.5}\ln^2(73NL)\sqrt{N}.
\end{align*}
\end{proof}

\section{Robustness Guarantee for Inexact Estimations}\label{app:robustalg}
% !TEX root = main.tex
This section studies the robustness property of $\alg$ when its indices can be inaccurate due to estimation errors of the value functions. Suppose that instead of the ideal indices $\{\ind_{\alg}(i)\}_{i \in \set{S}}$, the algorithm operates with a list of estimated indices $\ind_{\esti} = \{\ind_{\esti}(i)\}_{i \in \set{S}}$ (for example, based on value functions estimated by some reinforcement learning algorithm as our practical algorithm $\hindalg$ in Section~\ref{sec:prac-imp}). These indices define a priority ordering $\bo^{\esti}$ for job states (assuming a consistent tie-breaking rule like in $\alg$). Our result is that the priority scheduling algorithm based on $\bo^{\esti}$ has a suboptimality gap that depends linearly on the error of $\ind_{\esti}$.
\begin{proposition}\label{prop:robustalg}
Letting $\varepsilon \coloneqq \max_{i \in \set{S}} |\ind_{\esti}(i) - \ind_{\alg}(i)|$,  algorithm $\prio(\bo^{\esti})$ satisfies
\[
\subopt(N,\prio(\bo^{\esti})) \leq (4\varepsilon/\theta) N + U(\theta,\ln(NL))\sqrt{N},~\text{
where $U(\theta,\ln(NL))$ is as in Theorem~\ref{thm:pafou}.}
\]
\end{proposition}
The remainder of this section proves Proposition~\ref{prop:robustalg}. Fix the estimated indices and denote the error $\max_{i \in \set{S}} |\ind_{\esti}(i) - \ind_{\alg}(i)|$ by $\varepsilon$. Let $(\bolds{q}^{\esti}, \bolds{\nu}^{\esti})$ denote the equilibrium of the priority algorithm $\prio(\bo^{\esti})$ obtained by the water-filling procedure in Section~\ref{sec:fluid-sol}. The key part of the proof is to show that this equilibrium has optimal fluid cost except for an extra term depending on the estimation error (an analogue of Lemma~\ref{lem:optimal-fluid}). We defer its proof to Section~\ref{sec:lem-fluid-error}.
\begin{lemma}\label{lem:fluid-error}
The equilibrium of $\prio(\bo^{\esti})$ has fluid cost $\sum_{i \in \set{S}} c(i)(q_i^{\esti} - \nu_i^{\esti}) \leq C^\star + 4\varepsilon / \theta$.
\end{lemma}
\begin{proof}[Proof of Proposition~\ref{prop:robustalg}]
Since $\prio(\bo^{\esti})$ is a priority algorithm, for any system size $N$,
\begin{align*}
\subopt(N, \prio(\bo^{\esti})) &= C(N, \prio(\bo^{\esti})) - \inf_{\nalg \in \Pi} C(N, \nalg) \\
&\overset{\text{Lemma~\ref{lem:lower-bound}}}{\leq } C(N,\prio(\bo^{\esti})) - NC^\star \\
&\overset{\text{Lemma~\ref{lem:stochastic}}}{\leq } N\sum_{i \in \set{S}} c(i) (q_i^{\esti} - \nu_i^{\esti}) + U(\theta,\ln(NL))\sqrt{N}- NC^\star \\
&\overset{\text{Lemma~\ref{lem:fluid-error}}}{\leq} N(C^\star + 4\varepsilon/\theta) + U(\theta,\ln(NL))\sqrt{N}- NC^\star\\ 
&= 4(\varepsilon/\theta)N + U(\theta,\ln(NL))\sqrt{N}.
\end{align*}
\end{proof}

\subsection{Bounding the fluid suboptimality (Lemma~\ref{lem:fluid-error})}\label{sec:lem-fluid-error}
For ease of notation we use $\bo$ to denote $\bo^{\esti}.$ As in Section~\ref{sec:fluid-sol}, we denote by $\bo_k$ the rank-$k$ state in $\bo$ and by $\bo_{[k]}$ the first $k$ states in $\bo$. Similar to how we prove the fluid optimality in Appendix \ref{app:prio-optimal}, the proof bounds the suboptimality of $(\bolds{q}^{\esti},\bolds{v}^{\esti})$ by connecting it with the optimal dual. Recall that $\gamma^\star$ is the optimal capacity dual. Let $\beta_i^\star = \beta_i^\star(\gamma^\star) = \max(0,c(i)+V^f(\gamma^\star,i)-\gamma^\star)$. The vector $(\beta_i^\star)_{i \in \set{S}}$ is an optimal solution to \eqref{eq:dual} by Lemma~\ref{lem:dual-structure}. Also let $(\bolds{q}^\star,\bolds{\nu}^\star)$ be an optimal solution to \eqref{eq:orifluid}. By strong duality and rearranging terms, the suboptimality of the equilibrium, $
\sum_{i \in \set{S}} c(i)(q_i^{\esti} - \nu_i^{\esti}) - C^\star$, is equal to 
\begin{align}
\sum_{a \in \set{S}} \nu_a^\star c^f(a) - \sum_{a \in \set{S}} \nu_a^{\esti} c^f(a) &= (\lambda \sum_{i \in \set{S}} \pi(i)\beta_i^\star + \mu \gamma^\star) - \sum_{a \in \set{S}} \nu_a^{\esti} c^f(a) \nonumber \\
&= 
\underbrace{\gamma^\star \left(\mu - \sum_{i \in \set{S}} \nu^{\esti}_i\right)}_{\text{Term 1}}+\underbrace{\sum_{i \in \set{S}} \beta_i^\star\left(\lambda \pi(i) - \sum_{a \in \anc(i)} \nu^{\esti}_a\frac{\pi(i)}{\pi(a)}\right) + }_{\text{Term 2}} \nonumber\\
&\hspace{0.2in}+\underbrace{\sum_{a \in \set{S}} \nu^{\esti}_a\left(\gamma^\star+\sum_{i \in \sub(a)} \beta^\star_i \frac{\pi(i)}{\pi(a)} - c^f(a)\right)}_{\text{Term 3}}. \label{eq:subopt-decomp}
\end{align}
The proof proceeds by bounding these three terms separately. The main intuition is that the priority order under $\bo^{\esti}$ should approximately rank states according to the three classes in Lemma~\ref{lem:prio-optimal}.

\begin{claim}
Term 1 is equal to zero.
\end{claim}
\begin{proof}
For the first term, we can directly apply the steps in the proof of Lemma~\ref{lem:prio-optimal} in Appendix~\ref{app:prio-optimal}. In particular, Lemma~\ref{lem:feasible-nu} shows that if the partially-served state exists, then $\sum_{i \in \set{S}} \nu_i^{\esti} = \mu$. Moreover, if this state does not exist, then the proof of Lemma~\ref{lem:prio-optimal} for Condition~\ref{item:c-1} shows that $\gamma^\star = 0$ because the capacity constraint is not tight (see \eqref{eq:g-gamma}). Therefore, Term 1 is equal to zero.
\end{proof}

\begin{claim}
Term 2 is at most $2\varepsilon / \theta$.
\end{claim}
\begin{proof}
Let $h$ be the first position in $\bo$ such that $\ind_{\alg}(\bo_h) \leq \gamma^\star$ (if it does not exist then let $h = |\set{S}|+1$). Therefore, for any position $k < h$, its $\alg$ index satisfies $\ind_{\alg}(\bo_k) > \gamma^\star$. Moreover, for position $k > h$,  its $\alg$ index satisfies
\[
\ind_{\alg}(\bo_k) \leq \ind_{\esti}(\bo_k) + \varepsilon \leq \ind_{\esti}(\bo_h) + \varepsilon \leq \ind_{\alg}(\bo_h) + 2\varepsilon \leq \gamma^\star + 2\varepsilon,
\]
where the first inequality is by the definition of the error $\varepsilon$; the second is because $\bo_k$ ranks after $\bo_h$ under index $\ind_{\esti}$; the third is again by the error $\varepsilon$; and the last is by the definition of $h$. Therefore, for $i \in \set{S} \setminus \bo_{[h-1]}$, its dual variable satisfies
\[\beta_i^\star = \max(0,c(i)+V^f(\gamma^\star,i) - \gamma^\star) = \max(0,\ind_{\alg}(i) - \gamma^\star) \leq 2\varepsilon.\]
For the subset $\bo_{[h-1]}$, since every state in it has $c(i)+V^f(\gamma^\star,i) > \gamma^\star$, it is a subset of $\set{S}_{\mhigh}$ (defined in Lemma~\ref{lem:prio-optimal}).  Applying Lemma~\ref{lem:top-prob} and \eqref{eq:shi-prob} gives $\sum_{i \in \TopS(\bo_{[h-1]})} \lambda \pi(i) \leq \sum_{i \in \TopS(\set{S}_{\mhigh})} \lambda \pi(i) \leq \mu$. Therefore, the water-filling procedure in Section~\ref{sec:fluid-sol} ensures that for any $i \in \bo_{[h-1]}$, it satisfies $q_i^{\esti} = \nu_i^{\esti}$ and thus $\lambda \pi(i) = \sum_{a \in \anc(i)} \nu_a^{\esti} \pi(i) / \pi(a)$ by \eqref{eq:equiv-constraint}. We thus upper bound Term 2 by
\[
\text{Term 2} \leq \sum_{i \in \set{S} \setminus \bo_{[h-1]}} \beta_i^\star \lambda \pi(i) \leq 2\varepsilon \sum_{i \in \set{S}} \lambda \pi(i) \leq 2(\varepsilon/\theta).
\]
\end{proof}

\begin{claim}
Either Term 3 is at most $2\varepsilon / \theta$ or $\sum_{i \in \set{S}} c(i)(q_i^{\esti} - \nu_i^{\esti}) = C^\star$.
\end{claim}
\begin{proof}
Using Lemma~\ref{lem:sub-tree-beta}, we first simplify Term 3 by
\begin{equation}\label{eq:simterm3}
\text{Term 3} = \sum_{a \in \set{S}} \nu_a^{\esti} (\gamma^{\star} - V(\gamma^\star, a)) = \sum_{a \in \set{S}} \nu_a^{\esti} \max(0,\gamma^\star - \ind_{\alg}(a)),
\end{equation}
where we apply the fact that $V(\gamma^\star,a) = \min(\gamma^\star, \ind_{\alg}(a)).$
Let $\ell$ be the last position such that $\ind_{\alg}(\bo_{\ell}) \geq \gamma^\star$ (if it does not exist, let $\ell = 0$). For any position $k > \ell$, the corresponding state has $\ind_{\alg}(\bo_k) < \gamma^\star$. Moreover, for any $k < \ell$, 
\[
\ind_{\alg}(\bo_k) \geq \ind_{\esti}(\bo_k) - \varepsilon \geq \ind_{\esti}(\bo_{\ell}) - \varepsilon \geq \ind_{\alg}(\bo_{\ell}) - 2\varepsilon \geq \gamma^\star - 2\varepsilon.
\]
Let $m$ be the maximum position such that $\sum_{i \in \TopS(\bo_{[m]})} \lambda \pi(i) \leq \mu$. The water-filling procedure in Section~\ref{sec:fluid-sol} ensures that any state after position $m + 1$ will have no allocated capacity, i.e., $\nu^{\esti}_{\bo_k} = 0$ for $k > m + 1$. Next consider two cases, which correspond to the two possible results of the claim:
\begin{itemize}
\item Suppose $\ell \geq m + 1$. This implies that $\nu_{\bo_k}^{\esti} = 0$ for any $k > \ell$. Applying \eqref{eq:simterm3} and above bounds on $\ind_{\alg}(\bo_k)$ for $k \leq \ell$ and $k > \ell$ gives
\begin{align*}
\text{Term 3} &= \sum_{a \in \bo_{[\ell]}} \nu_a^{\esti}\max(0,\gamma^\star - \ind_{\alg}(a)) +  \sum_{a \in \set{S} \setminus \bo_{[\ell]}} \nu_a^{\esti}\max(0,\gamma^\star - \ind_{\alg}(a)) \\
&\leq 2\varepsilon \sum_{a \in \bo_{[\ell]}} \nu_a^{\esti} \leq 2\varepsilon\mu \leq 2\varepsilon/\theta,
\end{align*}
where the first inequality is because (i) $\ind_{\alg}(\bo_k) \geq \gamma^\star - 2\varepsilon$ for $k \leq \ell$ and (ii) $\nu_{\bo_k}^{\esti} = 0$ for $k > \ell \geq m+1$; the second inequality is because $\sum_{a \in \set{S}} \nu_a^{\esti} \leq \mu$ by the capacity constraint; and the last inequality is because $\mu \leq 1 \leq 1 / \theta$. This corresponds to the first case of the claim.
\item Suppose instead that $\ell < m+1$. Then the water-filling procedure ensures that $q_i^{\esti} = \nu_i^{\esti}$ for any $i \in \sub(\bo_{[\ell]})$, i.e., there should be no job with state in $\sub(\bo_{[\ell]}).$ Moreover, \eqref{eq:bound-opt} shows 
$C^\star \geq \lambda \sum_{i \not \in \sub(\set{S}_{\mhigh} \cup \set{S}_{\mequal})} \pi(i)c(i),$
where the set $\set{S}_{\mhigh} \cup \set{S}_{\mequal}$ contains all states with $\alg$ index larger than or equal to $\gamma^\star$. Since we showed that $\ind_{\alg}(\bo_k) < \gamma^\star$ for $k > \ell$, the set $\set{S}_{\mhigh} \cup \set{S}_{\mequal}$ is a subset of $\bo_{[\ell]}$, implying that $\sub(\set{S}_{\mhigh} \cup \set{S}_{\mequal})$ is a subset of $\sub(\bo_{[\ell]})$. Therefore, $C^\star \geq \lambda \sum_{i \in \set{S} \setminus \sub(\set{S}_{\mhigh} \cup \set{S}_{\mequal})} \pi(i)c(i) \geq \lambda \sum_{i \in \set{S} \setminus \sub(\bo_{[\ell]})} \pi(i)c(i).$ The cost of the equilibrium of $\prio(\bo)$ then satisfies
\[
\sum_{i \in \set{S}} c(i)(q_i^{\esti} - \nu_i^{\esti}) \leq \sum_{i \in \set{S} \setminus \sub(\bo_{[\ell]})} c(i)(q_i^{\esti} - \nu_i^{\esti}) \leq \sum_{i \in \set{S} \setminus \sub(\bo_{[\ell]})} c(i)(\lambda \pi(i)) \leq C^\star,
\]
where the first inequality is because $q_i^{\esti} = \nu_i^{\esti}$ for any $i \in \sub(\bo_{[\ell]})$; the second inequality is because $q_i^{\esti} \leq \lambda \pi(i)$ by \eqref{eq:def-fluid-q}. We thus show the second case of the claim. 
\end{itemize}
\end{proof}

\begin{proof}[Proof of Lemma~\ref{lem:fluid-error}]
By the last claim, if Term 3 is not at most $2\varepsilon / \theta$, we can already conclude that the equilibrium is fluid optimal, finishing the proof. If it is at most $2 \varepsilon / \theta$, we combine it with the first two claims and apply the decomposition in \eqref{eq:subopt-decomp} to get the desired result.
\end{proof}

\section{Asymptotic Optimality of Gittins-Index Algorithm (Remark~\ref{remark:gittins})}\label{app:gittins}
% !TEX root = main.tex
This section shows how our analytical framework can prove asymptotic optimality of a different algorithm based on Gittins index. To define Gittins index, we first show that for a fixed state $i$, there is a threshold such that the cost-to-go function $V(\gamma,i)$ is equal to $\gamma$ when $\gamma$ is below the threshold and less than $\gamma$ when $\gamma$ is larger than the threshold.   
The proof follows from standard results (e.g. lemma 3.5 of \cite{scully2025gittins}); for completeness, we provide it in Section~\ref{sec:gittins-exist}.
\begin{lemma}\label{lem:gittins-exist}
For any state $i$, there exists $g \geq 0$ such that $V(\gamma,i) = \gamma$ for $\gamma \leq g$ and $V(\gamma,i) < \gamma$ for $\gamma > g$.
\end{lemma}
With Lemma~\ref{lem:gittins-exist}, any state $i$ possesses a Gittins index, defined by 
\begin{equation}\label{eq:gittins}
\gi(i) \coloneqq \sup\{\gamma \geq 0 \colon V(\gamma,i) = \gamma\}.
\end{equation}
State-of-the-art computational methods solve the Gittins index in sub-cubic time in the size of the state space \cite{gast2023testing}. With these indices, we consider the following algorithm ($\gi$) that prioritizes jobs whose states have high Gittins index for service.

\begin{algorithm}[H]
\LinesNumbered
\DontPrintSemicolon
  \caption{
  \textsc{Gittins Index} ($\gi$)}\label{algo:gittins}
  \KwData{cost vector $\bolds{c}$ and transition probability $P$}
  \tcc{Offline training}
  Solve the Gittins index \eqref{eq:gittins} for all states \; 
  \tcc{Online scheduling}
  \For{$t = 1$ \KwTo $T$}{
    Observe the set of waiting jobs $\set{Q}(t)$ and the random services $R(t)$ \;
    Assign to each job $j \in \set{Q}(t)$ an index $\gi(S_j(t))$\; 
    Serve the $R(t)$ jobs from $\set{Q}(t)$ with highest indices using a consistent tie-breaking rule \;
    Incur cost of unserved jobs and observe new arrivals \;
  }
\end{algorithm}
Applying the analytical framework in Section~\ref{sec:analysis} shows that $\gi$ has vanishing suboptimality gap.
\begin{proposition}\label{prop:gittins}
The guarantee in Theorem \ref{thm:pafou} holds for $\gi$.
\end{proposition}
To prove Proposition~\ref{prop:gittins}, it is sufficient to prove that the equilibrium of $\textsc{GI}$ has optimal fluid cost (an analogue of Lemma~\ref{lem:optimal-fluid}). This is because Lemmas~\ref{lem:stochastic} and \ref{lem:lower-bound} do not depend on a specific priority algorithm and the proof of Theorem~\ref{thm:pafou} just combines Lemmas~\ref{lem:optimal-fluid}, \ref{lem:stochastic}, and \ref{lem:lower-bound}. Recall that the proof of Lemma~\ref{lem:optimal-fluid} follows directly from Lemma~\ref{lem:prio-optimal}, which partitions states into three classes $\set{S}_{\mhigh}, \set{S}_{\mequal}, \set{S}_{\mlow}$ depending on whether $c(i)+V^f(\gamma^\star,i)$ is higher, equal, or lower than $\gamma^\star$ for a given state $i$. By Lemma~\ref{lem:prio-optimal}, a sufficient condition for the equilibrium of a priority algorithm to be fluid optimal is ranking $\set{S}_{\mhigh}$ before $\set{S}_{\mequal}$ before $\set{S}_{\mlow}$. The below lemma shows that this is the case for $\textsc{GI}$, whose proof is in Section~\ref{sec:gittins-dual}.
\begin{lemma}\label{lem:Gittins-dual}
The Gittins index satisfies that (i) for any state $i \in \set{S}_{\mhigh}$, $\gi(i) > \gamma^\star$; (ii) for any state $i \in \set{S}_{\mequal}$, $\gi(i) = \gamma^\star$; (iii) for any state $i \in \set{S}_{\mlow}$, $\gi(i) < \gamma^\star$.
\end{lemma}
\begin{proof}[Proof of Proposition~\ref{prop:gittins}]
Since the $\textsc{GI}$ algorithm ranks states based on their Gittins index, Lemma~\ref{lem:Gittins-dual} implies that it ranks $\set{S}_{\mhigh}$ before $\set{S}_{\mequal}$ before $\set{S}_{\mlow}$. Lemma~\ref{lem:prio-optimal} thus shows that its equilibrium is fluid optimal. Combining this result with Lemmas~\ref{lem:stochastic} and \ref{lem:lower-bound} thus finishes the proof as in the proof of Theorem~\ref{thm:pafou} in Section~\ref{sec:proof-pafou}.
\end{proof}
\begin{remark}\label{remark:connection} 
There is certain equivalence between $\gi$ and $\alg$. By Lemma~\ref{lem:Gittins-dual} and Lemma~\ref{lem:prio-optimal}, they are equivalent in the \textbf{fluid} setting as both of them will serve all states in $\set{S}_{\mhigh}$, one state in $\set{S}_{\mequal}$ (indeed, they will serve the same state if they follow the same tie-breaking rule), and no state in $\set{S}_{\mlow}$. However, these two priority rules can have different orderings of the states and thus different behaviors in the \textbf{stochastic} setting. For example, consider a setting with $\lambda < \mu$, so $\gamma^\star = 0$ as the system is underloaded. There are three states $\{\rootNode, s_1, s_2\}$ and $\rootNode$ transitions to $s_1$ or $s_2$ with equal probability. Their instantaneous holding costs are $c(\rootNode)=12, c(1)=13, c(2) =1$ respectively. $\alg$ thus prioritizes $s_1$ over $\rootNode$ over $s_2$ since their indices equal their instantaneous holding costs (as $\gamma^\star=0$). However, $\gi$ prioritizes $\rootNode$ over $s_1$ over $s_2$ because their Gittins indices are $\gi(\rootNode) = 19, \gi(s_1) = c(1) = 13, \gi(s_2) = c(2) = 1.$ 
\end{remark}

\subsection{Existence of Gittins index (Lemma~\ref{lem:gittins-exist})}\label{sec:gittins-exist} 
To prove Lemma~\ref{lem:gittins-exist}, we show below properties of $V(\gamma,i)$ as a function of the dual $\gamma$ in Lemma~\ref{lem:val-properties}. 
\begin{lemma}\label{lem:val-properties}
For any state $i$, the function $\gamma \mapsto V(\gamma,i)$ satisfies two properties: 
\begin{enumerate}
 \item [(i)] it is concave and upper bounded by $V(\gamma,i) \leq \gamma$;
 \item [(ii)] it is non-decreasing and $1$-Lipschitz: for any $0 \leq \gamma < \gamma'$, 
 $0 \leq V(\gamma',i) - V(\gamma,i) \leq \gamma'-\gamma.$
\end{enumerate}
\end{lemma}
\begin{proof}
In the Markovian ski-rental problem of Section~\ref{sec:ski-rental}, let $\smallsquare$ denote the action of serving a job (buying) and $\triangleright$ denote the action of not serving it (renting). As in eq. (3.4) of \cite{scully2025gittins}, the value function $V(\gamma,i)$ is equal to the minimum cost among the finite set of stationary policies mapping from states to the two actions:
\begin{equation}\label{eq:optimal-val}
V(\gamma,i) = \min_{\pi \colon \set{S} \to \{\smallsquare, \triangleright\}} \left(E_{\pi}(i) + \gamma P_{\pi}(i)\right),
\end{equation}
where $E_{\pi}(i)$ is the expected cost of policy $\pi$ until the state transitions into the terminal state when starting from $i$ and $P_{\pi}(i)$ is the probability that policy $\pi$ takes action $\smallsquare$.

For Property (i), the concavity of $\gamma \mapsto V(\gamma,i)$ is because  \eqref{eq:optimal-val} is the minimum among a finite set of linear functions. Moreover, $V(\gamma,i) \leq \gamma$ because of the trivial policy of always serving a job.

For Property (ii), let $V^{\pi}(\gamma,i) = E_{\pi}(i)+\gamma P_{\pi}(i).$ This function is non-decreasing and $1-$Lipschitz in $\gamma$. The proof completes because  the minimum of finitely many such functions in \eqref{eq:optimal-val} preserves such properties.
\end{proof}
\begin{proof}[Proof of Lemma~\ref{lem:gittins-exist}]
For any state $i$, define a function $f(\gamma) = V(\gamma,i) - \gamma$. This function is concave because $V(\gamma,i)$ (as well as the linear function) is concave in $\gamma$ ((i) of Lemma~\ref{lem:val-properties}). Therefore, the super-level set $\{\gamma \geq 0 \colon f(\gamma) \geq 0\}$ is convex, i.e., is an interval of $\mathbb{R}_+.$ Since $V(\gamma,i) \leq \gamma$ ((i) of Lemma~\ref{lem:val-properties}), the set $\{\gamma\geq 0\colon V(\gamma,i) = \gamma\}$ is the same as the set $\{\gamma \geq 0 \colon f(\gamma) \geq 0\}$ and is an interval. This interval must be finite because $\lim_{\gamma \to \infty} V(\gamma,i) < \infty$ by the finiteness of Markov tree depth assumed in Section~\ref{sec:model}. Moreover, it starts from $0$ since $V(0,i) = 0.$ We  thus conclude that there exists a value $g \geq 0$ such that $V(\gamma,i) = \gamma$ for $\gamma < g$ and $V(\gamma,i) < \gamma$ for $\gamma > g$. The proof finishes by noting that $V(\gamma,i) - \gamma$ is Lipschitz (by (ii) of Lemma~\ref{lem:val-properties}) and thus continuity requires that $V(g,i) = g$.
\end{proof}
\subsection{Gittins index and optimal dual (Lemma~\ref{lem:Gittins-dual})}\label{sec:gittins-dual}
To prove Lemma~\ref{lem:Gittins-dual}, we define a function $V_{i,\triangleright}(\gamma) = c(i) + V^f(\gamma,i)$, which is the optimal cost of not serving a job. The proof of Lemma~\ref{lem:Gittins-dual} uses the property that $V_{i,\triangleright}(\gamma)$ is $(1-\theta)-$Lipschitz, where $\theta$ is the minimum abandonment probability $\theta =\min_{i \in \set{S}}1 - \sum_{k \in \child(i)} P(i,k)$.
\begin{lemma}\label{lem:contraction}
For any state $i$ and $0 \leq \gamma < \gamma'$, $ V_{i,\triangleright}(\gamma') - V_{i,\triangleright}(\gamma) \in [0,(1-\theta)(\gamma'-\gamma)].$
\end{lemma}
\begin{proof}
Fix any state $i$ and any two non-negative values $\gamma < \gamma'.$ The difference $V_{i,\triangleright}(\gamma') - V_{i,\triangleright}(\gamma)$ is:
\[
V_{i,\triangleright}(\gamma') - V_{i,\triangleright}(\gamma) = V^f(\gamma',i) - V^f(\gamma,i) = \sum_{k \in \child(i)} P(i,k)(V(\gamma',k) - V(\gamma,k)).
\]
By (ii) of Lemma~\ref{lem:val-properties}, the difference is non-negative because $V(\gamma',k) \geq V(\gamma, k)$ for any state $k$ when $\gamma' > \gamma$. In addition, applying the $1-$Lipschitz property in Lemma~\ref{lem:val-properties}, $V(\gamma',k) - V(\gamma,k) \leq \gamma' - \gamma$ for any state $k$. Therefore, 
\[
V_{i,\triangleright}(\gamma') - V_{i,\triangleright}(\gamma) = \sum_{k \in \child(i)} P(i,k)(V(\gamma',k) - V(\gamma,k)) \leq \sum_{k \in \child(i)} P(i,k) (\gamma'-\gamma) \leq (1-\theta)(\gamma'-\gamma).
\]
\end{proof}
\begin{proof}[Proof of Lemma~\ref{lem:Gittins-dual}]
The proof shows the three results separately.

For result (i), fix a state $i \in \set{S}_{\mhigh}$, which satisfies  $V_{i,\triangleright}(\gamma^\star) > \gamma^\star.$ Take $\gamma' = V_{i,\triangleright}(\gamma^\star) > \gamma^\star.$ By the monotonicity in Lemma~\ref{lem:contraction}, the value function at $\gamma'$ satisfies $V_{i, \triangleright}(\gamma') \geq V_{i, \triangleright}(\gamma^\star) = \gamma'$. Therefore, since $V(\gamma',i) = \min(\gamma', V_{i,\triangleright}(\gamma')) = \gamma'$, the definition of $\gi(i)$ in \eqref{eq:gittins} gives $\gi(i) \geq \gamma' > \gamma^\star.$

For result (ii), take any state $i \in \set{S}_{\mequal}$, which satisfies $V_{i,\triangleright}(\gamma^{\star}) = \gamma^\star$. By definition \eqref{eq:gittins}, the Gittins index satisfies $\gi(i) \geq \gamma^\star$. Moreover, for any $\gamma' > \gamma^\star$, the Lipschitz property in Lemma~\ref{lem:contraction} shows that 
\[
V_{i,\triangleright}(\gamma') \leq V_{i,\triangleright}(\gamma^\star) + (1-\theta)(\gamma' - \gamma^\star) = \gamma^\star + (1-\theta)(\gamma' - \gamma^\star) < \gamma'.
\]
As a result, $V(\gamma', i) = \min(\gamma', V_{i,\triangleright}(\gamma')) < \gamma'$, showing that $\gi(i) \leq \gamma^\star$ and thus $\gi(i) = \gamma^\star$.

For result (iii), take any state $i \in \set{S}_{\mlow}$, which satisfies $V_{i, \triangleright}(\gamma^\star) < \gamma^\star$ and thus $V(\gamma^\star,i) < \gamma^\star$. Lemma~\ref{lem:gittins-exist} then implies that $\gi(i) < \gamma^\star$.
\end{proof}

\section{Robustness Checks of Numerical Results (Section~\ref{ssec:robustness})}
\label{app:simulations}
% !TEX root = main.tex
This appendix includes robustness checks for our model and numerical results. Appendix~\ref{app:model-general} discusses how our model can directly capture the expected policy-violating views without separating out the probability of violation. Appendices~\ref{app:organic-synthetic} and \ref{app:paid} study the robustness of our performance gain in synthetic datasets mimicking user-generated content and online ads respectively. Appendix~\ref{app:uncalibrated} assesses how miscalibration in the probability of violation affects algorithmic performance. Appendix~\ref{app:pvio-update} explores how our algorithm may capture evolution in the probability of violation. Appendix~\ref{app:bandit-learning} examines the compatibility of different scheduling algorithms with prior learning-based approaches on the probability of violation.

\subsection{Instantaneous cost based on expected policy-violating views}\label{app:model-general}
Section~\ref{sec:instantiate} approximates the metric of policy-violating views by replacing the actual violation with a probability of violation. This section shows that theoretically this is unnecessary. Our model can be instantiated to minimize the expectation of policy-violating views \eqref{eq:def-vioviews}. For a job (content piece) $j$ in the system, its state $S_j(t)$ at time $t$ captures all available information related to its actual violation $\violating(j)$ and its number of views. The cost for this job at time $t$, $c(S_j(t))$, is defined by
\begin{equation}\label{eq:cost-both}
c(S_j(t)) = \expect{\violating(j)\times \view(j,t) \mid S_j(t)},
\end{equation}
i.e., the conditional expectation of its instantaneous violating views for period $t$ based on current information. In this way, whether or not the (predicted) probability of violation is independent from the view trajectory or may change across time is irrelevant: it is not part of the model, and we can include it  into the state definition, viewing it as information related to the actual violation of a content piece.

\subsection{Simulations on user-generated content with synthetic data}\label{app:organic-synthetic}
This subsection evaluates different scheduling algorithms on content moderation of user-generated content (UGC), i.e., content generated and shared by social media users. To capture the unique nature of UGC, this section generates synthetic view trajectories based on Hawkes processes, which have been widely utilized to model view trajectories of UGC \citep{rizoiu2017expecting, haimovich2021popularity}. As shown in Figure~\ref{fig:example-trajectory}, the view patterns of generated UGC are similar to the YouTube video dataset that Section~\ref{sec:video-real} focuses on, with both showing volatile number of views per period.

\paragraph{Data generation.} For both the training and test datasets, we generate the view trajectory of $10,000$ posts by a Hawkes process with exponentially decaying kernels motivated by \cite{haimovich2021popularity}, which describes a prediction model for views of Facebook posts. In our synthetic dataset, a post $j$ has a trajectory of length $500$ and a hyper-parameter $\beta_j$ sampled uniformly from $[0.8,2]$ dictating the rate of decay of its views. The creator of this post has a random number of followers $f_j$ given by $10\left(\text{Pareto}(1.1)+1\right)$. The post initially has a random number of views given by a Poisson distribution with rate $0.1f_j$.

For the number of views in the $d-$th period, $\view(j,d)$ for $2 \leq d \leq 500$, it is a Poisson random variable with mean being the sum of the following three sources, capped at $5000$:
\begin{enumerate}
\item Exposure from the followers: the $f_j$ followers may view this post at period $d$ and we assume that the view rate from these followers is $0.05f_j\exp(-\beta_j (d-1)).$
\item Exposure from past views: every past view of this post may trigger some new views. We assume that a view in period $s < d$ will lead to a view rate of $\exp(-\beta_j(d-s)) \cdot X_j$, where $X_j$ is a post-specific trigger rate given by a random variable uniform in $[0.3,0.9999]$, multiplied by a normalization factor $\exp(\beta_j)-1$.
\item Exposure from random shocks: with probability $0.01$, there is a random shock of rate given by a Pareto random variable of shape $1.5$.
\end{enumerate}
The probability of policy-violating $p\violating(j)$ is sampled from a $\mathrm{Beta}(\beta_j+4/\beta_j,6)$ distribution. This generation process allows the probability of a post being policy-violating to be correlated with its view trajectory, which reflects the possibility that some inherent features of a post can impact both its probability of being policy-violating and its number of views.

\begin{figure}[hbtp]
\centering
\includegraphics[width=\textwidth]{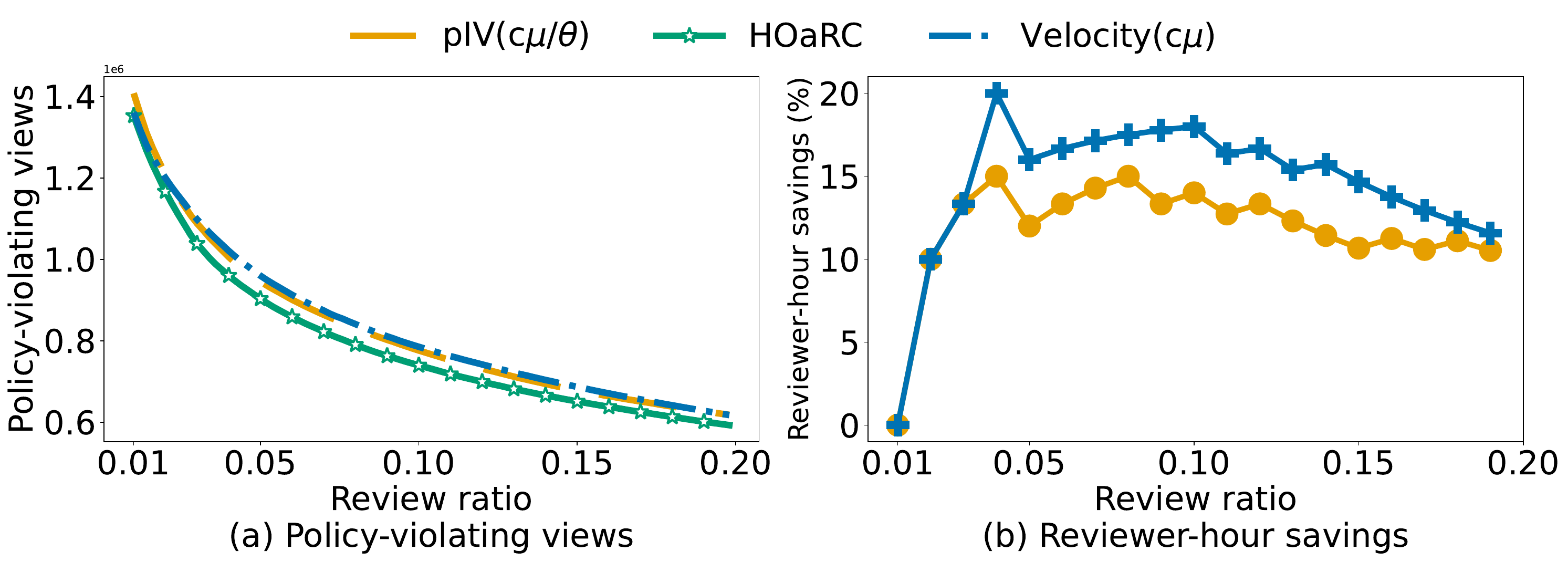}
\caption{UGC: policy-violating views of algorithms and reviewer-hour savings by $\hindalg$}
\label{fig:organic-calibrated}
\end{figure}

\paragraph{Simulation results.} As in Figure~\ref{fig:youtube-calibrated}, Figure~\ref{fig:organic-calibrated} shows the policy-violating views of the algorithms and reviewer-hour savings when using $\hindalg$ instead of existing heuristics. Similar to our results in Figure~\ref{fig:youtube-calibrated}, $\hindalg$ consistently outperforms existing heuristics with  $0.6\%$ to $6\%$ reduction in policy-violating views and $10\%$ to $20\%$ reduction in reviewer hours for the synthetic UGC setting.

\subsection{Simulations on ads content with synthetic data}\label{app:paid}

This subsection focuses on content moderation of online ads, whose view trajectories can be vastly different from those of user-generated content (the focus of Sections~\ref{sec:video-real} and \ref{app:organic-synthetic}) as ads views are more under control of the platform. In particular, in social media platforms, advertisers typically set up campaigns specifying a) a set of ads in the campaign, b) a metric they want to optimize (such as click-through rate), c) budget to be spent on the campaign, and d) time duration over which the campaign runs.  From the set of ads in the campaign, the platform attempts to identify the best performing ad(s) (as per the specified metric) and then spends most of the budget on this set of ads. Ad platforms typically use exploration-exploitation techniques to find the best ads \citep{chakrabarti2008mortal, schwartz2017customer}. All ads tend to get some views in the initial (exploration) period whereas later, only the set of chosen best performing {ads} get most of the views. Further, after the exploration period, the views on the best performing ads are typically stable over time due to pacing algorithms used by ad platforms \citep{agarwal2014budget}. This process is different from how user-generated content obtains views: the latter gets exposure through uncertain shares by users, and thus the view trajectories can exhibit more volatility. We can also see this difference in view trajectories through Figure~\ref{fig:example-trajectory}, where our (synthesized) trajectories of online ads are more stable compared to the trajectories of the other two settings. This subsection aims to understand to what extent the benefit of our algorithm for user-generated content, as shown in Sections~\ref{sec:video-real} and \ref{app:organic-synthetic}, remains for online ads.

\paragraph{Data generation.} Motivated by the above discussion, we create a synthetic dataset of ads as follows. The training set and test set are generated in the same manner, but with different initial random seeds. For each dataset, there are $20,000$ campaigns. Each campaign has $5$ ads. Denoting the $k-$th ad of campaign $u$ by ad $(u,k)$, the training (test) dataset $\set{D}_{\text{train}}$ ($\set{D}_{\text{test}}$) consists
of $20,000 \times 5 = 100,000$ content pieces where a content piece $j$ is an ad $(u,k)$ for some $u \leq 20,000$ and $k \leq 5$. The policy-violating probability $p\violating((u,k))$ is identical across all ads of the same campaign, which is sampled from a $\mathrm{Beta}(1,3)$ distribution. The actual violation $\violating((u,k))$ of an ad is an i.i.d. Bernoulli random variable with mean $p\violating((u,k)).$ 
To generate the view trajectory of ads, a campaign $u$ has a random variable $X_u$ sampled from a Pareto distribution with a shape parameter $0.8$, which captures the campaign's per-period budget. An ad $(u,k)$ has a mean reward (e.g. click-through rate) $r_{u,k}$ sampled from a $\mathrm{Beta}(1,5)$ distribution. The campaign lasts for a random amount of time, with the horizon $L$ uniformly sampled from $\{1,2,\ldots,300\}$. Then over a horizon of $L$ periods, the platform seeks to identify an ad from this campaign with the highest mean reward, which is a multi-armed bandit problem. We implement the UCB1 algorithm \citep{auer2002finite} for this purpose assuming that ad $(u,k)$ has a Bernoulli reward with mean $r_{u,k}$ when promoted (pulled) by the platform. The number of views ad $(u,k)$ will get for the $d-$th period, i.e., $\view((u,k), d)$ in the dataset \eqref{eq:dataset}, is equal to a Poisson random variable with rate $X_u$ if this ad is promoted for period $d$; and $0$ otherwise. As illustrated in the third plot of Figure~\ref{fig:example-trajectory}, the synthetic view trajectories for ads with {the} highest number of views grow roughly linearly. This reflects the characteristic of ad view trajectories where the ad identified as 
the best performing eventually gets a stable number of views in every period due to the platform's pacing algorithm.

\paragraph{Simulation results.} 
As shown in Figure~\ref{fig:bandit-calibrated}, $\hindalg$ consistently has fewer policy-violating views than the other heuristics when the review ratio is not too small (larger than 2\%). Although the improvement over $\textsc{Velocity}$ is not as substantial as the improvement for the YouTube setting (Section~\ref{sec:video-real}) and the synthetic user-generated content setting (Appendix~\ref{app:organic-synthetic}), Figure~\ref{fig:bandit-calibrated} shows that $\hindalg$ can still save several percentage of reviewer hours from $\textsc{Velocity}$.

Compared to our settings in Section~\ref{sec:video-real} and Appendix~\ref{app:organic-synthetic}, the online ads setting exhibits less volatility. Therefore, the cost for the current period is a good indicator of the future costs of a post, making $\textsc{Velocity}$ perform well. In contrast, the ability of $\hindalg$ (and of $\textsc{pIV}$)  to incorporate future effect is harmed by the prediction error. These combined explain why $\textsc{Velocity}$ is performing comparatively well in the online ads setting.

\begin{figure}[hbtp]
\centering
\includegraphics[width=\textwidth]{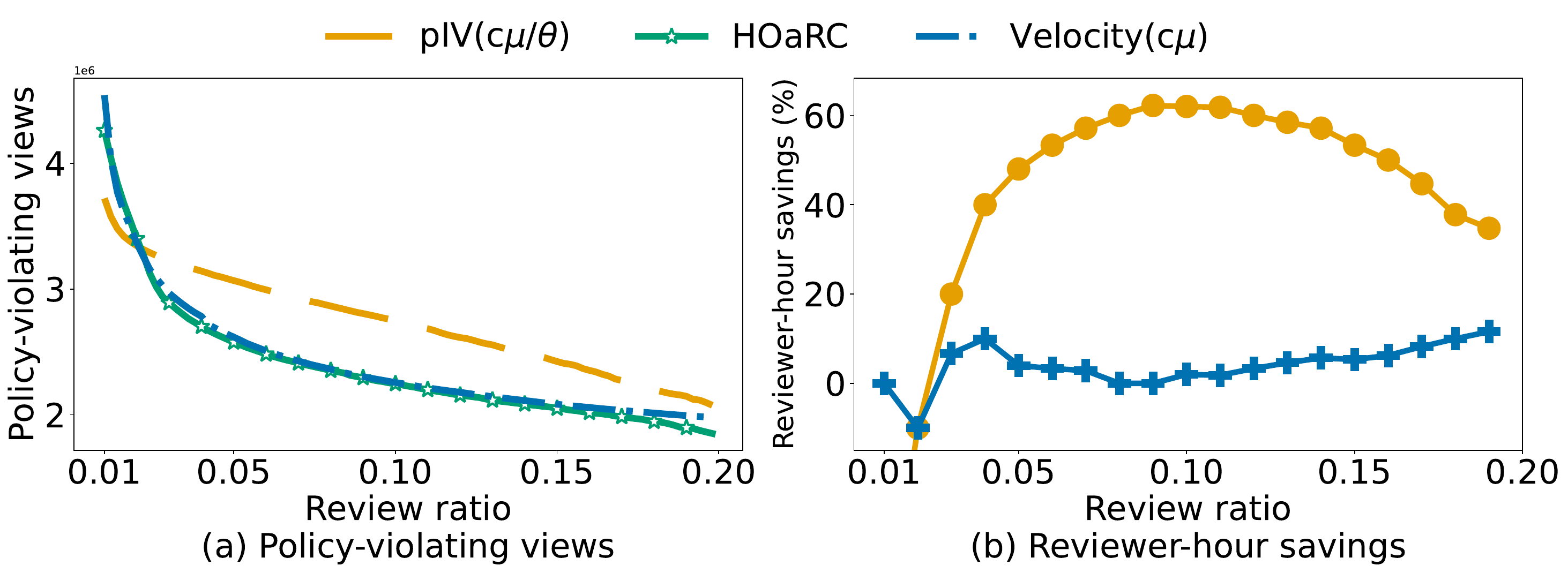}
\caption{Ads: policy-violating views of algorithms and reviewer-hour savings by $\hindalg$}
\label{fig:bandit-calibrated}
\end{figure}

\subsection{The impact of uncalibrated predictions}\label{app:uncalibrated}
In Section~\ref{sec:video-real}, the ground truth violation of content piece $j$, $\violating(j)$, is generated as a Bernoulli random variable with mean $p\violating(j)$. This assumes that the predicted probability of policy violation is \emph{well calibrated} \citep{degroot1983comparison}. In practice, such an assumption often fails to hold true due to non-stationarity in content trends; see the discussion in \cite{avadhanula2022}. 

To study the robustness of our algorithm to uncalibrated predictions of policy-violating probability, we apply perturbed values of $p\violating(j)$ when simulating the scheduling algorithms. Specifically,  the actual violation of content piece $j$, $\violating(j)$ is still generated as a Bernoulli random variable with mean $p\violating(j)$. Then for a given maximum calibration error $\varepsilon$, which varies in $\{0.03 \cdot k: 0 \leq k < 10\}$, the prediction for content piece $j$ has a calibration error $e_j$ that is i.i.d. generated from a uniform random variable on  $[-\varepsilon,\varepsilon].$ We simulate each scheduling algorithm with a perturbed version of the policy-violating probabilities, $\widetilde{p\violating}(j)$, which is equal to $p\violating(j) + e_j$ clipped within $[0,1]$. The review ratio is set as $5\%$.

\begin{figure}
\centering
\includegraphics[width=2in]{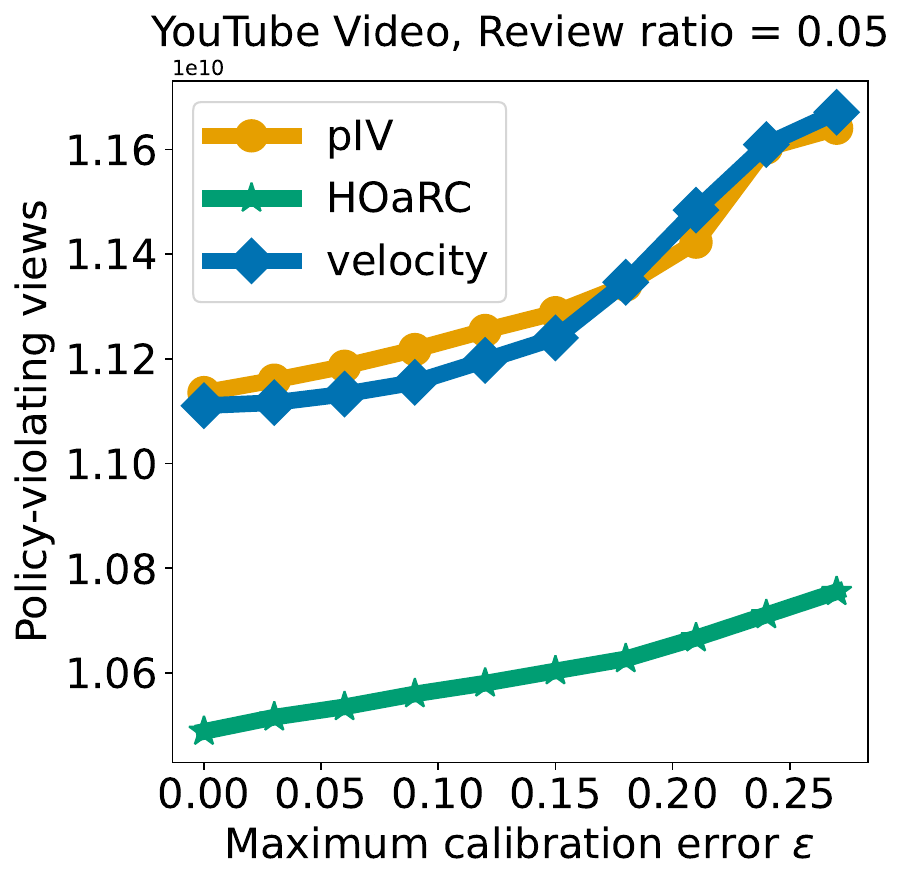}
\includegraphics[width=2in]{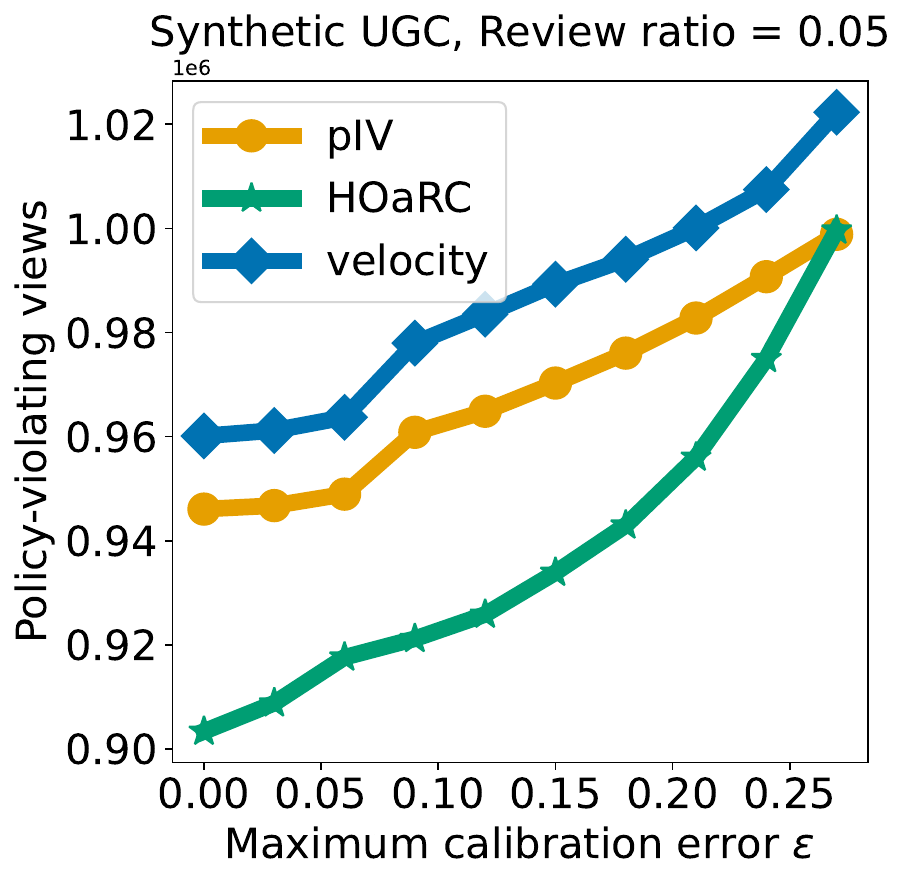}
\includegraphics[width=2in]{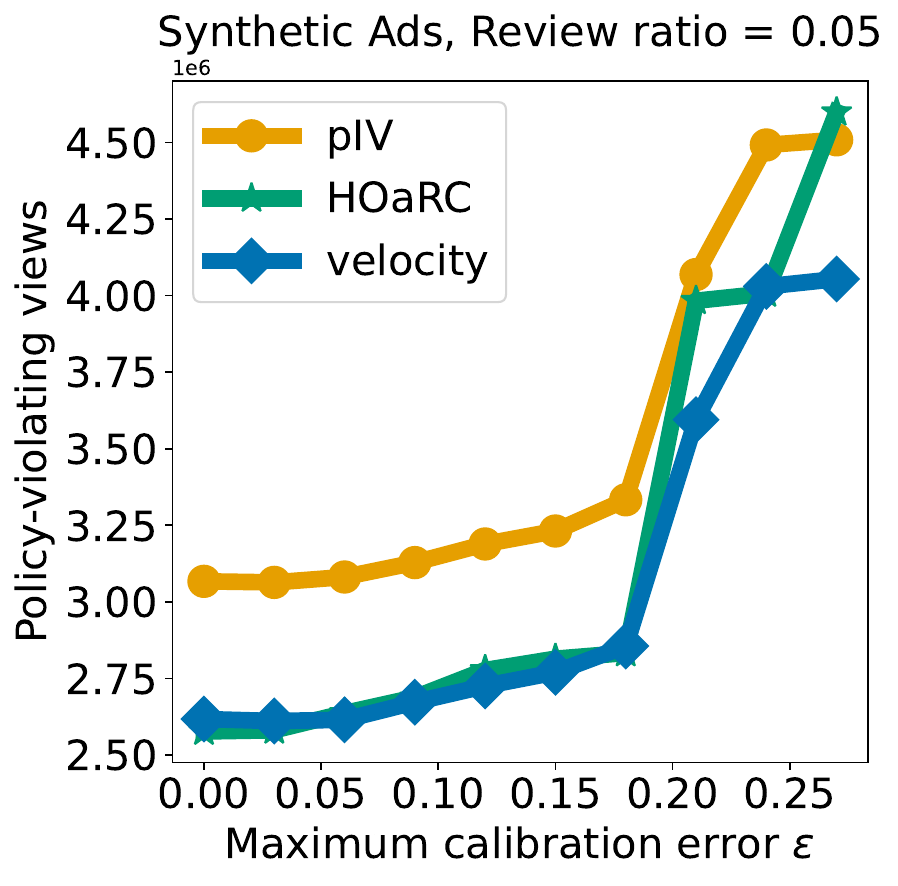}
\caption{Policy-violating views when the maximum calibration error varies ($5\%$ review ratio)}
\label{fig:uncalibrated-5}
\end{figure}

Figure~\ref{fig:uncalibrated-5} shows the policy-violating views of the considered scheduling algorithms as a function of the maximum calibration error $\varepsilon$ for the same datasets considered in Section \ref{sec:video-real} and Appendices~\ref{app:organic-synthetic}- \ref{app:paid}. Note that we do not include $\textsc{pViolating}$ as its policy-violating views are substantially higher than the others. A few observations follow from the figure. First, the policy-violating views of any algorithm generally increase with the maximum calibration error. Second, the gain for $\hindalg$ is generally robust to the calibration error in the YouTube setting and the synthetic user-generated content setting. Finally, for the online ads setting, which is not the setting $\hindalg$ is expected to perform much better (see the discussion in Appendix~\ref{app:paid}), the performance of $\hindalg$ remains close to the best performing benchmark except for high calibration errors. 

\subsection{Capturing evolution in the probability of violation}\label{app:pvio-update}

This subsection extends the setting in Section~\ref{sec:numerics} to capture evolution in the probability of violation. We first discuss how our model in Section~\ref{sec:model} can indeed capture this setting and then present simulation results that study when capturing the evolution is fruitful or not. Throughout this section, we use $p\violating(j,\tau)$ to denote the platform's (predicted) probability of violation for content piece $j$ in the $\tau-$th period since it arrives to capture the evolution.

A simple way to capture such evolution is to use the model instantiation in Appendix~\ref{app:model-general}, which does not separate out the probability of violation in the objective. To capture evolution in the probability of violation, we only need to include into the state of a job $S_j(t)$ the current probability of violation $p\violating(j,t)$. However, as discussed in Section~\ref{sec:instantiate}, the model in Appendix~\ref{app:model-general} requires costly data.

\noindent\textbf{Simulation set-up.} The above tradeoff between an accurate model and the difficulties to collect good data for it motivates the following simulation. The simulation tries to understand when we need to work with the accurate model in Appendix~\ref{app:model-general} to incorporate the platform's changing belief on the probability of violation.

The simulation set-up follows Section~\ref{sec:video-real} but incorporates time-varying probability of violation. In Section~\ref{sec:video-real}, we generated the probability of violation  uniformly in $[0,1].$ Here, we generate a series, $p\violating(j,0),p\violating(j,1),\ldots,p\violating(j,L)$ ($L$ is the maximum lifetime of a content piece), corresponding to the platform's belief on the probability of violation of content piece $j$ in different periods from its arrival. The generation process mimics how a platform may update its belief based on new information (such as user reports). Specifically, for a content piece $j$,
\begin{itemize}
\item its actual violation $\violating(j)$ is a Bernoulli random variable with mean $p(j) \sim \mathrm{Uniform}[0,1]$; the quantity $\violating(j)$ is only observable upon human review. The probability $p(j)$ corresponds to the correct prior the platform should have about this content piece. 
\item The platform starts with an initial prior $p\violating(j,0) \in (0,1)$ about the violation $\violating(j)$ when the content just arrives. This prior may or may not be equal to the correct prior $p(j)$.
\item For the $\tau-$th period after the content piece's arrival, with $\tau \geq 0$, the platform collects new information from its users about this content piece. In particular, if the content piece is policy-violating ($\violating(j) = 1$), the platform observes a signal $Y(j,\tau)\sim \mathrm{Bernoulli}(0.1)$, i.e., the probability that some user reports this content piece is $0.1$. On the other hand, if the content piece is not policy-violating ($\violating(j) = 0$), the platform may still observe a signal $Y(j,\tau) \sim \mathrm{Bernoulli}(0.01)$; that is, there can be false positives. 
\item Based on the observed signal $Y(j,\tau)$, the platform applies Bayes' rule to obtain a posterior probability of violation by:
\[
p\violating(j,\tau+1) = \left\{
\begin{aligned}
\frac{0.1p\violating(j,\tau)}{0.1p\violating(j,\tau)+0.01(1-p\violating(j,\tau))},&~~\text{if }Y(j,\tau) = 1\\
\frac{0.9p\violating(j,\tau)}{0.9p\violating(j,\tau)+0.99(1-p\violating(j,\tau))},&~~\text{if }Y(j,\tau) = 0.
\end{aligned}
\right.
\]
\end{itemize}
It is easy to see that the probability of violation $p\violating(j,\tau)$ converges to the actual violation $\violating(j)$ as $\tau \to \infty$. However, the convergence speed depends on the initial prior $p\violating(j,0).$ The simulation considers two settings for the prior: \textsc{CorrectPrior}, where the initial prior is set to the correct prior ($p\violating(j,0) = p(j)$); \textsc{WrongPrior}, where the initial prior is drawn uniformly at random from $[10^{-3},10^{-2}]$.

\noindent\textbf{Tested algorithms:} The simulation tests how algorithms perform under the $\textsc{CorrectPrior}$ and $\textsc{WrongPrior}$ settings. It considers four algorithms (we choose not to include $\textsc{pViolating}$ as its performance is clearly suboptimal from Figure~\ref{fig:youtube-calibrated}):
\begin{itemize}
\item \textsc{Velocity}: recall that $d(j,t)$ is the number of periods content piece $j$ has arrived until period $t$. The index for content piece $j$ in period $t$ is $p\violating(j,d(j,t)) \times \view(j,t-1)$. This is the same velocity scheduling algorithm in Section~\ref{sec:set-up}, but it plugs in the current probability of violation (while the original algorithm just uses the same probability for all periods.)
\item \textsc{pIV}: again, we change the original algorithm by using the current probability of violation. The index is $p\violating(j,d(j,t)) \times $predicted remaining views of content piece $j$ from period $t$.
\item $\hindalg$: this is the original algorithm in in Section~\ref{sec:prac-imp}, which trains a machine learning model to predict the expectation of $\min(\gamma,\text{future views})$. The index is given by 
\[
p\violating(j,d(j,t)) \times \left(\view(j,t-1) + \expect{\min(\gamma,\text{future views})\mid S_j(t)} \right),
\]
where the expectation corresponds to the output of a trained machine learning model with the current state / features $S_j(t)$ as inputs. Note that this algorithm still uses the current violation probability, although it does not model its future evolution in its future-value index.
\item $\hindalg_{\vioviews}$: different from the original $\hindalg$, which only models evolution in the view trajectory, this algorithm incorporates the evolution in the probability of violation. In particular, it uses the cost definition in \eqref{eq:cost-both} and trains a machine learning model that predicts $\min(\gamma,\violating(j) \times \text{future views})$. For a content piece $j$ in period $t$, it uses an index
\[
p\violating(j,d(j,t)) \times \view(j,t-1) + \expect{\min(\gamma,\violating(j)\times \text{future views}) \mid S_j(t)}.
\]
Note that the expectation term in the index is \emph{not} multiplied by the current probability of violation as  the former already captures the uncertainty in the violation variable. Moreover, the expectation term captures the future evolution of the probability of violation by including $p\violating(j,d(j,t))$ into the state representation $S_j(t)$. Lastly, different from $\textsc{pIV}$ and $\hindalg$, for which we do not include the probability of violation as a feature for machine learning because it was unrelated to the prediction target, we include it as a feature here because it is informative when the prediction target is the number of policy-violating views.
\end{itemize}

\noindent\textbf{Simulation results:} We simulate the above four algorithms in the \textsc{CorrectPrior} and the \textsc{WrongPrior} settings. Figure~\ref{fig:Prior}(a) shows the policy-violating views of these algorithms in the \textsc{CorrectPrior} setting under different review ratios. The plot shows that when the platform has correct prior ($p\violating(j,0) = p(j)$ for any content piece $j$), $\hindalg$ is performing sufficiently well, which is consistent with the results in Figure~\ref{fig:youtube-calibrated} when the platform simply uses the prior throughout a content piece's lifetime. Moreover, explicitly incorporating the evolution in the probability of violation is unnecessary, as $\hindalg_{\vioviews}$ is not performing better than $\hindalg$. Indeed, $\hindalg_{\vioviews}$ is slightly less efficient because it needs to solve a more difficult prediction problem of predicting $\violating(j) \times \text{future views}$. 
\begin{figure}[H]
\centering
\includegraphics[width=0.7\linewidth]{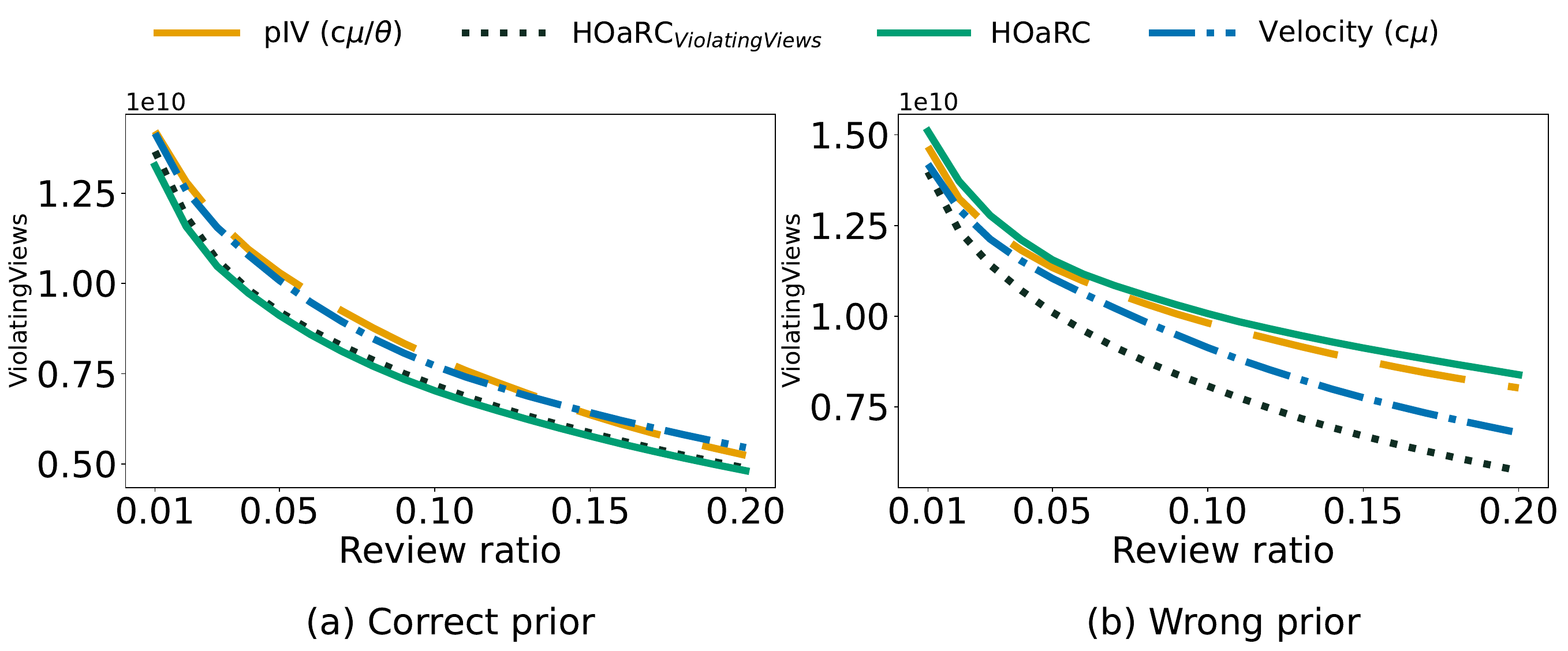}
\caption{Policy-violating views when the platform has a correct prior ((a)) or a wrong prior ((b)) of violation of each content and gradually improves it.}
\label{fig:Prior}
\end{figure}

When the platform does not have a correct prior, incorporating the evolution of $p\violating$ into the model becomes useful. Figure~\ref{fig:Prior}(b) shows the policy-violating views of the considered algorithms in the $\textsc{WrongPrior}$ setting. In this case, $\hindalg$ has degraded performance because the probability of violation it uses in a period is inaccurate. However, $\hindalg_{\vioviews}$ retains its strong performance and substantially outperforms other algorithms.

\subsection{Incorporating bandit learning on the probability of violation}\label{app:bandit-learning}
This subsection studies how our method, which aims to tackle uncertainty in the view trajectory, may integrate with recent bandit-learning-based approaches \citep{avadhanula2022,lykouris2024learning} for the uncertainty in the probability of violation. In what follows, we first review the basic idea of \cite{avadhanula2022} (\cite{lykouris2024learning} is building on top of it and provides further theoretical guarantee). We then discuss the simulation set-up and lastly the simulation results. 

\noindent\textbf{Review of the bandit-learning-based approach.} Our simulations so far have assumed that the probability of violation, $p\violating(j)$, is given and correct, or may evolve according to a simple process as in Appendix~\ref{app:pvio-update}. In practice, platforms use sophisticated machine learning classifiers to predict this probability \citep{halevy2022preserving} using features such as the contained texts, photos, and videos in a piece of content. Since these are large models, they are hard to train and only trained infrequently, causing their effectiveness to become weaker and weaker due to non-stationary content violation trends. To address this issue, \cite{avadhanula2022} propose using labels collected online to fine-tune an easy-to-train small aggregator of several hard-to-train classifiers. The algorithm in \cite{avadhanula2022} works as follows:
\begin{itemize}
    \item when a piece of content $j$ arrives, $m$ trained classifiers corresponding to different platform policies will output a list of probabilities that this content piece violates these policies, denoted by $x_{j,i}$ for classifier $i \in \{0,\ldots,m-1\}$. \cite{avadhanula2022} divide the $[0,1]$ interval into $b$ bins $\{\set{B}_k\}_{k \leq b}$ and create a feature vector $\btheta_j$ of dimension $m \times b$ such that $\btheta_{j, i\times b+k} = x_{j,i}\indic{x_{j,i} \in \set{B}_k}$ for any $i \in \{0,\ldots,m-1\}, k \in \{1,\ldots,b\}$.
    \item Let $\set{D}(t)$ be the dataset of reviewed content pieces before period $t$, i.e., $\set{D}(t) = \cup_{\tau < t} \set{R}_{\tau}$ where $\set{R}_{\tau}$ is the set of reviewed content in period $\tau$. 
    \item \cite{avadhanula2022} apply an upper-confidence-bound (UCB) estimate for the probability of violation $p\violating(j,t)$ of content piece $j$ in period $t$. Specifically, focusing on each dimension $d$ of the feature vector, \cite{avadhanula2022} estimate a one-dimensional linear model with an unknown parameter $\beta_d$ such that $\violating(j') \approx \beta_d \theta_{j',d}$ for $j' \in \set{D}(t).$ It then creates a UCB estimate of $\beta_d$, denoted by $\bar{\beta}_d$, and estimates the probability of violation by setting $p\violating(j,t) = \max_{d \leq mb} \theta_{j,d}\bar{\beta}_d$. The motivation of using a UCB estimate is to overestimate the probability of violation and thus to encourage reviews for content whose probability of violation has high uncertainty.
\end{itemize}
\noindent\textbf{Simulation set-up.} To simulate a setting where learning $p\violating$ is challenging, we incorporate datasets with real-world text content to the setting in Section~\ref{sec:video-real}. Following \cite{lykouris2024learning}, we use a dataset containing texts from Wikipedia (the test set in \cite{dataset-wiki}) and a dataset containing texts from a platform called Civil Comments (the test set in \cite{dataset-civil}). Each data-point of these datasets includes the original text and ground truth human labels on whether this text violates each of six policies. To obtain their feature vectors as discussed above, we query a pretrained machine learning model \citep{detoxify} for each text. The model outputs a list of six $[0,1]-$values representing the probabilities that this text violates
each of the six policies. We then create a feature vector for each text based on the above procedure by mapping each value into one of five bins.

The next step is to match these text data with the YouTube dataset in Section~\ref{sec:video-real} (which only contains view trajectories). For each data-point in the training set of Section~\ref{sec:video-real}, we randomly assign it a text data-point in the Wikipedia dataset. That is, we view this content piece as having the text from the assigned text data-point and the view trajectory from the original trajectory data-point. Similarly, for each data-point in the test set of Section~\ref{sec:video-real}, we randomly assign it a data-point in the Civil Comments dataset. This assignment process creates a distribution shift in the text population between the training set and the test set, under which bandit learning with online violation data may be useful.

The simulation considers two methods to obtain the probability of violation. Denote the training set by $\set{D}_{\text{train}}$ and the test set by $\set{D}_{\text{test}}$. The first method, \textsc{Offline}, trains an offline linear model $\{\hat{\beta}_d\}_{d \leq mb}$ based on the texts in the training set so that $\violating(j') \approx \hat{\beta}_d \theta_{j',d}$ for $j' \in \set{D}_{\text{train}}$. For each content piece $j$ in the test set, we set $p\violating(j,t) = \max_{d \leq mb} \theta_{j,d}\hat{\beta}_d$ for any period $t$. That is, the prediction is fixed to the output of the offline model. The second method, $\textsc{Bandit}$, follows the bandit-learning approach of \cite{avadhanula2022} discussed above, which gradually adapts to the test set based on online collected data.

\noindent \textbf{Simulation results.} The simulation tests six algorithmic combinations between $\{\textsc{Offline},\textsc{Bandit}\}$ for $p\violating$ estimation and $\{\textsc{Velocity}, \textsc{pIV}, \hindalg\}$ for scheduling principles. The below table reports their policy-violating views when the review ratio is $1\%, 5\%, 10\%$, or $15\%$. This table shows that (1) using bandit learning generally improves the performance (except for \textsc{Velocity}); (2) our algorithm $\hindalg$ in general obtains the best performance (smallest number of policy-violating views) when the learning approach is either $\textsc{Offline}$ or $\textsc{Bandit}$ across review ratios, highlighting that our new scheduling principle is compatible with existing learning-based approaches.
\setlength{\tabcolsep}{4pt}
\begin{table*}[ht]
\centering
\begin{tabular}{l|rr|rr|rr|rr}
\hline
& \multicolumn{2}{c|}{1\%} & \multicolumn{2}{c|}{5\%} & \multicolumn{2}{c|}{10\%} & \multicolumn{2}{c}{15\%} \\
\cline{2-9}
& $\textsc{Offline}$ & $\textsc{Bandit}$ & $\textsc{Offline}$ & $\textsc{Bandit}$ & $\textsc{Offline}$ & $\textsc{Bandit}$ & $\textsc{Offline}$ & $\textsc{Bandit}$ \\
\hline
$\textsc{Velocity}$ & 16.3 & 16.3 & 10.6 & 10.7 & 7.63 & 7.24 & 5.49 & 5.56 \\
$\textsc{pIV}$      & 17.3 & 17.7 & 10.7 & 10.2 & 6.09 & 5.94 & 4.15 & 4.11 \\
$\hindalg$          & 16.4 & 16.3 & 10.3 & 9.60 & 6.11 & 5.91 & 4.09 & 3.93 \\
\hline
\end{tabular}
\caption{Policy-violating views across review ratios, with all values in units of $10^{8}$.}
\label{tab:results_all_ratios}
\end{table*}

\section{Useful Facts}\label{app:fact}
% !TEX root = main.tex
This section states several useful analytical facts and concentration bounds. 

The below fact shows condition for $x$ such that $a\ln x \leq x / 2$.
\begin{fact}\label{fact:ln-order}
If $a \geq 1$ and $x \geq \exp(4a\ln(e\cdot a))$, then $a\ln x \leq x / 2.$
\end{fact}
\begin{proof}
Let $g(x) = x/2 - a\ln x$. Its derivative is $g'(x) = 1/2 - a/x$ which is positive when $x \geq 2a$. Moreover, by Taylor's expansion,
\[
g(e^{4a\ln (e\cdot a)}) = \frac{1}{2}e^{4a\ln (e\cdot a)}- 4a^2\ln (e \cdot a) \geq \frac{1}{2} +2a\ln (e\cdot a) + 4a^2\ln^2(e\cdot a) - 4a^2 \ln (e \cdot a) \geq 0,
\]
where the last inequality is because $4a^2 \ln^2(e\cdot a)  \geq 4a^2 \ln(e\cdot a).$ Using $e^{4a\ln(e\cdot a)} \geq 4a\ln(e \cdot a) \geq 2a$ shows that the function $g(x)$ is increasing and non-negative over $x \geq \exp(4a\ln(e\cdot a)).$
\end{proof}

We use Hoeffding's Inequality, whose proof can be found in, e.g., Theorem 2.8 of \cite{boucheron2013concentration}.
\begin{fact}[Hoeffding's Inequality]\label{fact:hoeffding}
Given $n$ independent random variables $X_1,\ldots,X_n$ such that $X_i \in [a_i,b_i]$ almost surely for all $i$. Letting $S = \sum_{i=1}^n (X_i 
- \expect{X_i})$, for any $x > 0$,
\[
\max\left(\Pr\{S \leq -x\},\Pr\{S \geq x\}\right) \leq \exp\left(-\frac{2x^2}{\sum_{i=1}^n (b_i-a_i)^2}\right).
\]
\end{fact}
We also use the following Chernoff bound restated from theorem 2.4 of \cite{chung2006complex}.
\begin{fact}[Chernoff bound]\label{fact:chernoff} Given $n$ independent random variables $X_1,\ldots,X_n$ such that $X_i \in \{0,1\}$ for all $i$. Letting $X = \sum_{i = 1}^n X_i,$ for any $\varepsilon > 0$,
\[
\Pr\{X \geq \expect{X} + \varepsilon\} \leq \exp\left(-\frac{\varepsilon^2}{2(\expect{X} + \varepsilon / 3)}\right).
\]
\end{fact}
\end{document}